\documentclass[12pt]{article}
\usepackage{amsmath}
\usepackage{enumerate}
\usepackage[sort&compress,round]{natbib}
\usepackage{graphicx} % Required for inserting images

\usepackage[left=1in,top=0.8in,right=1in,bottom=0.7in,head=.2in,nofoot]{geometry}

\usepackage[utf8]{inputenc}
\usepackage[bookmarksopen=true,bookmarks=true,pdfencoding=auto,psdextra,colorlinks,
]{hyperref}
\hypersetup{ 
    colorlinks,
    linkcolor={blue!80!black},
    citecolor={green!40!black},
    urlcolor={red!50!black}
}
\usepackage{bookmark}
\usepackage{mystyle}
\usepackage{natbib}

\allowdisplaybreaks
\usepackage{xspace}

\usepackage{xcolor}
\usepackage{colortbl}

\usepackage[font=small,labelfont=bf]{caption}
\usepackage{array}
\usepackage{ragged2e}
\usepackage{tabularx}

\newcolumntype{Y}{>{\centering\arraybackslash}X}
\newcolumntype{P}{>{\raggedleft\arraybackslash}X}
\usepackage{multirow}
\newcolumntype{C}[1]{>{\centering\arraybackslash}m{#1}}
\newcolumntype{L}{>{\raggedright\arraybackslash}X}

\usepackage{algorithm}
\usepackage{algorithmic}

\def\spacingset#1{\renewcommand{\baselinestretch}%
{#1}\small\normalsize} \spacingset{1}
\newcommand{\IQR}{\textup{IQR}}
\newcommand{\Median}{\textup{Med}}
\newcommand{\FDP}{\textup{FDP}\xspace}
\newcommand{\FDPex}{\textup{FDX}\xspace}
\newcommand{\FDR}{\textup{FDR}\xspace}
\newcommand{\FWER}{\textup{FWER}\xspace}

\newcommand{\cn}{c_n}

\usepackage{tikz}
\newcommand\sbullet[1][.5]{\mathbin{\vcenter{\hbox{\scalebox{#1}{$\bullet$}}}}}

\newcommand{\footremember}[2]{%
    \footnote{#2}
    \newcounter{#1}
    \setcounter{#1}{\value{footnote}}%
}
\newcommand{\footrecall}[1]{%
    \footnotemark[\value{#1}]%
}
\title{Causal Inference for Genomic Data\\with Multiple Heterogeneous Outcomes}
\author{Jin-Hong Du\footremember{cmustats}{Department of Statistics and Data Science, Carnegie Mellon University, Pittsburgh, PA 15213, USA.}\footremember{cmumld}{Machine Learning Department, Carnegie Mellon University, Pittsburgh, PA 15213, USA.} \and 
Zhenghao Zeng\footrecall{cmustats} \and
Edward H. Kennedy\footrecall{cmustats} \and
Larry Wasserman\footrecall{cmustats} \footrecall{cmumld} \and 
Kathryn Roeder\footrecall{cmustats} \footremember{cmucbd}{Computational Biology Department, Carnegie Mellon University, Pittsburgh, PA 15213, USA.}}
\date{\today}

\begin{document}

\maketitle

\begin{abstract}
    With the evolution of single-cell RNA sequencing techniques into a standard approach in genomics, it has become possible to conduct cohort-level causal inferences based on single-cell-level measurements. However, the individual gene expression levels of interest are not directly observable; instead, only repeated proxy measurements from each individual's cells are available, providing a derived outcome to estimate the underlying outcome for each of many genes. In this paper, we propose a generic semiparametric inference framework for doubly robust estimation with multiple derived outcomes, which also encompasses the usual setting of multiple outcomes when the response of each unit is available. To reliably quantify the causal effects of heterogeneous outcomes, we specialize the analysis to standardized average treatment effects and quantile treatment effects. Through this, we demonstrate the use of the semiparametric inferential results for doubly robust estimators derived from both Von Mises expansions and estimating equations. A multiple testing procedure based on Gaussian multiplier bootstrap is tailored for doubly robust estimators to control the false discovery exceedance rate. Applications in single-cell CRISPR perturbation analysis and individual-level differential expression analysis demonstrate the utility of the proposed methods and offer insights into the usage of different estimands for causal inference in genomics.
\end{abstract}

\noindent%
\vfill

\clearpage

\section{Introduction}

    In observational studies, causal inference on multiple outcomes is increasingly prevalent in scientific discoveries \citep{imbens2015causal}. 
    Recent advances in high-throughput techniques have enabled the collection of large-scale repeated measurements across multiple subjects in various domains. However, subject-level outcomes, such as averages or inter-correlations of measurements within each subject, are often unobserved. Instead, researchers rely on repeated measurements to construct estimates of these outcomes, referred to as \emph{derived outcomes} (\Cref{fig:causal-diagram}). For example, advancements in single-cell RNA sequencing (scRNA-seq) techniques \citep{review:2023} have enabled large-scale repeated gene expression measurements across multiple cells for each individual. These measurements allow researchers to construct derived outcomes (e.g., the sample average of gene expressions for an individual) as proxies for subject-level outcomes (e.g., the true average gene expression for that individual), facilitating individual-level comparisons \citep{zhang2022ideas}. The goal of subject-level causal inference is to determine which unobserved outcomes are causally affected by treatment by comparing derived outcomes between treatment and control groups. However, challenges such as unobservability of outcomes, subject heterogeneity \citep{qiu2023unveiling}, and non-identical outcome distributions limit the reliability of existing causal inference methods.

    One major challenge of subject-level causal inference on scRNA-seq data is the unobservability of the outcomes. Individual gene expression levels of subjects are not directly measurable; instead, repeated measurements from heterogeneous cells are aggregated to serve as proxies.
    Additionally, variability among subjects may arise from latent states unique to each individual that influence gene expression patterns but remain hidden from direct observation \citep{gcate2023}.
    Consequently, derived outcomes often violate the classic assumption of being independent and identically distributed due to biological processes, experimental conditions, and inherent cellular heterogeneity.
    To analyze subject-level brain functional connectivities, prior work by \citet{qiu2023unveiling} attempts to estimate average treatment effects (ATEs) using inverse probability weighting (IPW) estimators \citep{imbens2004nonparametric, tsiatis2006semiparametric}. However, their approach relies on accurate propensity score modeling and assumes outcome homogeneity, which may not hold for genomic data.

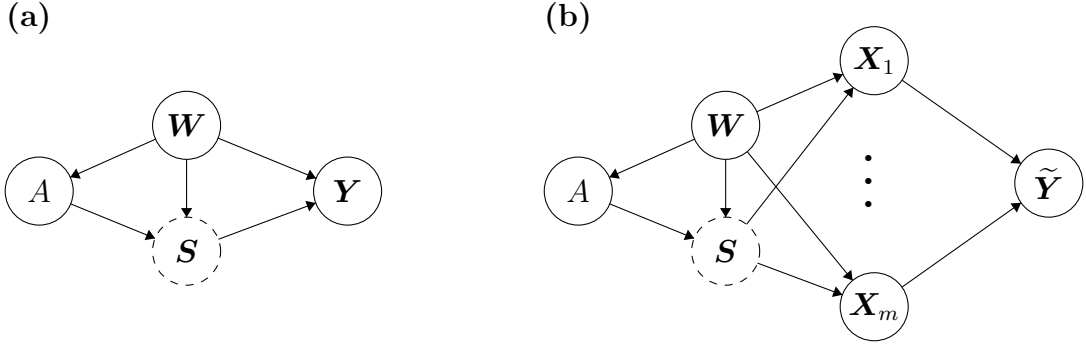
\begin{figure}[!t]
    \centering
    \begin{tikzpicture}[baseline,scale=0.15]
    \tikzstyle{every node}+=[inner sep=0pt]
    \node[font=\bf] at (13,-3) {(a)};
    \draw [black] (41.1,-18.3) circle (3);
    \draw (41.1,-18.3) node {$\bY$};
    \draw [black,dashed] (26.9,-23.6) circle (3);
    \draw (26.9,-23.6) node {$\bS$};
    \draw [black] (13.9,-18.3) circle (3);
    \draw (13.9,-18.3) node {$A$};
    \draw [black] (26.9,-12.5) circle (3);
    \draw (26.9,-12.5) node {$\bW$};
    \draw [black] (16.68,-19.43) -- (24.12,-22.47);
    \fill [black] (24.12,-22.47) -- (23.57,-21.7) -- (23.19,-22.63);
    \draw [black] (24.16,-13.72) -- (16.64,-17.08);
    \fill [black] (16.64,-17.08) -- (17.57,-17.21) -- (17.17,-16.3);
    \draw [black] (26.9,-15.5) -- (26.9,-20.6);
    \fill [black] (26.9,-20.6) -- (27.4,-19.8) -- (26.4,-19.8);
    \draw [black] (29.68,-13.63) -- (38.32,-17.17);
    \fill [black] (38.32,-17.17) -- (37.77,-16.4) -- (37.39,-17.33);
    \draw [black] (29.71,-22.55) -- (38.29,-19.35);
    \fill [black] (38.29,-19.35) -- (37.37,-19.16) -- (37.71,-20.1);
    \end{tikzpicture}\hspace{2cm}
    \begin{tikzpicture}[baseline,scale=0.15]
    \tikzstyle{every node}+=[inner sep=0pt]
    \node[font=\bf] at (13,-3) {(b)};
    \draw [black] (40,-6.8) circle (3);
    \draw (40,-6.8) node {$\bX_1$};
    \draw [black] (55.4,-17.6) circle (3);
    \draw (55.4,-17.6) node {$\tilde{\bY}$};
    \draw [black,dashed] (26.9,-23.6) circle (3);
    \draw (26.9,-23.6) node {$\bS$};
    \draw [black] (40,-28.5) circle (3);
    \draw (40,-28.5) node {$\bX_m$};
    \draw [black] (13.9,-18.3) circle (3);
    \draw (13.9,-18.3) node {$A$};
    \draw [black] (26.9,-12.5) circle (3);
    \draw (26.9,-12.5) node {$\bW$};
    \draw [black] (42.46,-8.52) -- (52.94,-15.88);
    \fill [black] (52.94,-15.88) -- (52.58,-15.01) -- (52,-15.83);
    \draw [black] (42.45,-26.77) -- (52.95,-19.33);
    \fill [black] (52.95,-19.33) -- (52.01,-19.39) -- (52.59,-20.2);
    \draw [black] (16.68,-19.43) -- (24.12,-22.47);
    \fill [black] (24.12,-22.47) -- (23.57,-21.7) -- (23.19,-22.63);
    \draw [black] (28.74,-21.23) -- (38.16,-9.17);
    \fill [black] (38.16,-9.17) -- (37.27,-9.49) -- (38.06,-10.1);
    \draw [black] (29.71,-24.65) -- (37.19,-27.45);
    \fill [black] (37.19,-27.45) -- (36.62,-26.7) -- (36.27,-27.64);
    \draw (39.5,-15.65) node {$\sbullet[0.5]$};
    \draw (39.5,-17.65) node {$\sbullet[0.5]$};
    \draw (39.5,-19.65) node {$\sbullet[0.5]$};
    \draw [black] (24.16,-13.72) -- (16.64,-17.08);
    \fill [black] (16.64,-17.08) -- (17.57,-17.21) -- (17.17,-16.3);
    \draw [black] (26.9,-15.5) -- (26.9,-20.6);
    \fill [black] (26.9,-20.6) -- (27.4,-19.8) -- (26.4,-19.8);
    \draw [black] (29.65,-11.3) -- (37.25,-8);
    \fill [black] (37.25,-8) -- (36.32,-7.86) -- (36.72,-8.77);
    \draw [black] (28.8,-14.82) -- (38.1,-26.18);
    \fill [black] (38.1,-26.18) -- (37.98,-25.24) -- (37.21,-25.88);
    \end{tikzpicture}
    \caption{The causal diagram for the causal inference problems studied in this paper.
    \textbf{(a) Multiple outcomes.} For a cell, its gene expression $\bY\in\RR^p$ is causally affected by the treatment $A\in\RR$, the latent state $\bS\in\RR^\ell$ and covariate $\bW\in\RR^q$ such as batch effects.   
    \textbf{(b) Multiple derived outcomes.}
    In the subject-level studies, a subject's overall gene expression $\bY$ is not directly observed.
    Instead, repeated measurements of gene expressions $\bX_1,\ldots,\bX_m\in\RR^{d}$ of $m$ cells from the subject provides a proxy $\tilde{\bY}$ for $\bY$.
    See \Cref{sec:setup} for formal definitions.
    Note that the treatment effect of $A$ on $\bY$ (or $\tilde{\bY}$) is mediated by the latent state $\bS$ even conditioned on the covariate $\bW$.
    When conditioned on $\bW$ and $A$, the outcomes $Y_1,\ldots,Y_p$ within the same subject are still not independent and identically distributed.
    }\label{fig:causal-diagram}
\end{figure}

    A second challenge arises from the heterogeneity of gene expression data, which often exhibit variable scaling and right-skewed distributions, complicating comparisons across outcomes.
    This heterogeneity challenges the common use of ATEs since relying solely on mean differences in counterfactual distributions can lead to misleading conclusions. 
    To estimate the treatment effects, one can rely on propensity score modeling or outcome modeling.
    For instance, one common strategy in scRNA-seq analyses is to model outcomes directly using parameter models such as Poisson or Negative Binomial models \citep{sarkar2021separating}, or zero-inflated models \citep{jiang2022statistics}.
    While using either IPW or regression estimators may seem intuitive, they are both sensitive to model specification.

    Doubly robust (DR) estimators \citep{robins1994estimation, scharfstein1999adjusting} offer a promising solution to mitigate sensitivity to model specification by combining IPW and outcome modeling. 
    DR estimators are consistent as long as either of the two nuisance estimators is consistent, and $\sqrt{n}$-consistent whenever the nuisance estimators converge at only $n^{-1/4}$ rates (or more generally if the product of their errors is of the order $n^{-1/2}$).
    Additionally, if both the nuisance models are correctly specified, in the sense that the product of their errors is smaller order than $n^{-1/2}$, the DR estimators achieve the semiparametric efficiency bound for the unrestricted model, allowing the regression and propensity score functions to be estimated flexibly at slower than $n^{-1/2}$ rates in a wide variety of settings \citep{laan2003unified,tsiatis2006semiparametric}.
    Recent work introduces a structure-agnostic framework for functional estimation, demonstrating that DR estimators are optimal for estimating ATEs when only black-box estimators of the outcome model and propensity score are available \citep{balakrishnan2023fundamental, jin2024structure}. These results suggest that DR estimators cannot be improved upon without making additional structural assumptions. Given these advantages, exploring doubly robust estimation in settings with multiple heterogeneous outcomes is crucial, as it mitigates model misspecification and enables reliable statistical testing when nonparametric methods are used for outcome and propensity score estimation.

    The unobservability of subject-level outcomes, heterogeneity in gene expression distributions, and sensitivity to model specification collectively challenge traditional causal inference methods in scRNA-seq studies. 
    To address these issues, we propose a semiparametric inference framework to handle multiple derived outcomes effectively. 
    Specifically, (i) we define causal estimands that capture meaningful counterfactual differences across multiple outcomes and establish identification conditions under hierarchical models where outcomes of interest are unobserved (\Cref{fig:causal-diagram}(b)); and (ii) we develop robust and efficient estimators tailored for these estimands. Additionally, we extend multiple-testing procedures to control statistical errors during simultaneous inference on high-dimensional derived outcomes. 
    Together, these methodological advancements provide a unified approach incorporating doubly robust estimation to handle multiple derived outcomes effectively.
    
    Focusing on multiple derived outcomes, we first establish generic results on semiparametric inference with doubly robust estimators.
    It also encompasses the usual setting of multiple outcomes when the response of each unit is available.
    By utilizing finite-sample maximal inequality for finite maximums, we obtain interpretable conditions of uniform estimation error control for the empirical process terms and the asymptotic variances.
    We derive the uniform convergence rates, in terms of only finitely many moments of the influence functions' envelope and the maximum second moments of the individual estimation errors, allowing for the number of outcomes $p$ to be potentially exponentially larger than the sample size $n$.

    To address the challenges of outcome heterogeneity in single-cell data, we further specialize our analysis to standardized average treatment effects for comparing treatment effects across different outcomes on a common scale and quantile treatment effects for robustness against outliers. This demonstrates how generic semiparametric inferential results for DR estimators derived from von Mises expansions and estimating equations can be applied in high-dimensional settings. Furthermore, we adapt Gaussian approximation results from \citet{chernozhukov2013gaussian} to DR estimators and implement a step-down procedure to control false discovery (exceedance) rates with guaranteed power \citep{genovese2006exceedance}.

    Our exploration includes two real-world application scenarios of the proposed causal inference methods.
    (1) \emph{Single-cell CRISPR perturbation analysis}: 
    Gene expressions from single cells are compared between perturbation and control groups in CRISPR experiments to identify target genes of individual perturbations and analyze the effects of perturbations \citep{dixit2016perturb}, as shown in \Cref{fig:causal-diagram}(a).
    (2) \emph{Individual level differential expression analysis}: 
    Aggregated gene expressions from individual subjects under two conditions (case and control) are analyzed to identify genes intrinsically affected by these conditions across subjects, corresponding to \Cref{fig:causal-diagram}(b).
    By applying our methods to these datasets, we demonstrate their practical utility while highlighting the strengths and limitations of different causal estimands. 
    These findings emphasize the importance of suitable causal inference techniques for the accurate interpretation of genomic data.

\textbf{Organization.}
In \Cref{sec:semi-inf-multi-outcomes}, we review and extend the classic semiparametric inference framework to the setting of multiple outcomes.
In \Cref{sec:setup}, we set up the problem of interest and discuss the identification conditions for the causal estimands.
In \Cref{sec:standarized-effects}, we analyze two DR estimators for standardized and quantile treatment effects and study their statistical properties.
In \Cref{sec:inference}, we study the statistical error of simultaneous inference and propose a multiple testing procedure for controlling the false discovery rate.
In \Cref{sec:simu} and \Cref{sec:real-data}, we conduct simulations and scRNA-seq data analysis to validate the proposed simultaneous causal inference method.
A detailed review of related work is provided in \Cref{app:related-work}.

\textbf{Notation.} Throughout our exposition, we will use the following notational conventions. 
    We denote scalars in non-bold lower or upper case (e.g., $X$), vectors in bold upper case (e.g., $\bX$), and matrices in non-italic bold upper case (e.g., $\Xb$).
    For $a,b\in\RR$, we write $a\vee b = \max\{a,b\}$ and $a\wedge b = \min\{a,b\}$.
    For any random variable $X$, its $L_q$ norm is defined by $\|X \|_{L_q}= (\EE[| X |^q] )^{1/q}$ for $q=1,\ldots,\infty$.
    For $p\in\NN$, $[p] :=\{1,\ldots,p\}$.
    For a set $\cA$, let $|\cA|$ be its cardinality.
    For (potentially random) measurable functions $f$, we denote expectations with respect to $Z$ alone by $\PP f(Z) = \int f \rd\PP$, and with respect to both $Z$ and the observations where $f$ is fitted on by $\EE [f(Z)]$.
    The empirical expectation is denoted by $\PP_n f(Z) = \frac{1}{n}\sum_{i=1}^nf(Z_i)$ for $\iid$ samples $Z_1,\ldots,Z_n$ of $Z$.
    Similarly, the population and empirical variances are denoted by $\VV$ and $\VV_n$, respectively.
    We write the (conditional) $L_p$ norm of $f$ as $\|f\|_{L_p} = \left[\int f(z)^p \rd \PP(z) \right]^{1/p}$ for $p\geq1$.
    We use ``$o$'' and ``$\cO$'' to denote the little-o and big-O notations and let ``$\op$'' and ``$\Op$'' be their probabilistic counterparts.
    For sequences $\{a_n\}$ and $\{b_n\}$, we write $a_n\ll b_n$ or $b_n\gg a_n$ if $a_n=o(b_n)$; $a_n\lesssim b_n$ or $b_n\gtrsim a_n$ if $a_n=\cO(b_n)$; and $a_n\asymp b_n$ if $a_n=\cO(b_n)$ and $b_n=\cO(a_n)$. All the constants $c,c_1,c_2$ and $C, C_1,C_2$ may vary from line to line.
    Convergence in distribution and probability are denoted by ``$\dto$'' and ``$\pto$''.

\section{Semiparametric inference with multiple outcomes}\label{sec:semi-inf-multi-outcomes}
    Prior to delving into our main topic, this section takes a brief excursion into the formulation of semiparametric inference within the context of multiple outcomes, laying the foundation for addressing our specific problem in subsequent sections.

    Let $\bZ_1,\ldots,\bZ_n$ be observations of i.i.d. samples from a population $\cP$.
    In the presence of multiple outcomes, we are interested in estimating $p$ target estimands $\tau_j:\cP\mapsto \RR$ for $j=1,\ldots,p$. 
    Suppose $\tau_j$ admits a von Mises expansion; that is, there exists an influence function $\varphi_j(z;\PP)$ with $\int \varphi_j(z;\PP) \rd \PP(z) = 0$ and $\int \varphi_j(z;\PP)^2 \rd \PP(z) <\infty$, such that
    \begin{align*}
        \tau_j(\overline{\PP}) - \tau_j({\PP}) = \int \varphi_j(z;\overline{\PP}) \rd (\overline{\PP} - {\PP}) + T_{{\rm R},j},
    \end{align*}
    where $T_{{\rm R},j}$ is a second-order remainder term (which means it only depends on products or squares of differences between $\overline{\PP}$ and ${\PP}$).
    The influence function quantifies the effect of an infinitesimal contamination at the point $z$ on the estimate, standardized by the mass of the contamination.
    Its historical development and definition under diverse sets of regularity conditions can be found at \citet[Section 2.1]{hampel2011robust}.    
    The above expansion suggests a one-step estimator that corrects the bias of the plug-in estimator $\tau_j(\hat{\PP})$:
    \begin{align}
        \hat{\tau}_j(\PP) &:= \tau_j(\hat{\PP}) + \PP_n\{\varphi_j(\bZ;\hat{\PP})\}, \label{eq:one-step-est}
    \end{align}
    where $\hat{\PP}$ is an estimate of $\PP$.
    For many estimands, such as ATE and expected conditional covariance, the one-step estimator is also a DR estimator.
    Although for certain estimands like expected density, the standard one-step estimator is not doubly robust, it still has nuisance errors that consist of a second-order term \citep{kennedy2021semiparametric}.
    Then, the one-step estimator $\hat{\tau}_j$ of the $j$th target estimand $\tau_j$ admits a three-term decomposition of the estimation error \citep[Equation (10)]{kennedy2022semiparametric}\footnote{For certain estimands, such as average treatment effects and expected conditional covariance functionals, it is usually more convenient to use the uncentered influence functions $\phi_j(\bZ;\PP) := \varphi_j(\bZ;\PP) +\tau_j(\PP)$ in the decomposition.}:
    \begin{align}
        \hat{\tau}_j(\PP) - \tau_j(\PP) = \underbrace{(\PP_n - \PP) \{\varphi_j(\bZ; \PP)\}}_{T_{{\rm S},j}} + \underbrace{(\PP_n - \PP) \{\varphi_j(\bZ; \hat{\PP}) - \varphi_j(\bZ; \PP)\}}_{T_{{\rm E},j}} + T_{{\rm R},j}.
        \label{eq:decomposition}
    \end{align}
    In the above decomposition, the first term after $\sqrt{n}$-scaling has an asymptotic normal distribution by the central limit theorem. 
    That is, $\sqrt{n}\,T_{{\rm S},j} \dto \cN(0,\Var[\varphi_j(\bZ; \PP)])$.    
    The higher-order term $T_{{\rm R},j}$ is usually negligible and has an order of $\op(n^{-1/2})$ under certain conditions.
    On the other hand, the empirical process term $T_{{\rm E},j}$ will be asymptotically negligible (i.e., $\op(n^{-1/2})$) under Donsker assumption \citep{van2000asymptotic} or sample splitting \citep{chernozhukov2018double, kennedy2020sharp}, because it is a sample average with a shrinking variance. 
    In our problem setting with an increasing number of outcomes, uniform control over all the empirical process terms $T_{{\rm E},j}$ for $j=1,\ldots,p$ is desired to facilitate the construction of simultaneous confidence intervals.
    Below is an extension of Lemma 2 from \citet[Appendix B]{kennedy2020sharp} to the setting of multiple outcomes.

    \begin{lemma}[Uniform control of the empirical process terms]\label{lem:emp-process-term}
        Let $\PP_n$ denote the emprical measure over $\cD_n=(\bZ_1,\ldots,\bZ_n)\in \cZ^n$, and let $g_j:\cZ\to \RR$ be a (possibly random) function independent of $\cD_n$ for $j=1,\ldots, p$ with $p \geq 2$.
        Let $G(\cdot) := \max_{1\leq j\leq p}| g_j(\cdot) |$ denote the envelope of $g_1,\ldots,g_p$.
        If $\max_{1\leq j\leq p}\|g_j\|_{L_2}< \infty$ and $\|G\|_{L_q}<\infty$ for some $q\in\NN\cup\{\infty\}$, then the following statements hold:
        \begin{align*}
            \EE\left[\max_{j=1,\ldots,p}\ |(\PP_n-\PP)g_j | \;\middle|\; \{g_j\}_{j=1}^p \right] &\lesssim \left(\frac{\log p}{n}\right)^{1/2} \max_{1\leq j\leq p} \|g_j\|_{L_2} + \left(\frac{\log p}{n}\right)^{1-1/q} \|G\|_{L_q}.
        \end{align*} 
    \end{lemma}

    The proof of \Cref{lem:emp-process-term} utilizes a finite-sample maximal inequality established in \citet{kuchibhotla2022least} for high-dimensional estimation problems.
    When specified to particular target estimands, \Cref{lem:emp-process-term} suggests that $T_{{\rm E},j}$ in \eqref{eq:decomposition} can be uniformly controlled over $j=1,\ldots,p$, if one can derive the uniform $L_2$-norm bound and the $L_q$-norm bound of the envelope for the estimation error of the influence functions $g_j = \varphi_j(\bZ;\hat{\PP}) - \varphi_j(\bZ;\PP)$.
    In particular, if $g_j$'s are bounded, and $\log p \cdot (\max_{1\leq j\leq p} \|g_j\|_{L_2}^2 \vee n^{-1/2} ) = o(1)$, then we have that $\max_{1\leq j\leq p}T_{{\rm E},j} = \op(n^{-1/2})$, which is negliable after scaled by $\sqrt{n}$.
    This allows the number of outcomes $p$ to be potentially exponentially larger than the number of samples $n$. 
    It is important to note that similar bounds on the empirical process term can still be derived even when the nuisance functions are not trained on an independent sample, provided certain complexity measures of the function class $\cF_j$ that $g_j$ belongs to are properly bounded; see \Cref{rmk:donsker} in \Cref{app:subsec:emp} for more details.

    \begin{remark}[One-step estimator from estimating equations]
        Above, we construct the one-step estimator based on the influence function $\varphi_j$ of $\tau_j$ from von Mises expansion.
        One can also construct efficient estimators of pathwise differentiable functionals through estimating equations, which is related to the quantile estimand as we will discuss in \Cref{subsec:standarized-quantile}.
    \end{remark}
    
    When $T_{{\rm E},j}+T_{{\rm R},j}=\op(n^{-1/2})$, by central limit theorem we have that $\sqrt{n}(\hat{\tau}_j(\PP) - \tau_j(\PP)) \dto \cN(0,\sigma_j^2)$ where $\sigma_j^2 = \Var[\varphi_j(\bZ; \PP)]$. 
    To construct confidence intervals, one can use the sample variance $\hat{\sigma}_j^2 = \VV_n[\varphi_j(\bZ;\hat{\PP})]$ to consistently estimate the asymptotic variance.
    To derive the properties of test statistics and confidence intervals in high dimensions when $p\gg n$, it is necessary to establish strong control on the uniform convergence rate of the variance estimates over $j=1,\ldots,p$; see, for example, \citet[Proposition 2]{qiu2023unveiling} and \citet[Comment 2.2]{chernozhukov2013gaussian}.
    In this regard, the following lemma provides general conditions for bounding the uniform estimation error.

    \begin{lemma}[Uniform control of the variance estimates]\label{lem:var}
    Denote ${\varphi}_j = \varphi_j(\bZ; {\PP})$, $\hat{\varphi}_j = \varphi_j(\bZ; \hat{\PP})$, $\Phi = \max_{1\leq j\leq p}|\hat{\varphi}_j-\varphi_j|$, and $\Psi = \max_{1\leq j\leq p}|\varphi_j|$.
    Suppose the following conditions hold:
    \begin{enumerate}[label=(\arabic*)]
        \item Envelope: $\max_{1\leq j\leq p}|\hat{\varphi}_j+\varphi_j|\lesssim 1$ and $\|\Psi\|_{L_q} + \max_{k=1,2}\|\Phi^k\|_{L_q} \lesssim r_{1n}$ for some $q>1$,
        
        \item Estimation error:  $\max_{1\leq j\leq p}\|\hat{\varphi}_j-\varphi_j\|_{L_2} \lesssim r_{2n}$, $\max_{1\leq j\leq p} |\PP[\hat{\varphi}_j - {\varphi}_j]| \lesssim r_{3n}$,
    \end{enumerate}
     with probability tending to one.
    Then, it holds that 
        \[\max_{1\leq j\leq p} |\hat{\sigma}_j^2 - \sigma_j^2| \lesssim \Op\left( \left(\frac{\log p}{n}\right)^{1 - 1/q} r_{1n} + \left(\frac{\log p}{n}\right)^{1/2} r_{2n} + r_{3n} \right).\]
    \end{lemma}
    
    Note that $r_{1n}$ and $r_{2n}$ are allowed to potentially diverge.
    We will utilize \Cref{lem:var} for the purpose of multiple testing, as demonstrated in \Cref{sec:inference}. 
    Typically, to accommodate an exponentially large number of outcomes $p$ relative to $n$ while maintaining valid statistical inference, it suffices to ensure that variance estimates are consistent at any polynomial rate of the sample size $n$;  specifically, $\max_{1\leq j\leq p} |\hat{\sigma}_j^2 - \sigma_j^2| = \op(n^{-\alpha})$ for some constant $\alpha>0$.

\section{Subject-level causal inference with multiple outcomes}\label{sec:setup}

    We consider an increasingly popular study design where scRNA-seq data are collected from multiple individuals and the question of interest is to find genes that are causally differentially expressed between two groups of individuals, based on repeated single-cell measurements.

    \subsection{Causal inference with multiple derived outcomes}\label{subsec:general}
    
    Suppose a subject can be either in the case or control group, indicated by a binary random variable $A\in\{0,1\}$ and we sequence the expressions $\Xb=[\bX_{1},\ldots,\bX_{m}]^{\top}\in\RR^{m\times d}$ of $d$ genes in $m$ cells\footnote{For notational simplicity, we treat the number of cell $m$ as fixed across subjects, though the method also applies when the number of cell $m_i$ varies for subjects $i=1,\ldots,n$.}, along with subject-level covariates $\bW\in\RR^{q}$.
    Let $\Xb(a)$ denote the \emph{potential response} of gene expressions.
    To characterize the complex biological processes, suppose $\bS(a) \in\RR^\ell$ is the latent potential state after receiving treatment $A=a$, which fully captures the effects of the treatment on the individual.
    We assume $\bX_1(a),\ldots,\bX_m(a)$ are conditionally independent and identically distributed given $\bS(a)$ and $\bW$\footnote{Technically, the potential response can be denoted as $\bX_m(\bS(a))$; however, because $\bS(a)$ is unobservable and the variable to intervene is $A$, we use a simplified notation $\bX_m(a)$ to denote the potential response.}.
    Marginally, however, they can be highly dependent because of repeated measurements from the same individual.
    As an example of genomics data, $\bS$ can be the chromatin accessibility that governs the translation and expression of genes, while $\bX_m$ is the resulting expression level of those genes.

    Suppose the collection of treatment assignment, covariates, subject level parameters, and potential responses $(A,\bW,\bS(0),\bS(1),\Xb(0),\Xb(1)) $ is from some super-population $\cP$.
    We require the consistency assumption on the observed response.
    \begin{assumption}[Consistency]\label{asm:consistency}
        The observed response is given by $\Xb=A \Xb(1) + (1-A)\Xb(0)$.
    \end{assumption}

    When comparing gene expressions between two groups of individuals, the $p$-dimensional subject-level parameter of interest, is the \emph{potential outcome} $\bY(a)\in\RR^p$, a functional that maps the conditional distribution of $\bX_1(a)$ given $\bS(a)$ and $\bW$ to $\RR^p$:    
    \begin{align}
        \bY(a) =  \EE\left[f (\bX_1(a)) \,\middle|\, \bS(a),\bW\right],   \label{eq:Ya-expectation}
    \end{align}
    for some prespecified function $f:\RR^{d}\rightarrow \RR^{p}$.
    The choice of $f$ should align with the user's research goals and the specific aspects of the data they intend to capture.
    For example, when $f$ is the identity map and $p=d$, the potential outcome $\bY(a)$ represents the conditional mean; 
    when $f(\bX) = \bX\bX^{\top} - \EE[\bX\bX^{\top}\mid \bS(a),\bW]$ and $p=d$, it represents the intrasubject covariance matrix.
    When considering conditional means among the potential responses $X_{1j}(a)$'s, it can also be nodewise regression coefficients as considered in \citet{qiu2023unveiling}.
    Intuitively, $\bY(a)$ is an individual / within-group characteristic that depends on the conditional distribution of $\bX_{1}(a)$ given $\bS(a), \bW$.
    From \Cref{asm:consistency}, we also have $\bY = A\bY(1)+(1-A)\bY(0)$.
    Compared to the classical causal inference setting, the subject-level outcome $\bY$ is not observed for each subject, while only the repeated measurements of gene expressions $\Xb$ from multiple cells are available and can be used to construct a derived outcome $\tilde{\bY}$.
    For a given $f$, we consider a statistic $\tilde{\bY}:= g(\Xb)$ for some function $g:(\bX_1,\ldots,\bX_m)\mapsto \tilde{\bY}$.
    There can be different choices of $g$ to estimate $\bY(a)$ by $\tilde{\bY}(a)$.
    For instance, if $f$ is a linear function for the potential outcome $\bY(a)$ in \eqref{eq:Ya-expectation}, then $g$ can be a simple sample average as a natural choice of the derived outcome; alternatively, $g$ can also be the median-of-means estimator as the derived outcomes.

    Under the derived outcomes framework, \citet{qiu2023unveiling} studied the IPW estimator for ATE:
    \begin{align}
        \qquad\qquad\tau_{j}^{\ATE} &= \EE[Y_{j}(1) - Y_{j}(0)] ,\qquad\qquad j=1,\ldots,p.\label{eq:ATE}
    \end{align}
    By focusing on the expected potential outcomes, we next describe the identification condition and semiparametric inferential results based on derived outcomes $\tilde{\bY}$.

    \paragraph{Identification.}    
    We require two extra classic causal assumptions for observational studies.

    \begin{assumption}[Positivity]\label{asm:positivity}
        The propensity score $\pi_a(\bW):=\PP(A=a \mid \bW) \in (0,1)$.
    \end{assumption}

    \begin{assumption}[No unmeasured confounders]\label{asm:nuc}
        $A\indep \Xb(a)\mid \bW$, for all $a\in\{0,1\}$.
    \end{assumption}

    The above assumptions on the propensity score and the potential responses are standard for observational studies in the causal inference literature \citep{imbens2015causal,kennedy2022semiparametric}.    
    \Cref{asm:positivity} suggests that both treated and control units of each subject can be found for any value of the covariate with a positive probability.
    \Cref{asm:nuc} ensures that the treatment assignment is fully determined by the observed covariate $\bW$.
    These assumptions are required to estimate functionals of $\Xb(a)$ with observed variables $(A,\bW,\Xb)$.
    
    Let $\bZ = (A,\bW,\Xb,\bY)$ denote the tuple of observed random variables and unobserved outcomes $\bY$.
    Because the outcome $\bY$ is not observed for each subject, we are interested in constructing a proxy $\tilde{\bY}= g(\Xb)$ of $\bY$ from repeated measurements $\bX_1,\ldots,\bX_m$ from the same subject. 
    Analogously, we denote $\tilde{\bY}(a) := g(\Xb(a))$ for the potential outcomes and quantify the bias as $\bDelta_m(a) := \EE[\tilde{\bY}(a) \mid \bW,\bS(a)] - \bY(a)$.
    Below, we introduce a notion of asymptotic unbiased estimate in \Cref{def:unbias}, where the expected bias is negligible uniformly over multiple outcomes.
    
    \begin{definition}[Asymptotic unbiased estimate]\label{def:unbias}
        For $a\in\{0,1\}$, the derived outcome $\tilde{\bY}(a)$ is asymptotic unbiased to $\bY(a)$ if the bias tends to zero: $\max_{1\leq j \leq p}|\EE[\Delta_{mj}(a)]| = o(1)$ as $m\rightarrow\infty$.
    \end{definition}

    Note that when $\tilde{\bY}(a)$ is marginally unbiased, i.e., $\EE[\bY(a)] = \EE[\tilde{\bY}(a)]$, it also implies that $\tilde{\bY}(a)$ is asymptotically unbiased to $\bY(a)$.
    Therefore, our framework also includes the common setting where all the outcomes $\bY(a) = \tilde{\bY}(a) = \bX_1(a)$ (with $m=1$) are observed.
    When $\tilde{\bY}(a)$ is an asymptotic unbiased estimate of $\bY(a)$, Lemma 1 from \citet{qiu2023unveiling} suggests that the counterfactual of unobserved outcomes can be identified asymptotically, as detailed in \Cref{prop:identification}.
    \begin{proposition}[Identification of linear functionals]\label{prop:identification}
        Under \Crefrange{asm:consistency}{asm:nuc}, if $\tilde{\bY}(a)$ is asymptotically unbiased to $\bY(a)$, then $\EE[\bY(a)]$ can be identified by $\EE[\EE[\tilde{\bY}\mid \bW,A=a]]$ as $m\rightarrow\infty$.
    \end{proposition}

    \paragraph{Semiparametric inference.}
    When the target causal estimands are the expectation of the potential outcomes $\tau_j = \EE[Y_j(a)]$ for $j=1,\ldots,p$, one can adopt results from \Cref{sec:semi-inf-multi-outcomes} to establish the asymptotic normality of certain estimators under proper assumptions on the convergence rate of the nuisance function estimates.
    However, because $\bY$ is unobservable, we are not able to directly estimate its influence function and hence its influence-function-based one-step estimator \eqref{eq:one-step-est}.
    Instead, we can rely on the influence function of $\tilde{\tau}_j=\EE[\tilde{Y}_j(a)]$: $\tilde{\varphi}_j(\bZ; \PP) = \ind\{A=a\} \pi_a(\bW)^{-1} (\tilde{Y}_j - \mu_{aj}(\bW) ) + \mu_{aj}(\bW) - \tilde{\tau}_j(\PP)$, where $\pi_a(\bW) = \PP(A=a\mid\bW)$ and $\mu_{aj}(\bW) = \EE[\tilde{Y}_j\mid A=a,\bW]$ for $j=1,\ldots,p$, are the propensity score and regression functions, respectively.
    This, in turn, yields an analog of the one-step estimator \eqref{eq:one-step-est}: 
    \begin{align*}
        \hat{\tau}_j(\PP) &:= \tilde{\tau}_j(\hat{\PP}) + \PP_n\{\tilde{\varphi}(\bZ;\hat{\PP})\},
    \end{align*}
    which further implies the decomposition of the estimation error for the causal estimand $\tau_j$:
    \begin{align}
        \hat{\tau}_j(\PP) - \tau_j(\PP) = T_{{\rm S},j} + T_{{\rm E},j} + T_{{\rm R},j} + \EE[\Delta_{mj}], \label{eq:decomp-tilde}
    \end{align}
    where the sample average term $T_{{\rm S},j}$, the empirical process term $T_{{\rm E},j}$ and the reminder term $T_{{\rm R},j}$ are as in \eqref{eq:decomposition} with $\varphi_j$ replaced by $\tilde{\varphi}_j$.
    The asymptotic variance ${\sigma}_j^2 = \VV[\tilde{\varphi}_j(\bZ;\hat{\PP})]$ can be estimated by the empirical variance $\hat{\sigma}_j^2=\VV_n[\tilde{\varphi}_j(\bZ;\hat{\PP})]$ analogously.
    However, the application of \Cref{lem:emp-process-term,lem:var} would require the verification of conditions for the perturbed influenced function $\tilde{\varphi}_j$ instead of ${\varphi}_j$.
    Similar ideas apply to the one-step and doubly robust estimators of other target estimands.

    \subsection{Beyond average treatment effects}
    For single-cell gene expressions exhibiting different scales and skew-distributed, simply comparing the average treatment effects \eqref{eq:ATE} may not be reliable.
    One approach to improve on the naive estimand is to consider standardized average treatment effects (STE):
    \begin{align}
        \tau_{j}^{\STE} &= \frac{\EE[Y_{j}(1) - Y_{j}(0)]}{\sqrt{\Var[Y_{j}(0)]}} ,\qquad\qquad j=1,\ldots,p\label{eq:STE}
    \end{align}
    which allows for consistent and comparative analysis across different scales and variances, enhancing the interpretability and comparability of treatment effects in diverse and complex datasets \citep{kennedy2019estimating}.
    Another approach is to consider quantile effects (QTE):
    \begin{align}
        \tau_{j}^{\QTE_\varrho} &= Q_{\varrho}[Y_{j}(1)] - Q_{\varrho}[Y_{j}(0)],\qquad j=1,\ldots,p,\label{eq:QTE}
    \end{align}
    where $Q_{\varrho}[U]$ denote the $\varrho$-quantile of random variable $U$.
    In particular, when $\varrho=0.5$, the $\varrho$-quantile equals the median $Q_{\varrho}(U)= \Median(U)$, and we reveal one of the commonly used robust estimand $\tau_{j}^{\QTE} = \Median[Y_{j}(1)] - \Median[Y_{j}(0)]$ for location-shift hypotheses.
    QTE may be more robust and less affected by the outliers of gene expressions \citep{kallus2019localized,chakrabortty2022general}.

    Note that the identification condition and semiparametric inferential results in \Cref{subsec:general} do not apply directly to target estimands other than ATE.
    Therefore, efforts are required to generalize the results to include STE and QTE for multiple derived outcomes.
    This demonstrates the utility and validity of our semiparametric inferential framework on one-step estimators defined through the von Mises expansion and the formulation of estimating equations, respectively, as investigated next.

\section{Doubly robust estimation}\label{sec:standarized-effects}    
    In this section, we analyze the DR estimators for STE and QTE, which exemplify the application of general theoretical results in \Cref{sec:semi-inf-multi-outcomes} to specific target estimands.

    \subsection{Standarized average effects}\label{subsec:standarized-STE}

    Recall the standardized average treatment effects $\tau_j^{\STE}$ defined in \eqref{eq:STE}, for $j=1,\ldots,p$.
    The following lemma provides the identified forms of STE based on observational data.

    \begin{lemma}[Identification of standarized average effects]\label{lem:identification-STE}
        Under \Crefrange{asm:consistency}{asm:nuc}, if $\Var[Y_{j}(0)]>0$ and $\tilde{Y}_j(a)^k$ is asymptotically unbiased to $Y_j(a)^k$ for $k=1,2$ and $a=0,1$ such that $k+a\leq 2$, i.e., the bias of the derived outcomes $\Delta_{mkj}(a) := \EE[\tilde{Y}_j(a)^k \mid \bW,\bS(a)] - Y_j(a)^k$ satisfies that $\delta_m :=\max_{k,a}\max_{1\leq j\leq p}|\EE[\Delta_{mkj}(a)]|= o(1)$, then as $m\rightarrow\infty$, the standardized average treatment effect $\tau_j^{\STE}$ can be identified by $\tau_j^{\STE} = \tilde{\tau}_j^{\STE} + \op(1)$ where
        \begin{align}
            \tilde{\tau}_j^{\STE} := \frac{\EE[\EE(\tilde{Y}_j \mid A=1, \bW)] - \EE[\EE(\tilde{Y}_j \mid A=0, \bW)]}{\sqrt{ \EE[\EE(\tilde{Y}_j^2 \mid A=0, \bW)] - \EE[\EE(\tilde{Y}_j \mid A=0, \bW)]^2 }}. \label{eq:STE-ttau}
        \end{align}
    \end{lemma}
    As suggested by \Cref{lem:identification-STE}, estimating STE requires estimating the conditional expectation of $\tilde{Y}_j^k$ given $A=a$ and $\bW$.
    For this purpose, we consider the DR estimator of treatment effect $\EE[\tilde{Y}_{j}(1)]$:
    \begin{align*}
        \tilde{\phi}_{akj}(\bZ; \pi_a,\bmu_a) &:= \frac{\ind\{A=a\}}{\pi_a(\bW)} (\tilde{Y}_j^k - \mu_{akj}(\bW) ) + \mu_{akj}(\bW),
    \end{align*}    
    where $\bmu_a:\RR^d\rightarrow\RR^{2\times p}$ is the mean regression function with entry $\mu_{akj}(\bW)=\EE[\tilde{Y}_j^k\mid \bW,A=a]$ and $\pi_a(\bW) = \PP(A=a \mid \bW)$ is the propensity score function.
    By plugging in the DR estimators for individual counterfactual expectations consisting in \eqref{eq:STE-ttau}, we obtain a natural estimator for the STE:
    \begin{align}
        \hat{\tau}_j^{\STE} &= \frac{\PP_n\{\tilde{\phi}_{11j}(\bZ; \hat{\pi}_1,\hat{\bmu}_1)  - \tilde{\phi}_{01j}(\bZ; \hat{\pi}_0,\hat{\bmu}_0) \}}{ \sqrt{\PP_n\{\tilde{\phi}_{02j}(\bZ; \hat{\pi}_0,\hat{\bmu}_0) \} -\PP_n\{\tilde{\phi}_{01j}(\bZ; \hat{\pi}_0,\hat{\bmu}_0) \}^2 }},\label{eq:STE-est}
    \end{align}
    which is also the DR estimator of $\tilde{\tau}_j^{\STE}$.
    The following theorem shows that under mild conditions, the above estimator $\hat{\tau}_j^{\STE}$ is doubly robust for estimating ${\tau}_j^{\STE}$, with the remainder terms uniformly controlled over all outcomes.

    \begin{theorem}[Linear expansion of STE]\label{thm:lin-STE}
        Under \Crefrange{asm:consistency}{asm:nuc} and the identification condition in \Cref{lem:identification-STE}, consider the one-step estimator \eqref{eq:STE-est}, where $\PP_n$ is the empirical measure over $\cD=\{\bZ_1,\ldots,\bZ_n\}$ and $(\hat{\pi}_a,\hat{\bmu}_a)$ is an estimate of $({\pi}_a,{\bmu}_a)$ for $a=0,1$ from samples independent of $\cD$.
        Suppose the following hold for $k=1,2$ and $a=0,1$ with probability tending to one:
        \begin{enumerate}[label=(\arabic*)]
            \item Boundedness: There exists $c,C>0$ and $\epsilon\in(0,1)$ such that $\max\{|Y_j|,|\tilde{Y}_j|\}<C$, $\max\{\|\mu_{akj}\|_{L_{\infty}}$, $\|\hat{\mu}_{akj}\|_{L_{\infty}}\}<C$, $\Var[Y_{j}(0)]>c$ for all $j\in[p]$, and $\pi_a, \hat{\pi}_a \in[\epsilon,1-\epsilon]$.

            \item\label{cond:nuisance} Nuisance: The rates of nuisance estimates are $\max_{j\in[p]}\|\hat{\mu}_{akj}-\mu_{akj}\|_{L_{2}} = \cO(n^{-\alpha})$ and $\|\hat{\pi}_a-\pi_a\|_{L_2} = \cO(n^{-\beta})$ for some $\alpha,\beta\in(0,1/2)$ such that $\alpha+\beta>1/2$.
            
        \end{enumerate}
        Then as $m,n,p\rightarrow\infty$, it holds that $\hat{\tau}_j^{\STE} - \tau_j^{\STE} = \PP_n\{\tilde\varphi_{j}^{\STE}\} + \varepsilon_j$, $j=1,\ldots,p$,
        where the residual terms satisfy $\max_{j\in[p]}|\varepsilon_j| = \Op(n^{-(\alpha+\beta)} + \vartheta^{\STE}\sqrt{(\log p)/n}  + (\log p) /n + \delta_m)$ with $\vartheta^{\STE}:=n^{-(\alpha \wedge \beta)}$ and the influence function is given by
        \begin{align}
            \tilde\varphi_{j}^{\STE} 
            &= 
            \frac{\tilde{\phi}_{11j}-\tilde{\phi}_{01j}}{\sqrt{\Var[\tilde{Y}_{j}(0)]}} -\tilde{\tau}_j^{\STE}\left[\frac{\tilde{\phi}_{02j}+\mathbb{E}[\tilde{Y}_{j}(0)^2] - 2 \mathbb{E}[\tilde{Y}_{j}(0)] \tilde{\phi}_{01j}}{2 \Var[\tilde{Y}_{j}(0)]}\right]. \label{eq:varphi-STE}
        \end{align}
    \end{theorem}

    The proof of \Cref{thm:lin-STE} requires the analysis of the linear expansions for the individual counterfactual expectations $\EE[Y_j^k(a)]$ (see \Cref{lem:ATE}). 
    It then requires the application of the delta method to derive the uniform convergence rates of the residuals over multiple outcomes.
    For the residuals' rate, the term $n^{-(\alpha+\beta)}$ is the product of the two nuisance estimation errors, which shows the benefit of the double robustness property, while the term $\vartheta^{\STE}\sqrt{(\log p)/n}  + (\log p) /n$ is related to the empirical process terms of individual counterfactuals by applying \Cref{lem:emp-process-term}.
    From triangular-array central limit theorem (\Cref{lem:lindeberg}), a direct consequence of \Cref{thm:lin-STE} is the asymptotic normality of individual STE estimators, as presented in \Cref{prop:AN-STE}.

    \begin{corollary}[Asymptotic normality]\label{prop:AN-STE}
        Under conditions in \Cref{thm:lin-STE}, when $(\vartheta^{\STE}\vee n^{-1/4})\sqrt{\log p} =o(1)$, $\delta_m = o(n^{-1/2})$ and $\sigma_j^2:= {\Var}(\tilde\varphi_{j}^{\STE}(\bZ;{\pi},{\bmu}))\geq c$ for some constant $c>0$, it holds~that,
        \[
            \qquad\qquad\qquad\sqrt{n}(\hat{\tau}_j^{\STE} - \tau_j^{\STE}) \dto\cN(0, \sigma_j^2),\qquad j=1,\ldots,p.
        \] 
    \end{corollary}
    Compared to \Cref{def:unbias}, \Cref{prop:AN-STE} requires a stronger condition on the rate of the bias $\delta_m$, which is mild.
    For instance, when $\bX_1(a),\ldots,\bX_m(a)$ are $\iid$ conditional on $(\bW,\bS(a))$, the bias is zero, i.e., $\EE[\Delta_{mkj}(a)]\equiv 0$;
    when they are weakly dependent, for example, \citet[Proposition S1]{qiu2023unveiling} show that the bias is of order $o(n^{-1/2})$ under the $\beta$-mixing condition when $n^{1/2}\log p=o(m)$.
    
    Because the influence function of STE \eqref{eq:STE-est} is a complicated function of all the nuisances and the observations in $\cD$, it is hard to show usual sample variance of the estimated influence function $\hat{\varphi}_j^{\STE} = \tilde\varphi_{j}^{\STE}(\bZ;\{\hat{\pi}_a,\hat{\bmu}_a\}_{a\in\{0,1\}})$ provides a consistent estimate.
    In the following proposition, we thus rely on extra independent observations to estimate the asymptotic variance.
    However, one can employ the cross-fitting procedure \citep{chernozhukov2018double} on $\cD$ to decouple the dependency of $\hat{\varphi}_j^{\STE}$ and the observations used to compute the empirical variance.
    This ensures that the variance estimation errors are of polynomial rates of $n$ uniformly in $p$ outcomes when $\log (p)/n \leq Cn^{-c}$.

    \begin{proposition}[Consistent variance estimates]\label{prop:Var-STE}
        Under the same conditions in  \Cref{thm:lin-STE}, let $\hat{\varphi}_j^{\STE} $ be the estimated influence function \eqref{eq:varphi-STE} with $(\EE[\tilde{Y}_j(0)], \EE[\tilde{Y}_j(0)^2],\tilde{\tau}_j^{\STE})$ estimated by the doubly robust estimators on $\cD$ , and $\PP_n'$ be the empirical measure over a separate independent sample $\cD'=\{\bZ_{n+1},\ldots,\bZ_{2n}\}$.
        Define the sample variance on $\cD'$ as $\hat{\sigma}_j^2 = \VV_n'(\hat{\varphi}_j^{\STE})$.
        It holds that $\max_{j\in[p]}|\hat{\sigma}_j^2 - {\sigma}_j^2| = \Op(r_{\sigma}^{\STE})$ where $r_{\sigma}^{\STE} = \sqrt{\log p / n} + \vartheta^{\STE}$.
    \end{proposition}

    \subsection{Quantile effects}\label{subsec:standarized-quantile}

    In practice, examining quantile effects offers a robust alternative to mean-based analysis, particularly when confronted with highly variable treatment assignment probabilities and heavy-tailed outcomes. 
    Estimating causal effects on the mean is a challenging problem in such scenarios because the signal-noise ratio is generally small.
    In cases where the mean is undefined but the median exists (such as the Cauchy distribution), using the median may result in more powerful tests for the location-shift hypothesis \citep{diaz2017efficient}.

    We first introduce the DR estimator for the median effect \eqref{eq:QTE} when $\varrho=0.5$, while the proposal naturally extends to other quantile levels $\varrho$ as well.
    For $j\in[p]$, let $\theta_{aj}$ be the $\varrho$-quantile of the counterfactual response ${Y}_j(a)$, which solves the following equation:
    \begin{align}
        0 = \EE[{\psi}({Y}_j(a), \theta)], 
        \text{ where } \psi(y, \theta) &:= \ind\{y \leq \theta\} - \varrho.\label{eq:quantile-pop}
    \end{align}
    Since the potential outcome $Y_j(a)$ is not directly observed, we need to rely on the counterfactual derived outcomes $\tilde{Y}_j(a)$ to identify the quantile of $Y_j(a)$. The following lemma summarizes the identification results of general M-estimators for functionals of $Y_j(a)$ using the derived outcomes $\tilde{Y}_j(a)$.
    \begin{lemma}[Identification of M-estimators]\label{lem:identification-es}
        Under \Crefrange{asm:consistency}{asm:nuc}, consider the causal estimand $\theta_{aj}$ as the solution to the estimation equations:
        \[
            \qquad\qquad\qquad M_j(\theta) = \EE[F_j(Y_j(a), \theta)] = 0,\qquad j=1,\ldots,p,
        \]
        where $M_j$ is differentiable and the magnitude of its derivative $|M_j'|$ is uniformly lower bounded around $\theta_{aj}$: $\min_{1 \leq j \leq p} \inf_{\theta \in \cB(\theta_{aj}, \delta)}|M_j'(\theta)| \geq c >0 $ for some constant $\delta >0$.
        Suppose $\tilde{Y}_j(a)$'s are derived outcomes such that $F_j(\tilde{Y}_j(a),\theta)$ is asymptotically unbiased to $F_j(Y_j(a),\theta)$, i.e. as $m\rightarrow\infty$, $\Delta_{mj}(a,\theta) = \EE[F_j(\tilde{Y}_j(a),\theta)\mid \bS(a),\bW] - F_j({Y}_j(a),\theta)$ satisfies that $\delta_m := \max_{j\in[p]} \sup_{\theta \in \cB(\theta_{aj}, \delta)} $ $|\EE[\Delta_{mj}(a,\theta)]| = o(1)$.
        Let $\tilde{\theta}_{aj} \in \cB(\theta_{aj}, \delta)$ be the solution  to the estimating equation 
        \[\EE[\EE[F_j(\tilde{Y}_j,{\theta}) \mid A=a,\bW]] = 0,
        \] 
     then $\theta_{aj}$ can be identified by $\tilde{\theta}_{aj}$ as $m\rightarrow\infty$ .
    \end{lemma}
    Under conditions in Lemma \ref{lem:identification-es} with $F_j(Y_j(a),\theta) = \psi(Y_j(a), \theta)$, we can focus on estimating the quantile of $\tilde{Y}_j(a)$ to approximate the quantile of $Y_j(a)$. 
    Specifically, consider a doubly robust expansion of the above question: $0 = \EE[\psi(\tilde{Y}_j, \theta)] = - \EE[\omega_{aj}(\bZ, \theta)]$, where the estimating function is given by 
    \begin{align*}
        \omega_{aj}(\bZ, \theta) &= \frac{\ind\{A=a\}}{\pi_a(\bW)}({\nu}_{aj}(\bW,{\theta})- \psi(\tilde{Y}_{j}, {\theta})  )  - {\nu}_{aj}(\bW,{\theta}).
    \end{align*}
    Here, $\nu_{aj}(\bW,\theta) = \EE[\psi(\tilde{Y}_j, \theta) \mid \bW,A=a] = \PP(\tilde{Y}_j \leq \theta\mid\bW,A=a) - \varrho$ is the excess (conditional) cumulative distribution functions (cdfs) of $\tilde{Y}_j(a)$, and $\pi_a$ is the propensity score function as before.
    One may then expect to obtain an estimator of $\theta_{aj}$ by solving the empirical version of \eqref{eq:quantile-pop} for $\theta$:
    \begin{align}
        0 &= \PP_n[\hat{\omega}_{aj}(\bZ, \theta)] ,\label{eq:quantile-emp}
    \end{align}
    where $
        \hat{\omega}_{aj}(\bZ, \theta) = \frac{\ind\{A=a\}}{\hat{\pi}_a(\bW)}(\hat{\nu}_{aj}(\bW,{\theta})- \psi(\tilde{Y}_j, {\theta})  )  - \hat{\nu}_{aj}(\bW,{\theta})$ is the estimated influence function and $(\hat{\pi}_a,\hat{\nu}_{aj}+\varrho)$ are the estimated propensity score and cdf functions with range $[0,1]$.

    However, directly solving \eqref{eq:quantile-emp} is not straightforward due to its non-smoothness and non-linearity in $\theta$. 
    A reasonable strategy to adopt instead is a one-step update approach \citep{van2000asymptotic,tsiatis2006semiparametric} using the influence function:
    \begin{align}
        \hat{\theta}_{aj} = \hat{\theta}_{aj}^{\init} + \frac{1}{\hf_{aj}(\hat{\theta}_{aj}^{\init})} \PP_n[\hat{\omega}_{aj}(\bZ, \hat{\theta}_{aj}^{\init})],\label{eq:quantile-est}
    \end{align}
    where $\hat{\theta}_{aj}^{\init}$ is an initial estimator of $\theta_{aj}$ and $\hf_{aj}(\hat{\theta}_{aj}^{\init})$ is the estimated density of $\tilde{Y}_j(a)$ at $\hat{\theta}_{aj}^{\init}$.

    For $a=0,1$, let $\bbf_a=(f_{aj})_{j\in[p]}$, $\tilde{\btheta}_a = (\tilde{\theta}_{aj})_{j\in[p]}$, and $\bnu_a=(\nu_{aj})_{j\in[p]}$ be the vectors of true density functions, the $\varrho$-quantiles, and the excess cdf functions cdfs of $\tilde{Y}_j(a)$, respectively.
    Moreover, let $\hat{\bbf}_a=(\hf_{aj})_{j\in[p]}$, $\hat{\btheta}_a = (\hat{\theta}_{aj})_{j\in[p]}$, and $\hat{\bnu}_a=(\hat{\nu}_{aj})_{j\in[p]}$ be the corresponding vectors of estimated nuisances.
    Based on \eqref{eq:quantile-est}, an estimator for $\tau_j^{\QTE}$ is given by
    \begin{align}
        \hat{\tau}^{\QTE}_j = \hat{\theta}_{1j}-\hat{\theta}_{0j} .\label{eq:QTE-est}
    \end{align}
    The following theorem provides the asymptotic normality of the one-step estimator \eqref{eq:QTE-est}.

\begin{theorem}[Linear expansion of QTE]\label{thm:lin-QTE}
    Under \Crefrange{asm:consistency}{asm:nuc}, suppose the identification conditions in Lemma \ref{lem:identification-es} hold with $F_j=\psi$ defined in \eqref{eq:quantile-pop}.
    Consider the one-step estimator \eqref{eq:QTE-est} for the median treatment effect, where $\PP_n$ is the empirical measure over $\cD=\{\bZ_1,\ldots,\bZ_n\}$ and $(\hat{\btheta}_{a}^{\init},{\bm\hf}_a,\hat{\pi}_a,\hat{\bnu}_a)$ is an estimate of $(\tilde{\btheta}_a,{\bm f}_a,{\pi}_a,{\bnu}_a)$ from samples independent of $\cD$ for $a=0,1$.
    Suppose the following conditions hold for $a=0,1$ with probability tending to one:
    \begin{enumerate}[label=(\arabic*)]
        \item\label{asm:quantile-boundedness} Boundedness: The quantile $\theta_{aj}$ is in the interior of its parameter space.
        There exists $C,c>0$ and $\epsilon,\delta\in(0,1)$ such that $\max_{1 \leq j \leq p}\max\{|Y_j|,|\tilde{Y}_j|\}<C$, and $\pi_a,\hat{\pi}_a \in[\epsilon,1-\epsilon]$, and $f_{aj}$ is uniformly bounded : $c \leq  f_{aj} \leq C$ for all $j $ and has a bounded derivative in a neighborhood $\cB(\tilde{\theta}_{aj},\delta)$ for all $j\in[p]$: $\max_{1 \leq j \leq p}\max_{\theta \in \cB(\tilde{\theta}_{aj},\delta)}|f_{aj}'(\theta)|\leq C$.

        \item \label{asm:quantile-init} Initial estimation: The initial quantile and density estimators satisfy that $\max_{j\in[p]}|\hat{\theta}_{aj}^{\init} - \tilde{\theta}_{aj}| = \cO(n^{-\gamma} )$ and $ \max_{j\in[p]}|\hf_{aj}(\hat{\theta}_{aj}^{\init}) - f_{aj}(\tilde{\theta}_{aj})| = \cO(n^{-\kappa} )$ with $\gamma>1/4,\kappa>0$ such that $\gamma+\kappa>1/2$.
        
        \item Nuisance: The rates of nuisance estimates satisfy $\max_{j\in[p]}\sup_{\theta\in \cB(\tilde{\theta}_{aj},\delta)}\|\hat{\nu}_{aj}(\cdot,\theta)-\nu_{aj}(\cdot,\theta)\|_{L_{2}} = \cO(n^{-\alpha})$ and $\|\hat{\pi}_a-\pi_a\|_{L_2} = \cO(n^{-\beta})$ for some $\alpha,\beta\in(0,1/2)$ such that $\alpha+\beta>1/2$.
    \end{enumerate}
    Then as $m,n,p\rightarrow\infty$, it holds that $\hat{\tau}_j^{\QTE} - \tau_j^{\QTE} = \PP_n\{\tilde\varphi_{j}^{\QTE}\} + \varepsilon_j$, $j=1,\ldots,p$, where the residual term satisfy $\max_{j\in[p]}|\varepsilon_j| = \Op(\vartheta^{\QTE}\sqrt{(\log p)/n} + (\log p)/n + n^{-(\alpha+\beta)\wedge(\gamma+\kappa)\wedge(2\gamma)}+\delta_m) $ with $\vartheta^{\QTE}:=n^{-(\alpha \wedge \beta \wedge \kappa \wedge \frac{\gamma}{2})}$ and the influence function is given by
    \begin{align*}
        \tilde\varphi_j^{\QTE}(\bZ; \{\tilde{\btheta}_a,{\bm f}_a,{\pi}_a,{\bnu}_a\}_{a\in\{0,1\}}) &= [f_{1j}(\tilde{\theta}_{1j})]^{-1}{\omega}_{1j}(\bZ, \tilde{\theta}_{1j}) - [f_{0j}(\tilde{\theta}_{0j})]^{-1}{\omega}_{0j}(\bZ, \tilde{\theta}_{0j}).
    \end{align*}  
\end{theorem}

    \Cref{app:subsec:est-QTE} provide details for obtaining initial estimators for the quantiles and the corresponding densities.
    Similar to STE, we can also obtain individual asymptotic normality for the DR estimator \eqref{eq:QTE-est} of QTE and consistently estimate its variance.

    \begin{proposition}[Asymptotic normality of QTE]\label{prop:AN-QTE}
        Under the conditions in \Cref{thm:lin-QTE}, when $(\vartheta^{\QTE}\vee n^{-1/4})\sqrt{\log p} =o(1)$, $\delta_m = o(n^{-1/2})$ and $\sigma_j^2 := {\Var}(\tilde\varphi_{j}^{\QTE})\geq c$ for some constant $c>0$, it holds
        \begin{align*}
            \qquad\qquad\qquad\sqrt{n}(\hat{\tau}^{\QTE}_j - \tau_j^{\QTE}) \dto \cN(0,{\sigma}_j^2),\qquad j=1,\ldots,p.
        \end{align*}
        Define the sample variance $\hat{\sigma}_j^2 = \VV_n(\hat{\varphi}_j^{\QTE})$ for the estimated influence function $\hat{\varphi}_j^{\QTE} := \tilde\varphi_j^{\QTE}(\bZ; \{\hat{\btheta}_a^{\init},$ $\hat{\bm f}_a,\hat{\pi}_a,\hat{\bnu}_a\}_{a\in\{0,1\}})$.
        It further holds that $\max_{j\in[p]}|\hat{\sigma}_j^2-{\sigma}_j^2| = \Op(r_{\sigma}^{\QTE})$ where $r_{\sigma}^{\QTE}=(\log p)/n + \sqrt{(\log p)/n} \vartheta^{\QTE} + \vartheta^{\QTE}$.
        
    \end{proposition}
    
    Apart from the mild rate requirement on the nuisance functions, no metric entropy conditions are assumed in \Cref{thm:lin-QTE} and \Cref{prop:AN-QTE}. 
    This allows one to estimate nuisances with machine learning methods and achieve asymptotical normality with cross-fitting.
    While the doubly-robust estimators for QTE have also been considered by \citet{chakrabortty2022general,kallus2019localized} for a single outcome ($p=1$), they both require metric entropy or Donsker class conditions.

\section{Simultaneous inference}\label{sec:inference}

    \subsection{Large-scale multiple testing}

    For a target estimand $\tau_j \in\{\tau_j^\STE,\tau_j^\QTE\}$, the asymptotic normality established in \Cref{prop:AN-STE} and \Cref{prop:AN-QTE} can be utilized to test the null hypotheses $H_{0j}:\tau_j=\tau_j^*$ for $j=1,\ldots,p$.
    This implies that one can control the Type-I error of the individual tests using the statistics $t_j=\sqrt{n} (\hat{\tau}_j-\tau_j^*)/\hat{\sigma}_j$, with empirical variance given in \Cref{prop:Var-STE,prop:AN-QTE}.
    The confidence intervals for individual causal estimates can also be constructed.
    To conduct simultaneous inference, however, the tests above are too optimistic when multiple tests are of interest.
    Therefore, to obtain valid inferential statements, we must perform a multiplicity adjustment to control the family-wise error explicitly.
    This subsection provides simultaneous tests and confidence intervals for causal effects with multiple outcomes.

    For $j\in[p]$, let $\varphi_{ij}=\tilde{\varphi}_j(\bZ_i) $ and $\hat{\varphi}_{ij} = \hat{\varphi}_{j}(\bZ_i) $ be the influence function value and its estimate evaluated at the $i$th observation $\bZ_i = (A_i,\bW_i,\Xb_i)$, as defined in \Cref{prop:Var-STE,prop:AN-QTE} for $\tau_j$ being $\tau_j^\STE$ and $\tau_j^\QTE$, respectively.
    We require a condition from \citet[Theorem J.1]{chernozhukov2013gaussian} for feasible inference. 

    \begin{assumption}[Bounded variances and covariances]\label{asm:cor}
        There exist a constant $a,c_1\in(0,1)$ and a set of informative hypotheses $\cA^*\subseteq[p]$ such that $|\cA^*|\geq a p$, $\max_{j\in \cA^{*c}}\sigma_j^2 = o(1)$, $\min_{j\in \cA^*}\sigma_j^2\geq c_1$ and $\max_{j_1\neq j_2\in \cA^*}$ $|\Cor(\varphi_{1j1}, \varphi_{1j_2}) | \leq 1 - c_1$.
    \end{assumption}

    When the value of $\sigma_j$ is 0, the population distribution of the $j$th influence function is degenerated and has no variability. 
    In \Cref{asm:cor}, the first condition precludes the existence of such super-efficient estimators over $\cA^*$, which is commonly required even in classical settings where the number of variables $p$ is small compared to the sample size $n$ \citep{belloni2018high}.
    In practice, one can use a small threshold $\cn$ to screen out outcomes that have small variations and obtain a set of informative outcomes $\cA_1=\{j\in[p]\mid \hat{\sigma}_j\geq\cn\}$.
    
    For DR estimators derived in the previous section, the following Gaussian approximation result over a family of null hypotheses allows for data-dependent choices of the set of hypotheses and suggests a multiplier bootstrap procedure \citep{chernozhukov2013gaussian} for simultaneous inference.

    \begin{lemma}[Gaussian approximation for nested hypotheses]\label{prop:max-Gaussian}
        For $\tau_j=\tau_j^\STE$ and $\vartheta=\vartheta^\STE$, suppose conditions in \Cref{prop:Var-STE} and \Cref{asm:cor} hold. 
        Futher assume that there exist some constants $c_2,C_2>0$ such that $\max\{\log(pn)^7/n, \log (pn)^2\vartheta, \sqrt{n\log(pn)}\delta_m\}\leq C_2 n^{-c_2}$.
        For all $\cS \subseteq \cA^*\subseteq [p]$, define ${M}_{\cS}=\max_{j\in\cS}|\sqrt{n}(\hat{\tau}_j-\tau_j)/\hat{\sigma}_j|$, $\hat{\bvarphi}_i = (\hat{\varphi}_{ij})_{j\in\cS}$, $\hat{\bE}_{\cS}=n^{-1}\sum_{i=1}^n\hat{\bvarphi}_i\hat{\bvarphi}_i^{\top}$, and $\hat{\bD}_{\cS}=\diag((\hat{\sigma}_j)_{j\in\cS})$.
        Consider null hypotheses $H_0^{\cS}$ indexed by $\cS$ that $\forall\ j\in\cS$, $\tau_j =\tau_j^*$.
        As $m,n,p\rightarrow \infty$, it holds that
        \begin{align*}
            \sup_{H_0^{\cS}:\cS \subseteq \cA^*}\sup_{x\in\RR} | \PP(\overline{M}_{\cS} > x) - \PP(\|\bg_{\cS}\|_{\infty} > x\mid \{\bZ_i\}_{i=1}^n)| \pto 0,
        \end{align*}
        where $\bg_{\cS} \sim \cN(\zero,\hat{\bD}_{\cS}^{-1} \hat{\bE}_{\cS} \hat{\bD}_{\cS}^{-1})$.
        The conclusion also holds for $\tau_j=\tau_j^\QTE$ and $\vartheta=\vartheta^\QTE$ under conditions in \Cref{prop:AN-QTE} and \Cref{asm:cor}.
    \end{lemma}

    When $m$ is sufficiently large such that the error of derived outcomes $\delta_m$ is ignorable, the rate conditions in \Cref{prop:max-Gaussian} can be satisfied if the \emph{logarithm} of the numbers of hypotheses grows slower than $n^{\frac{1}{7}}\wedge\vartheta^{-\frac{1}{2}}$ for at least a polynomial factor of $n$.
    \Cref{prop:max-Gaussian} suggests that if $\cA_1\subseteq\cA^*$ only contains informative hypotheses, then the distribution of the maximal statistic $M_1=\max_{j\in {\cA}_1}|t_j|$ can be well approximated by $\bg_1\sim \cN(\zero, \bD_{n1}^{-1}\bE_{n1}\bD_{n1}^{-1})$, where $\bE_{n1} = n^{-1}\sum_{i=1}^n\hat{\bvarphi}_{i1}\hat{\bvarphi}_{i1}^{\top}$ is the sample covariance matrix with $\hat{\bvarphi}_{i1} = (\hat{\varphi}_{ij})_{j\in{\cA}_1} $ and $\bD_{n1} = \diag((\hat{\sigma}_j)_{j\in{\cA}_1})$ is the diagonal matrix of the estimated standard deviations.
    This allows us to simulate the null distribution efficiently using the multiplier bootstrap procedure.
    To generate $B$ bootstrap samples, for all $b=1,\ldots,B$, we first sample $n$ standard normal variables $\varepsilon_{11}^{(b)},\ldots,\varepsilon_{n1}^{(b)} \overset{\iid}{\sim}\cN(0,1)$ and then apply a linear transformation to obtain the multivariate normal vectors $\bg_1^{(b)} = (\sqrt{n}\bD_{n1})^{-1}\sum_{i=1}^n \varepsilon_{i1}^{(b)} \hat{\bvarphi}_{i1}$.
    It is easy to verify that $\bg_1^{(1)},\cdots,\bg_1^{(B)}\overset{\iid}{\sim} \cN(\zero, \bD_{n1}^{-1}\bE_{n1}\bD_{n1}^{-1})$ conditioned on $\{\bZ_i\}_{i=1}^n$.
    Based on the bootstrap samples, we can estimate the upper $\alpha$ quantile of $M_1$ by $\hat{q}_{1}(\alpha) = \inf \left\{ x \,\middle|\, B^{-1}\sum_{b=1}^B \ind\{\|\bg_1^{(b)}\|_{\infty} \leq x\} \geq 1 -\alpha \right\} $.
    To test multiple hypotheses $H_{0j}:\tau_j=\tau_j^*$ for $j\in\cA_1$, we reject those in the set $\hat{\cA} = \{j\in\cA_1\mid |t_j| > \hat{q}_1(\alpha) \}$.
    The next proposition shows that the informative hypotheses can be identified, and the family-wise error rate (FWER) can be controlled.

    \begin{proposition}[Type-I error control]\label{prop:simul-inference}
        For $(\tau_j,r_{\sigma})$ being $(\tau_j^\STE,r_{\sigma}^\STE)$ or $(\tau_j^\QTE,r_{\sigma}^\QTE)$, suppose conditions in \Cref{prop:max-Gaussian} hold .
        Let $\cV^*=\{j\mid H_{0j}\text{ is false},j=1,\ldots,p\}\cap \cA^*$ be the set of informative non-null hypotheses.
        If $\max\{r_{\sigma},\max_{j\in \cA^{*c}}\sigma_j^2\} = o(\cn)$, then as $m,n,p\rightarrow\infty$, it holds that $\lim\ \PP(\cA^* = {\cA}_1 ) = 1$ and $\limsup\  \PP(\hat{\cA}\cap \cV^{*c}\neq\varnothing) \leq \alpha$. 
    \end{proposition}

    As suggested by \Cref{prop:simul-inference}, because the lower bound of informative variance in \Cref{asm:cor} is unknown, a slowly shrinking threshold $c_n$ is needed to recover the true candidate set $\cA^*$ and control the FWER.
    In practice, one can set $c_n$ as a small value, such as 0.01, to exclude uninformative tests.
    If lowly expressed genes have already been excluded, thresholding may not be necessary.

    \subsection{False discovery rate control}
    When $p$ is large, controlling for the false discovery proportion (FDP) or the false discovery rate (FDR) is more desirable to improve the powers when performing simultaneous testing.
    The FDP is the ratio of false positives to total discoveries, while the FDR is the expected value of the FDP.
    To control the FDP, we adopt the step-down procedure \citep{genovese2006exceedance} to test the sequential hypotheses,
    \[H_{0}^{(\ell)}: \forall\ j\in\cA_{\ell},\ \tau_j = \tau_j^*,\qquad\text{versus}\qquad H_a^{(\ell)}:\exists \ j\in\cA_{\ell},\ \tau_j \neq \tau_j^*, \qquad \ell=1,2,\ldots\]
    where $\cA_1,\cA_{2},\ldots$ is a sequence of nested sets.
    The proposed multiple testing method in \Cref{algo:multiple-testing} incorporates both the Gaussian multiplier bootstrap and step-down procedure, which aims to control the FDP exceedance rate $\FDPex:=\PP(\FDP>c)$, the probability that \FDP surpasses a given threshold $c$ at a confidence level $\alpha$. 
    This provides a strengthened control on \FDP and is asymptotically powerful, as shown in the following theorem.
    
    {\spacingset{1.}\footnotesize
    \begin{algorithm}[!t]
        \caption{Multiple testing on doubly robust estimation of treatment effects}\label{algo:multiple-testing}
        \begin{algorithmic}[1]
            \REQUIRE The estimated centered influence function values $\hat{\varphi}_{ij}$, the estimated variance $\hat{\sigma}_j^2$ for $i=1,\ldots,n$ and $j=1,\ldots,p$. The FDP exceedance threshold $c$ and probability $\alpha$, and the number of bootstrap samples $B$.
            
            \STATE Initialize the iteration number $\ell=1$, the candidate set ${\cA}_1=\{j\in[p]\mid \hat{\sigma}_j^2\geq\cn\}$, the set of discoveries $\cV_1= \varnothing$, and the statistic $t_j=\sqrt{n} (\hat{\tau}_j-\tau_j^*)/\hat{\sigma}_j$ for $j\in[p]$.
            
            \WHILE{not converge}

                \STATE Draw multiplier bootstrap samples $\bg_\ell^{(b)} = (\sqrt{n}\bD_{n\ell})^{-1}\sum_{i=1}^n \varepsilon_{i\ell}^{(b)} \hat{\bvarphi}_{i\ell}$, where $\varepsilon_{i \ell}^{(b)}$'s are independent samples from $\cN(0,1)$ for $i=1,\ldots,n$ and $b=1,\ldots,B$.

                \STATE Compute the maximal statistic $M_{\ell}= \max_{j\in{\cA}_{\ell}} |t_j|$.
                
                \STATE Estimate the upper $\alpha$-quantile of $M_{\ell}$ under $H_{0}^{(\ell)}:\forall\ j\in\cA_{\ell},\ \tau_j = \tau_j^*$ by
                \begin{align*}
                    \hat{q}_{\ell}(\alpha) &= \inf \left\{ x \,\middle|\, \frac{1}{B}\sum_{b=1}^B \ind\{\|\bg_\ell^{(b)}\|_{\infty} \leq x\} \geq 1 -\alpha \right\} .
                \end{align*}

                \IF{$M_{\ell} > \hat{q}_{\ell}(\alpha)$}
                    \STATE Set $j_{\ell}= \argmax_{j\in{\cA}_{\ell}}|t_j|$, ${\cA}_{\ell+1}=  {\cA}_{\ell}\setminus \{j_{\ell}\}$, and $\cV_{\ell+1} = \cV_{\ell} \cup \{ j_{\ell}\}$.
                \ELSE
                    \STATE Declare the treatment effects in ${\cA}_{\ell}$ are not significant and stop the step-down process.   
                \ENDIF
                
                \STATE $\ell\gets \ell+1$.
            \ENDWHILE
            \STATE Augmentation: Set $\cV$ to be the union of $\cV_{\ell}$ and the $\lfloor | \cV_{\ell}| \cdot c /(1 - c) \rfloor$ elements from ${\cA}_{\ell}$ with largest magnitudes of $|t_j|$.
            \ENSURE The set of discoveries $\cV$.
        \end{algorithmic}
    \end{algorithm}}

    \begin{theorem}[Multiple testing]\label{thm:FDPex}
        Under the conditions of \Cref{prop:simul-inference}, consider testing multiple hypotheses $H_{0j}: \tau_j = 0$ versus $H_{aj}:\tau_j\neq 0$ for $j = 1, \ldots , p$ based on the step-down procedure with augmentation.
        As $m,n,p\rightarrow\infty$, the set of discoveries $\cV$ returned by \Cref{algo:multiple-testing} satisfies that
        \begin{itemize}
            
            \item (\FDPex) $\limsup\PP(\FDP >c)\leq \alpha$ where $\FDP = |\cV\cap \cV^{*c}|/|\cV|$.

            \item (Power) $\PP(\cV^* \subset \cV) \to 1$ if $\min_{j\in \{j\in[p]\mid \tau_j\neq 0\}}|\tau_j| \geq c' \sqrt{\log(p)/n }$ for some constant $c'>0$.
        \end{itemize}
    \end{theorem}

    \Cref{thm:FDPex} extend previous results by \citet{belloni2018high} for many approximate means and by \citet{qiu2023unveiling} for IPW estimators to DR estimators.
    On the one hand, \citet{belloni2018high} directly imposes assumptions on the influence functions and linearization errors, while we need to analyze the effect of nuisance functions estimation for the doubly robust estimators.
    On the other hand, \citet{qiu2023unveiling} requires sub-Gaussian assumptions and $\sqrt{n}$-consistency of the maximum likelihood estimation for the propensity score to establish Gaussian approximation for their proposed statistics, which does not apply to our problem setups.

\section{Simulation}\label{sec:simu}

    We consider a simulation setting with $p=8000$ genes and generate an active set of genes $\cV^*=\cA^*\subset[p]$ with size 200.
    We draw covariates $\bW\in\RR^d$ with $\iid$ $\cN(0,1)$ entries and the treatment $A$ follows a logistic regression model with probability $\PP(A=1\mid \bW) = 1/ ( 1 + \exp(\one_d^{\top}\bW/(d+1)) )$.
    Then, we generate the counterfactual gene expressions.
    For a gene $j$, the single-cell gene expression $X_j(0)$ is drawn from a Poisson distribution with mean $\lambda_j = \exp(\bW^{\top}\bb_j) \in\RR$ where the entries of both the coefficients $\bb_j\in\RR^d$ with 1 as the first entry and the remaining entries independently drawn from $\cN(0,1/4)$.
    The gene expressions $X_j(1)$ for $j\not\in\cV^*$ are generated from the same model, while for gene $j\in\cV^*$, we consider two treatment mechanisms that favor the mean-based and quantile-based tests, respectively; see \Cref{app:subsubsec:simu} for more details about the data generating processes.

    Next, we draw $m$ observations $\bX_1(A),\ldots,\bX_m(A)$ independently, which are summed up as the overall gene expression $\tilde{\bY}(A)$.
    Then, the observed gene expression matrix is given by $\Xb = A\Xb(1) + (1-A)\Xb(0)$ and analogously $\tilde{\bY} = A\tilde{\bY}(1) + (1-A)\tilde{\bY}(0)$.
    We then draw $n$ independent observed samples $\{(A_i,\bW_i,\Xb_i, \tilde{\bY}_i)\}_{i=1}^n$.
    The parameters are set to be $d=5$, $m=100$, $n\in\{100,200,300,400\}$.
    For nuisance function estimation, we employ Logistic regression to estimate the propensity score and Poisson generalized linear model (GLM) with the log link to estimate the mean regression functions.
    For quantile-based methods, the initial estimators of the quantile and density are described in \Cref{app:subsec:est-QTE}.

    \begin{figure}[!t]
        \centering
        \includegraphics[width=0.8\textwidth]{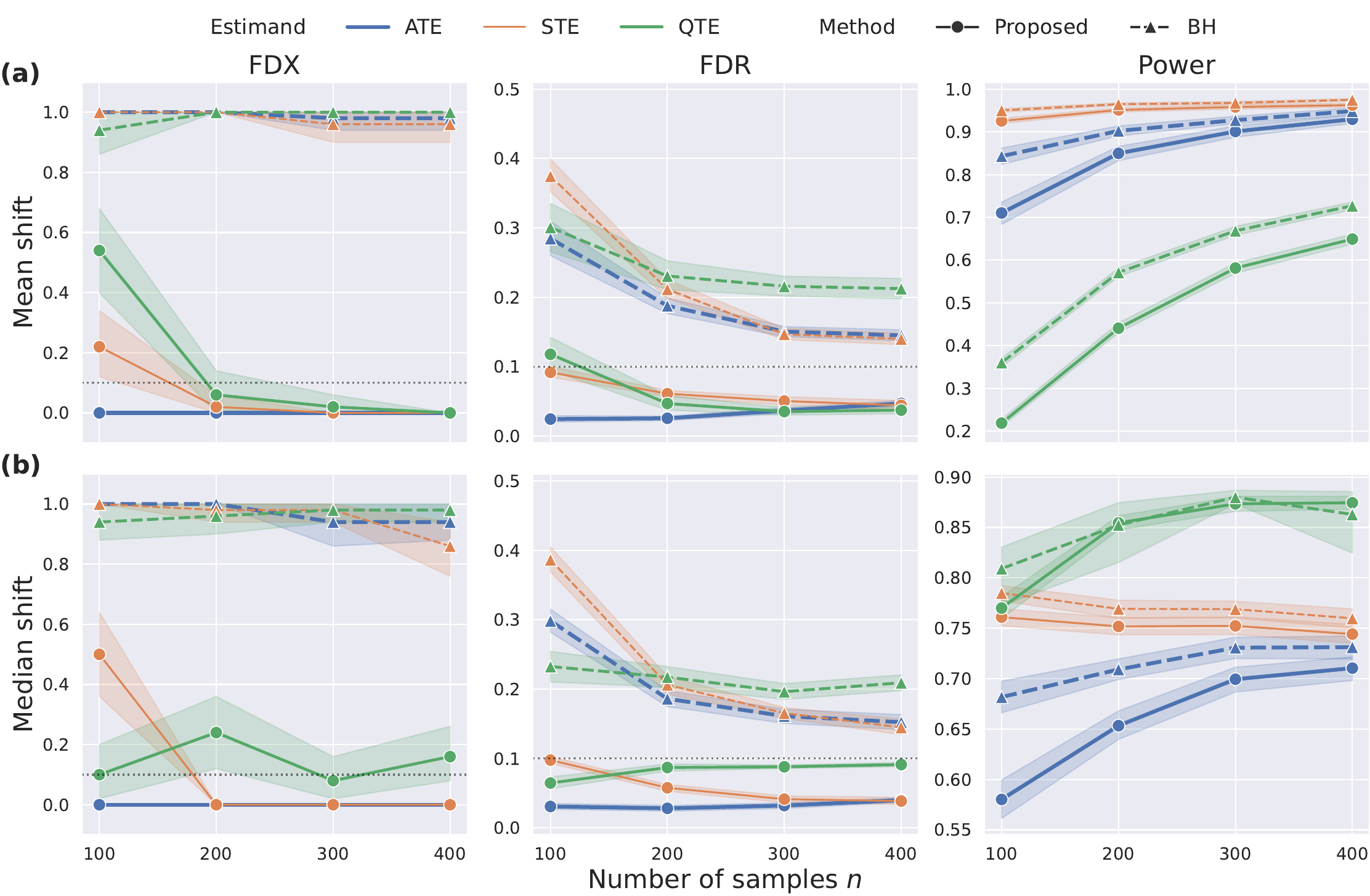}
        \caption{
        Simulation results of the hypothesis testing of $p=8000$ outcomes based on different causal estimands and FDP control methods for detecting differential signals under \textbf{(a)} mean shifts and \textbf{(b)} median shifts averaged over 50 randomly simulated datasets without sample splitting.        
        The gray dotted lines denote the nominal level of 0.1.}
        \label{fig:simu}
    \end{figure}
    
    To quantify the performance of different test statistics and multiple testing procedures, we compare the empirical \FDPex, \FDR, and power.
    We aimed to control \FDPex over 0.1 at 0.05, namely $\PP(\FDP > 0.1) \leq 0.05$.
    We also compared with the Benjamini-Hochberg (BH) procedure with targeting FDR controlled at $0.05$.
    The experiment results are summarized in \Cref{fig:simu} without sample splitting and in \Cref{fig:simu-split} of \Cref{app:subsubsec:simu} with cross-fitting.
    As shown in \Cref{fig:simu}, in the high signal-noise ratio (SNR) setting, the proposed method controls both \FDPex and \FDR at the desired level for all three causal estimands when the sample size is relatively large, i.e., $n>100$.
    On the other hand, the BH procedure fails to control the \FDPex and \FDP at any sample sizes because the p-values are not close to uniform distribution (see \Cref{fig:BH-stat}), though the gaps of \FDP become smaller when the sample size gets larger.
    Though the BH procedure has lower \FDP with sample splitting (see \Cref{fig:simu-split}), it still fails to control \FDPex for all estimands.
    This indicates that the proposed multiple testing procedure consistently outperforms the BH procedure by correctly accounting for the dependencies among the test statistics and providing valid statistical error control. 

    In the low SNR setting, we see that the quantile-based estimand has larger powers than mean-based tests, which is expected because of the designed data-generating process.
    Such a low SNR scenario is often encountered with scRNA-seq data.
    In this case, the proposed method still has better control of both \FDPex and \FDR compared to the BH procedure.
    Although the QTE test is slightly anti-conservative regarding \FDPex, it still controls the \FDR well.
    Furthermore, standardized tests are more powerful than unstandardized estimands, especially when the sample size is small.
    Overall, the results in \Cref{fig:simu} demonstrate the valid FDP control of the proposed multiple testing procedure for various causal estimands and suggest that testing based on different estimands could be helpful in different scenarios. In contrast, the commonly used BH procedure in genomics may be substantially biased due to the complex dependency among tests.

\section{Real data}\label{sec:real-data}

    \subsection{LUHMES data with CRISPR repression}\label{subsec:LUHMES-data}
        In this subsection, we aim to validate the proposed causal inference procedure with multiple outcomes on the single-cell CRISPR experiments.
        In this case, we directly measure one observation $\bY$ instead of many observations $\Xb$ per subject, simplifying the model.
        As a concrete data example, \citet{lalli2020high} employed a comprehensive single-cell functional genomics approach to unravel the molecular underpinnings of genes associated with neurodevelopmental disorders. 
        Utilizing a modified CRISPR-Cas9 system, they conducted gene suppression experiments on 13 genes linked to Autism Spectrum Disorder (ASD). 
        The experimental setup involved 14 groups of LUHMES neural progenitor cells, encompassing 13 treatment groups where each group had a specific single gene knockdown and one control group with no targeted gene suppression. 
        Single-cell RNA sequencing was applied to assess the resultant gene expression changes in response to each gene knockdown. 
        At varying stages of maturation, the cells classified as ``early'' or ``late'' resulted in 28 unique cell groups (14 groups and 2 maturation stages). 
        A critical scientific objective was to analyze and compare the transcriptional profiles across these diverse groups of neuron cells, thereby gaining insights into the transcriptional dynamics influenced by ASD-related gene modifications.

        We focus on the late-stage cells and perturbations related to 4 target genes \emph{PTEN}, \emph{CHD2}, \emph{ASH1L}, and \emph{ADNP} that are shown to be influential for neural development at this stage \citep[Figure 4E]{lalli2020high}.
        We include 6 covariates: the intercept, the logarithm of the library size, two cell cycle scores (`S.Score' and `G2M.Score'), the processing batch, and the pseudotime states (categorical with values in $\{4,5,6\}$).
        We compare gene expressions of each perturbation to the control group, and filter out genes expressed in less than 50 cells, where the resulting sample sizes are shown in \Cref{tab:LUHMES-data}.
        We aimed to control \FDPex over 0.01 at 0.05, namely $\PP(\FDP > 0.01) \leq 0.05$.
        The stringent threshold gives a reasonable number of discoveries that we can visually compare across ATE and STE tests.
        We didn't evaluate quantile-based tests in the current data analysis because of the sparse counts, which makes the quantile estimates of most of the genes exactly zero.

        \begin{figure}[!t]
            \centering
            \includegraphics[width=0.95\textwidth]{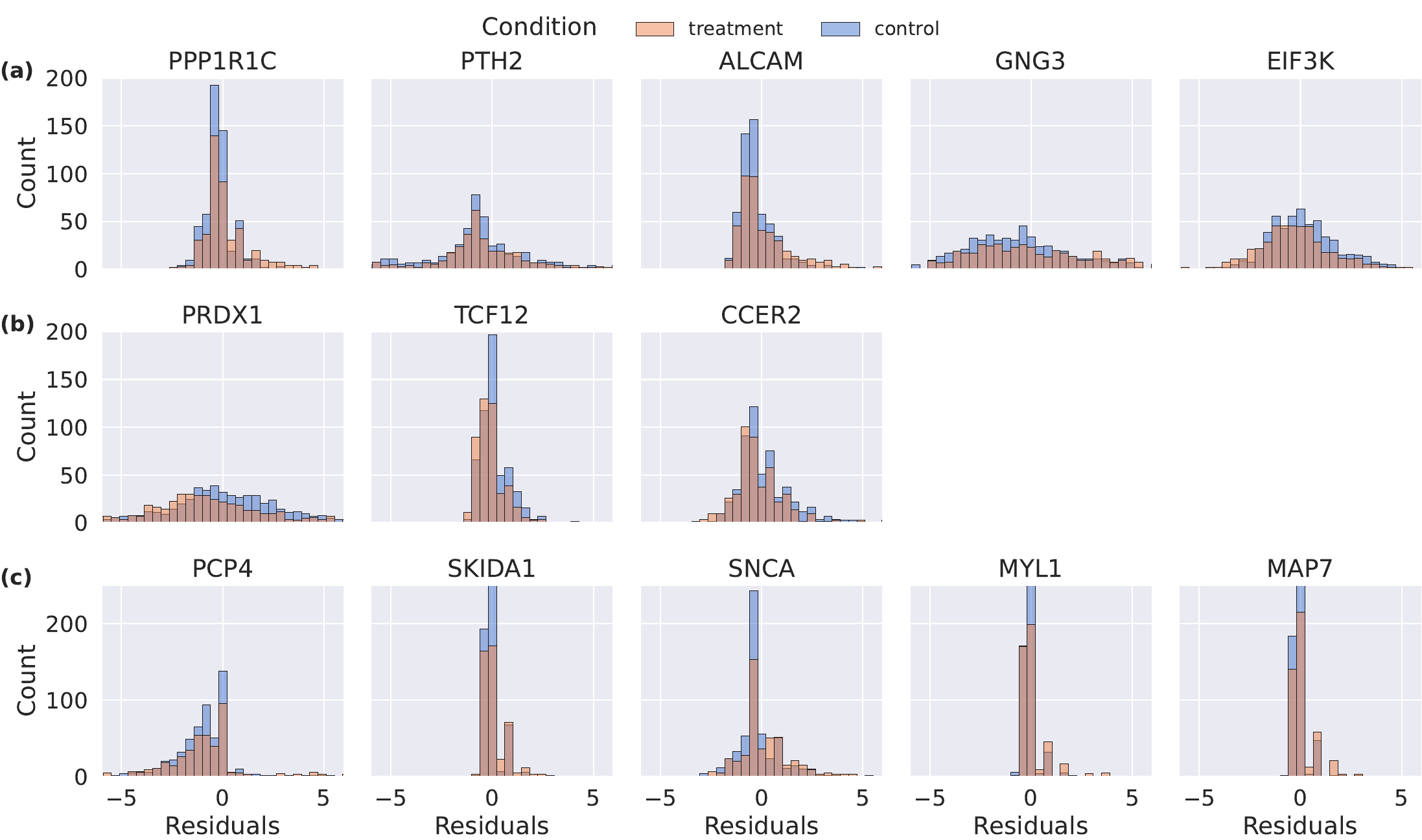}
            \caption{
            Histogram of covariate-adjusted residuals of gene expressions under different conditions (treatment with target gene \emph{PTEN} knockdown and control/non-targeting) in the late-stage samples of the LUHMES dataset.
            The residuals are obtained by fitting NB GLMs of the gene expressions on all covariates except the condition.
            Each row shows several examples of genes that are
            \textbf{(a)} significant for both ATE and STE tests;
            \textbf{(b)} significant for STE tests but insignificant for ATE tests;
            \textbf{(c)} significant for ATE tests but insignificant for STE tests.}
            \label{fig:LUHMES_PTEN_residuals}
        \end{figure}

        For perturbations targeting \emph{PTEN}, the ATE tests identify 32 genes as differentially expressed, while the STE tests identify 26 genes, and 23 genes are determined by both tests.
        The detailed results of all perturbations are summarized in \Cref{tab:LUHMES-DE-genes}.
        Next, we inspect the covariate-adjusted residuals of the gene expressions from parts of these discoveries.
        As shown in \Cref{fig:LUHMES_PTEN_residuals}, the significant genes for STE tests show a clear separation between the control and the perturbed groups, compared to those only significant for ATE tests.
        However, most of the significant genes for ATE tests exhibit zero-inflated expression in \Cref{fig:LUHMES_PTEN_residuals}(c).
        The ATE tests of these genes may be biased because of the bimodal and skew distributions.

        To further inspect the conjecture, we summarize the absolute differences of the covariate-adjusted residual mean and the mean/median of the GLM-fitted response, as well as the proportion of zero counts of the significant genes in \Cref{fig:LUHMES_PTEN_metric}.
        The common discoveries have a significant difference in the residual mean, which is expected.
        The ATE-only discoveries have a larger difference in the fitted mean, 
        although the difference in the median of the fitted responses is negligible.
        Because the proportion of zero counts for ATE-only discovered genes is also higher, it is evident that the ATE tests are biased because of the zero-inflated expressions.
        This suggests that STE tests provide more robust and reliable results than ATE tests in the presence of zero-inflated expressions.

        \begin{figure}[!t]
            \centering
            \includegraphics[width=0.9\textwidth]{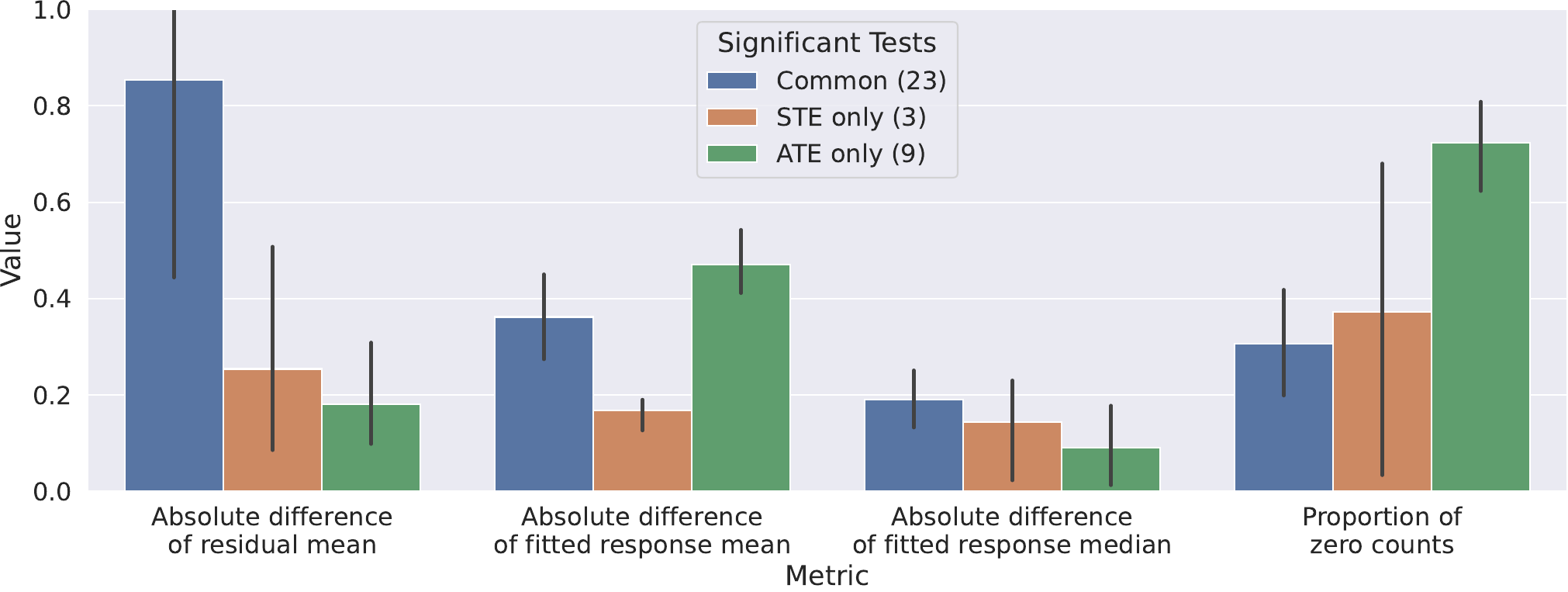}
            \caption{Summarized statistics for significant genes from ATE and STE tests under different conditions (treatment with target gene \emph{PTEN} knockdown and control/non-targeting) in the late-stage samples of the LUHMES dataset.
            The first three metrics show the absolute differences (between the treatment and the control group) of the covariate-adjusted residual mean and the mean and median of the GLM-fitted response, respectively.
            The last metric shows the proportion of zero counts of the significant genes.
            }
            \label{fig:LUHMES_PTEN_metric}
        \end{figure}

    \subsection{Lupus data}
        In this subsection, we consider a subject-level analysis of systemic lupus erythematosus (SLE) multiplexed single-cell RNA sequencing (mux-seq) dataset from \citet{perez2022single}.        
        The dataset includes 1.2 million peripheral blood mononuclear cells, spanning 8 major cell types from 261 participants, of whom 162 are SLE patients, and 99 are healthy individuals of either Asian or European descent.
        SLE is an autoimmune condition that primarily affects women and people with Asian, African, and Hispanic backgrounds.
        The objective of this cell-type-specific differential expression (DE) analysis is to enhance our understanding of SLE's diagnosis and treatment.
        We follow the preprocessing procedure in \citet{gcate2023} to aggregate the gene expressions of cells of each cell type in each subject, which serve as the derived outcomes for each cell type of each subject. 
        The resulting pseudo-bulk dataset contains five cell types T4, cM, B, T8, and NK, with numbers of subjects and genes $(n,p)$ being (256,1255), (256,1208), (254,1269), (256,1281), and (256,1178), respectively.
        For each subject, we also measure the SLE status (condition) and $d=6$ covariates: the logarithm of the library size, sex, population, and processing cohorts (4 levels).

    \begin{figure}[!t]
            \centering
            \includegraphics[width=0.9\textwidth]{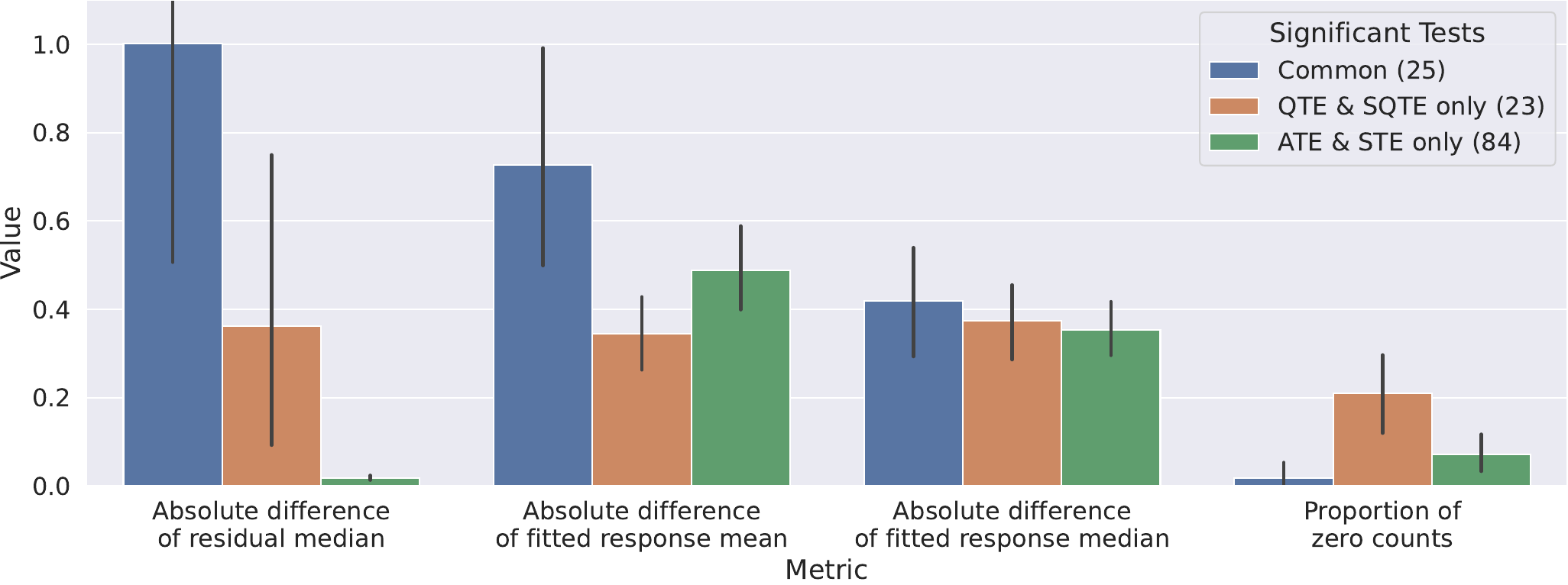}
            \caption{
            Summarized statistics for significant genes from mean-based tests (ATE and STE) and quantile-based tests (QTE and SQTE) under different conditions (case and control) in the T4 cell type of the lupus dataset.
            The first three metrics show the absolute differences (between the treatment and the control group) of the covariate-adjusted residual mean and the mean and median of the GLM-fitted response, respectively.
            The last metric shows the proportion of zero counts of the significant genes.
            }
            \label{fig:lupus-T4-metric}
        \end{figure}

        Similar to \Cref{subsec:LUHMES-data}, we aim to control the \FDPex rate over 0.1 at 0.05.
        The inference results are compared based on causal estimands. 
        Besides ATE, STE, and QTE, we also examine the standardized quantile treatment effects (SQTE), which equals the QTE normalized by the interquartile range of counterfactuals $\bY(0)$; see \Cref{app:subsubsec:SQTE} for precise definitions.               
        For the quantile-based methods, we exclude genes that have zero medians.
        The discoveries of different tests in the current pseudo-bulk analysis are summarized in \Cref{fig:lupus-upset-extra}.
        We observed that the standardized tests are generally more conservative than the unstandardized tests, as expected from our analysis in \Cref{subsec:LUHMES-data}.
        Most of the former discoveries are also parts of the latter.
        Furthermore, the ATE tests are anti-conservative, yielding dozens of extra discoveries that are insignificant for other estimands.
        Overall, the mean-based tests are much more powerful than the quantile-based tests.

        We visualize the summarized statistics in \Cref{fig:lupus-T4-metric} to further compare the mean-based and quantile-based tests.
        As expected, the quantile-based tests are more sensitive to the median differences between the two counterfactual distributions.        
        The results suggest that different causal estimands focus on different characteristics of the distributions, and one may get more reliable inferential results by focusing on the common discoveries by different tests.

\section{Discussion}

    This paper investigates semiparametric causal inference approaches for analyzing multiple derived outcomes arising from the increasingly popular study design of subject-level scRNA-seq data analysis.
    The doubly robust estimators for standardized and quantile treatment effects are proposed and analyzed to overcome the challenge of heterogeneous subjects and outcomes.
    Building on the Gaussian approximation results for the doubly robust estimators, we propose a multiple testing procedure that provably controls the false discovery rate and is asymptotically powerful.
    In simulation and real data analysis with single-cell count data, we use simple parametric and non-parametric models to estimate the nuisance functions and evaluate the multiple testing procedures.
    Notably, the proposed semiparametric inference framework allows one to incorporate more advanced machine learning and deep learning methods to obtain valid inference results.

    The present study, while contributing valuable insights, is not devoid of limitations.
    First, the doubly robust estimation can lead to bias if both nuisance functions are estimated using data-adaptive methods (e.g., machine learning) and only one is consistent \citep{van2014targeted, benkeser2017doubly}.
    Model misspecification is one of the main causes of the inconsistency, especially in the presence of high-dimensional conditional outcome models.
    The issue of misspecification itself is challenging for any statistical problem.
    In theory, if we have prior knowledge that the outcome models lie in a certain function class (e.g., the $\alpha$-H\"older functions), then our estimator allows a wide range of nonparametric methods to model nuisance functions and hence avoids model misspecification.
    In practice, the estimation of the nuisance functions requires prior knowledge.
    For genomics data, the gene expression counts are usually modeled by practitioners using generalized linear models, provided that the low-quality data are filtered beforehand \citep{sarkar2021separating}.
    For this reason, we mitigate the risk of misspecification of outcome models in the application by incorporating such prior knowledge.

    Second, we only consider STE and QTE as two specific examples, which may not be appropriate in all cases.
    For example, for lowly expressed genes of single-cell data with zero-inflation patterns, the null distribution is not unimodal, so any test that compares the location statistics is not ideal. 
    When more than half of the samples are zero for a specific gene, the mean/quantile regressions for computing the test statistics may be inaccurate, and this will affect the test results.
    To mitigate these challenges, one approach is to pre-screen and exclude genes with low expression levels, thereby reducing the impact of zero inflation. 
    Alternatively, one can also quantify treatment effects using other measures of the distribution.
    For instance, the fold change is the ratio of expected expressions under two conditions, whose sign indicates whether a gene is up-regulated or down-regulated.
    The proposed method naturally applies to these kinds of target estimands.
    Yet, the utility of different causal estimands in scientific discoveries is worth further exploration.

\section*{Acknowledgement}
    We thank Arun K. Kuchibhotla for insightful discussions of pivotal maximal inequality for finite maximums.
    This work used the Bridges-2 system at the Pittsburgh Supercomputing Center (PSC) through allocation MTH230011P from the Advanced Cyberinfrastructure Coordination Ecosystem: Services \& Support (ACCESS) program.  This project was funded by the National Institute of Mental Health (NIMH) grant R01MH123184.

% {
% \putbib[references,ref_multiple_outcome]
% }
% \end{bibunit}

% \clearpage
% \begin{bibunit}[apalike]

\appendix
\counterwithin{theorem}{section}
\setcounter{table}{0}
\renewcommand{\thetable}{\thesection\arabic{table}}
\setcounter{figure}{0} 
\renewcommand\thefigure{\thesection\arabic{figure}}
\renewcommand{\thealgorithm}{\thesection\arabic{algorithm}}
% \numberwithin{equation}{section}
% \renewcommand{\thesection}{\Alph{section}}
% \beginsupplement
\clearpage

\begin{center}
\Large
{\bf
Supplementary material
}
\end{center}

The appendix includes related work, the proof for all the theorems, computational details, and extra experiment results.
The structure of the appendix is listed below:
\bigskip

{\spacingset{1.1}
\begin{table}[!ht]
\centering
\begin{tabularx}{0.95\textwidth}{l l l }
    \toprule
    \multicolumn{2}{c}{\textbf{Appendix}} & \textbf{Content} \\
    \midrule \addlinespace[0.5ex] 
    \Cref{app:related-work} &  & Related work\\\addlinespace[0.5ex] \cmidrule(l){1-3}\addlinespace[0.5ex]
    \multirow{3}{*}{\Cref{app:sec:semi}} & \ref{app:subsec:emp} & Proof of \Cref{lem:emp-process-term} in \Cref{sec:semi-inf-multi-outcomes}\\
    & \ref{app:subsec:var} & Proof of \Cref{lem:var} in \Cref{sec:semi-inf-multi-outcomes}\\
    & \ref{app:subsec:semi-helper} & The helper \Cref{lem:maximal-ineq} used in the proof of \Cref{lem:emp-process-term,lem:var}\\\addlinespace[0.5ex] \cmidrule(l){1-3}\addlinespace[0.5ex]
    \multirow{2}{*}{\Cref{app:sec:setup}} & \ref{app:sec:iden} & Proof of \Cref{prop:identification} in \Cref{sec:setup}\\
    & \ref{app:subsec:iden-M} & Proof of \Cref{lem:identification-es} in \Cref{subsec:standarized-quantile} \\\addlinespace[0.5ex] \cmidrule(l){1-3}\addlinespace[0.5ex]
    \multirow{6}{*}{\Cref{app:sec:TE}} & \ref{app:subsec:iden-STE} & Proof of \Cref{lem:identification-STE} in \Cref{subsec:standarized-STE}\\
    & \ref{app:subsec:lin-STE} & Proof of \Cref{thm:lin-STE} in \Cref{subsec:standarized-STE}\\
    & \ref{app:subsec:Var-STE} & Proof of \Cref{prop:Var-STE} in \Cref{subsec:standarized-quantile}\\
    & \ref{app:subsec:lin-QTE} & Proof of \Cref{thm:lin-QTE} in \Cref{subsec:standarized-quantile}\\
    & \ref{app:subsec:AN-QTE} & Proof of \Cref{prop:AN-STE} in \Cref{subsec:standarized-quantile}\\
    & \ref{app:subsec:TE-helper} & Helper lemmas: \Cref{lem:ATE,lem:err-omega}\\\addlinespace[0.5ex] \cmidrule(l){1-3}\addlinespace[0.5ex]
    \multirow{4}{*}{\Cref{app:sec:inference}} & \ref{subsec:max-Gaussian} & Proof of \Cref{prop:max-Gaussian} \\
     & \ref{app:subsec:simul-inference} & Proof of \Cref{prop:simul-inference} \\
     & \ref{app:subsec:FDPex} & Proof of \Cref{thm:FDPex}\\
     & \ref{app:subsec:multiple-test-helper} & Helper lemmas: \Cref{lem:max-Gaussian-quantile} for quantile approximation\\
     \addlinespace[0.5ex] \cmidrule(l){1-3}\addlinespace[0.5ex] 
    \multirow{2}{*}{\Cref{app:sec:experiment}} & \ref{app:subsec:est-QTE} & Initial estimators for QTE. \\
     & \ref{app:subsec:extra-experiment} & Extra experimental results. \\
    \addlinespace[0.5ex] \arrayrulecolor{black}
    \bottomrule
\end{tabularx}
\vspace*{-1.7\baselineskip}
\end{table}
}
\addcontentsline{toc}{part}{\appendixname}

\clearpage
\section{Related work}\label{app:related-work}
In the context of causal inference, assessing causal effects on multiple outcomes requires accounting for the association among outcomes \citep{sammel1999multivariate}.
Earlier work on this problem relies on outcome modeling approaches based on linear mixed models or latent variable models
\citep{thurston2009bayesian,teixeira2011statistical,lupparelli2020joint,gcate2023} and scaled linear models \citep{lin2000scaled,roy2003scaled}.
\citet{pocock1987analysis,yoon2011alternative} study hypothesis testing of the treatment effects with the adjustment for multiplicities.
\citet{mattei2013exploiting,mealli2013using,mercatanti2015improving,mealli2016identification} use multiple outcomes, coupled with conditional independence assumptions, to address identification problems in causal studies with intermediate/post-treatment variables.
In most of the aforementioned work, the number of outcomes is typically assumed to be low-dimensional.

The focus on multiple outcomes also shifts from outcome modeling to more general setups.
\citet{flanders2015general} propose a general definition of causal effects, showing how it can be applied in the presence of multivariate outcomes for specific sub-populations of units or vectors of causal effects.
For randomized experiments, \citet{li2017general} establish finite population central limit theorems in completely randomized experiments where the response variable may be multivariate and causal estimands of interest are defined as linear combinations of the potential outcomes.
For observational studies with derived outcomes, 
Recently, \citet{qiu2023unveiling} propose an inverse probability weighting estimator for testing multiple average treatment effects (ATEs), which relies on the correct specification of the propensity score model and fast convergence rate of the propensity score estimation.

Although ATE is the most fundamental and popular causal estimand \citep{imbens2004nonparametric,tsiatis2006semiparametric}, other estimands could be more robust to quantify the treatment effect between the counterfactual distributions.
In the canonical setting with a single outcome, \citet{athey2021semiparametric} explore the application of quantile methodologies to estimate the overall treatments within the context of randomized trials with heavy-tailed outcome distributions. 
For studies based on observational data, \citet{belloni2017program} and \citet{kallus2019localized} introduce localized debiased machine learning techniques for quantile treatment effects (QTEs), which incorporates multiple sample partitioning. 
Concurrently, \citet{chakrabortty2022general} study the estimation of QTEs within a semi-supervised framework. 
For a comprehensive overview of QTE estimation literature, the reader is directed to the references contained therein.
The above methods require sample splitting and/or metric entropy conditions to validate the asymptotic normality of the proposed estimators, even for a single outcome.

For the analysis of multiple outcomes extending beyond ATEs, \citet{kennedy2019estimating} propose DR estimators designed to evaluate both scaled average treatment effects and scaled quantile effects. 
These estimands are used to rigorously test the global hypothesis that all treatment effects are equal. 
However, the asymptotic properties of the quantile-based estimators are not analyzed.
Further, they only consider a low-dimensional set of outcomes.
In contrast, the current paper focuses on high-dimensional settings when the number of outcomes could be potentially exponentially larger than the sample size.
This also requires correctly addressing the issue of multiple hypothesis testing.

\section{Proof in \Cref{sec:semi-inf-multi-outcomes}}\label{app:sec:semi}
    \subsection{Proof of \Cref{lem:emp-process-term}}\label{app:subsec:emp}

    \begin{proof}[Proof of \Cref{lem:emp-process-term}]        
        Denote the scaled empirical process $\sqrt{n}(\PP_n-\PP)$ by $\GG_n$.
        By \Cref{lem:maximal-ineq}, it follows that
        \begin{align*}
            \EE\left[\max_{j=1,\ldots,p}\ |\GG_ng_j | \;\middle|\; g_1, \dots g_p \right] &\lesssim \sqrt{\log p} \max_{1\leq j\leq p} \|g_j\|_{L_2} + \frac{(\log p)^{1-1/q}}{n^{1/2 - 1/q}} \|G\|_{L_q}.
        \end{align*}
        Dividing $\sqrt{n}$ on both sides finishes the proof.
    \end{proof}

    \begin{remark}
        We note that \Cref{lem:emp-process-term} also provides a probabilistic bound:    
        \begin{align*}
            \max_{j=1,\ldots,p}\ |(\PP_n-\PP)g_j | = \Op\left( \left(\frac{\log p}{n}\right)^{1/2}\max_{1\leq j\leq p} \|g_j\|_{L_2} + \left(\frac{\log p}{n}\right)^{1 - 1/q} \|G\|_{L_q} \right),
        \end{align*}
        by applying Markov inequality on the non-negative random variable $\max_{j=1,\ldots,p}\ |(\PP_n-\PP)g_j |$.
    \end{remark}

    \begin{remark}[Donsker condition]\label{rmk:donsker}
    Without an independent sample, similar bounds on the empirical process term can still be derived, provided certain complexity measures of the function class $\cF_j$ that $g_j$ belongs to are properly bounded. 
    The complexity measure is related to the covering number of $\cF_j$ under the $L_2$ norm induced by distribution $Q$, denoted by $N\left(\varepsilon, \cF_j, L_2(Q)\right) $.
    In particular, if for some envelope function $F$,
    \[\int_0^{\infty} \sup_Q\sqrt{\log N\left(\varepsilon\|F\|_{Q,2}, \cF_j, L_2(Q)\right)} \rd \varepsilon < \infty,\]
    and the supreme is taken over all probability distributions $Q$ on $\cZ$, then $\cF_j$ is $P$-Donsker for every probability measure $P$ such that $P^*F^2<\infty$ under certain measurability conditions; see \citet[Theorem 2.5.2]{vaart1996weak} for the exact conditions.
    To control single empirical process terms, a sufficient condition is that each $\cF_j$ has a polynomial covering number:
    \[
        \sup_Q \; N\left(\varepsilon,\, \cF_j,\, L_2(Q)\right) \leq\left(\frac{K}{\varepsilon}\right)^V,\qquad \forall\ 0<\epsilon<1,
    \]
    where $K,V>0$ are constants, which is similar to \citet[Assumption 3.4]{chakrabortty2022general}. 
    If further the functions $g_j = \varphi_j(\bZ;\hat{\PP}) - \varphi_j(\bZ;\PP)$ are uniformly bounded, by using a Bernstein-type tail bound on the empirical process (e.g., \citet[Theorem 2.14.9]{vaart1996weak}) and a union bound argument, we can obtain a similar upper bound on $\max_{j=1,\ldots,p}\ |(\PP_n-\PP)g_j |$ as in Lemma 1, with $\max_{1\leq j\leq p} \|g_j\|_{L_2}$ replaced by $\max_{1\leq j\leq p} \sup_{g \in \mathcal{F}_j} \|g\|_{L_2}$. 
    When the function classes $\cF_j=\cF$ for all $j$ are the same, this can be further improved to $\sup_{g \in \mathcal{F}} \|g\|_{L_2}$.
    However, verifying such assumptions on metric entropy in practice is challenging, and training the nuisance functions on an independent sample avoids this issue. Hence, in this paper, we use sample splitting instead of Donsker-type conditions.

    If the Donsker class condition is assumed, \Cref{thm:lin-STE,thm:lin-QTE} can be established without an independent sample.
    More specifically, the bias term $T_{{\rm R},j}$ in the decomposition \eqref{eq:decomposition} can be bounded under the same rate conditions on the product of two nuisance estimations in \Cref{thm:lin-STE,thm:lin-QTE}.
    This does not rely on independent sample splitting.
    Sample splitting is mainly used to control the empirical process term $T_{{\rm E},j}$ via \Cref{lem:emp-process-term} for the proof of \Cref{thm:lin-STE,thm:lin-QTE}, as described in Section 2.
    \end{remark}

    \subsection{Proof of \Cref{lem:var}}\label{app:subsec:var}
    
    \begin{proof}[Proof of \Cref{lem:var}]          Below, we condition on the event when conditions (1) and (2) hold.
        We begin by decomposing the error into two terms:
        \begin{align*}
            \hat{\sigma}_j^2-\sigma_j^2   &=  \VV_n(\hat{\varphi}_j) - \Var[\varphi_j]\\
            &=  \PP_n[(\hat{\varphi}_j - \PP_n[\hat{\varphi}_j])^2] - \PP[(\varphi_j - \PP[\varphi_j])^2] \\
            &=  \PP_n[\hat{\varphi}_j^2] - \left( \PP_n [\hat{\varphi}_j] \right)^2 - \PP[\varphi_j ^2] +\left( \PP[\varphi_j] \right)^2\\
            &= \PP_n [\hat{\varphi}_j^2 - \varphi_j^2]+ (\PP_n - \PP) [\varphi_j^2] - \left( \PP_n [\hat{\varphi}_j] - \PP[\varphi_j]\right) \left(\PP_n [\hat{\varphi}_j ] + \PP[\varphi_j] \right)\\
            &= T_{1j} + T_{2j} + T_{3j}.
        \end{align*}
        For the first term, we have
        \begin{align*}
            \max_{1\leq j\leq p}|T_{1j}| = &\,  \max_j \left|\PP_n [(\hat{\varphi}_j - \varphi_j)(\hat{\varphi}_j + \varphi_j)] \right| \lesssim  \max_j \left|\PP_n [|\hat{\varphi}_j - \varphi_j|] \right|,
        \end{align*}
        where we use the boundedness of $(\hat{\varphi}_j+\varphi_j)$'s envelope.
        
        For the third term, we have
        \begin{align*}
            \max_{1\leq j\leq p}|T_{3j}| \lesssim &\,  \max_j \left|( \PP_n-\PP) [\varphi_j] \right| +  \max_j \left|( \PP_n-\PP) [\hat{\varphi}_j - \varphi_j] \right| + \max_j |\PP[\hat{\varphi}_j - \hat{\varphi}_j]|.
        \end{align*}
        Combining the above results yields that
        \begin{align*}
            \max_{1\leq j\leq p}|\hat{\sigma}_j^2-\sigma_j^2  | \lesssim &\,  \max_j \left|( \PP_n-\PP) [ \hat{\varphi}_j - \varphi_j ] \right| + \max_j \left|( \PP_n-\PP) [ |\hat{\varphi}_j - \varphi_j| ] \right|\\
            &\,+ \max_{k=1,2} \max_j \left|( \PP_n-\PP) [\varphi_j^k] \right| + \max_j \PP[|\hat{\varphi}_j - \varphi_j|] ,
        \end{align*}    
        with probability tending to one.

        Define $\Psi = \max_{1\leq j\leq p} |\hat{\varphi}_j - \varphi_j|$ and $\Phi = \max_{1\leq j\leq p} |\varphi_j|$.
        By \Cref{lem:emp-process-term}, it follows that
        \begin{align*}
            \max_{1\leq j\leq p}|\hat{\sigma}_j^2-\sigma_j^2  | \lesssim &\,  \left(\frac{\log p}{n}\right)^{1/2}\max_j \|\hat{\varphi}_j - \varphi_j\|_{L_2} + \left(\frac{\log p}{n}\right)^{1-1/q} \|\Psi\|_{L_q} \\
            &\qquad + \left(\frac{\log p}{n}\right)^{1/2} \max_{k=1,2} \max_j \|\varphi_j^k\|_{L_2} + \left(\frac{\log p}{n}\right)^{1-1/q}\max_{k=1,2}\|\Phi^k\|_{L_q} \\
            &\qquad + \max_j \|\hat{\varphi}_j - \varphi_j\|_{L_1}\\
            &= \left(\frac{\log p}{n}\right)^{1/2}\max_j \left(\|\hat{\varphi}_j - \varphi_j\|_{L_2} + \max_{k=1,2} \|\varphi_j^k\|_{L_2} \right) + \max_j \|\hat{\varphi}_j - \varphi_j\|_{L_1} \\
            &\qquad + \left(\frac{\log p}{n}\right)^{1-1/q} \left(\|\Psi\|_{L_q} + \max_{k=1,2}\|\Phi^k\|_{L_q}\right),
        \end{align*}
        which holds with probability at least $1-n^{-c}$.
    \end{proof}

    \subsection{Helper lemmas}\label{app:subsec:semi-helper}    
    \begin{lemma}[Maximal inequality adapted from Proposition B.1 of \citet{kuchibhotla2022least}]\label{lem:maximal-ineq}
        Let $\bX_1, \ldots, \bX_n$ be mean zero independent random variables in $\RR^p$ for $p \geq 1$. Suppose there exists $q\in\NN$ such that for all $i \in\{1, \ldots, n\}$
        $$
        \EE\left[\xi_i^q\right]<\infty \quad \text { where } \quad \xi_i:=\max _{1 \leq j \leq p}\left|X_{i, j}\right|
        $$
        and $\bX_i:=\left(X_{i, 1}, \ldots, X_{i, p}\right)^{\top}$. If $V_{n, p}:=\max _{1 \leq j \leq p} \sum_{i=1}^n \mathbb{E}\left[X_{i, j}^2\right]$, then
        $$
        \EE\left[\max _{1 \leq j \leq p}\left|\sum_{i=1}^n X_{i, j}\right|\right] \leq \sqrt{6 V_{n, p} \log (1+p)}+\sqrt{2}(3 \log (1+p))^{1-1 / q}\left(2 \sum_{i=1}^n \EE\left[\xi_i^q\right]\right)^{1 / q}.
        $$
    \end{lemma}

% \clearpage
\section{Identification conditions}\label{app:sec:setup}
    \subsection{Proof of \Cref{prop:identification}}\label{app:sec:iden}
    \begin{proof}[Proof of \Cref{prop:identification}]
        Note that
        \begin{align*}
            &  \EE(\tilde{\bY} \mid A=a, \bW) \\
            =& \EE(\tilde{\bY}(a) \mid A=a, \bW) \\
            =& \EE(\tilde{\bY}(a) \mid \bW) \\
            =& \EE[\EE[\tilde{\bY}(a) \mid \bW,\bS(a)]\mid \bW]\\
            =& \EE[\bY(a)\mid \bW] + \EE[\bDelta_m(a) \mid \bW] .
        \end{align*}
        where the first equality is from the consistency and positivity assumption, and the second equality is from the unmeasured confounder assumption.
        From the asymptotic unbiasedness as in \Cref{def:unbias}, it follows that
        \begin{align}
            \EE[\EE[\tilde{\bY} \mid A=a, \bW]] =  \EE[\bY(a)] + \EE[\bDelta_m(a) ] = \EE[\bY(a)] + \bo(1) ,\label{eq:prop:identification-eq-1}
        \end{align}
        where the little-o notation is with respect to $m$ and uniform in $p$.
        This completes the proof.
    \end{proof}

    \subsection{Proof of \Cref{lem:identification-es}}\label{app:subsec:iden-M}
    
    \begin{proof}[Proof of \Cref{lem:identification-es}]
        Define $\tilde{M}_j(\theta) = \EE[F_j(\tilde{Y}_j(a), \theta)]$.
        From the assumption, we have that
        \begin{align*}
            \max_{j\in[p]} |M_j(\theta) - \tilde{M}_j(\theta)| = \max_{j\in[p]} |\EE[\Delta_{mj}(a, \theta)]|.
        \end{align*}
        By Taylor expansion, we have
        \begin{align*}
            M_j(\tilde{\theta}_{aj}) - M_j(\theta_{aj}) &= M_j'(\bar{\theta}_{aj})(\tilde{\theta}_{aj} - \theta_{aj}).
        \end{align*}
        for some $\bar{\theta}_{aj}$ between $\tilde{\theta}_{aj}$ and $\theta_{aj}$.
        These two results imply that
        \begin{align*}
            \max_{j\in[p]}|\tilde{\theta}_{aj} - \theta_{aj}| &= \max_{j\in[p]}\frac{1}{|M_j'(\bar{\theta}_{aj})|} |M_j(\tilde{\theta}_{aj}) - M_j(\theta_{aj})| \\
            &\leq \frac{1}{c} \max_{j\in[p]}|M_j(\tilde{\theta}_{aj})|\\
            &=\frac{1}{c} \max_{j\in[p]}|\tilde{M}_j(\tilde{\theta}_{aj}) + {M}_j(\tilde{\theta}_{aj})-\tilde{M}_j(\tilde{\theta}_{aj})|\\
            & = \frac{1}{c} \max_{j\in[p]}|{M}_j(\tilde{\theta}_{aj})-\tilde{M}_j(\tilde{\theta}_{aj})|
        \end{align*}
        Similar to the proof of \Cref{prop:identification}, we have that
        \begin{align*}
            \tilde{M}_j({\theta}) &= \EE[\EE[F_j(\tilde{Y}_j,{\theta}) \mid A=a,\bW]] = M_{j}(\theta) + \EE[\Delta_{mj}(a,\theta)] .
        \end{align*}
        Thus we have
        \[
        \max_{j\in[p]}|\tilde{\theta}_{aj} - \theta_{aj}| \leq \frac{1}{c} \max_{j\in[p]}|\EE[\Delta_{mj}(a,\tilde{\theta}_{aj})]| \leq \frac{1}{c} \max_{j\in[p]} \sup_{\theta \in \cB(\theta_{aj}, \delta)}|\EE[\Delta_{mj}(a,{\theta})]| \lesssim \delta_m \rightarrow 0
        \]
        as $m \rightarrow \infty$.
    \end{proof}

\section{Doubly robust estimation}\label{app:sec:TE}
\subsection{Proof of \Cref{lem:identification-STE}}\label{app:subsec:iden-STE}

    \begin{proof}[Proof of \Cref{lem:identification-STE}]
        For all $i=1,2$, and $j=1,\ldots,p$, define
        \begin{align*}
            \psi_{aij} &= \EE[Y_j(a)^i],\qquad \tilde{\psi}_{aij} = \EE[\EE[\tilde{Y}_{j}^i \mid A=a, \bW]].
        \end{align*}
        From \Cref{prop:identification} we have that $\psi_{aij} = \tilde{\psi}_{aij} + o(1)$ as $m\rightarrow\infty$.
        It follows that
        \begin{align*}
            \tilde{\tau}_j &= \frac{\EE[\EE(\tilde{Y}_j \mid A=1, \bW)] - \EE[\EE(\tilde{Y}_j \mid A=0, \bW)]}{\sqrt{ \EE[\EE(\tilde{Y}_j^2 \mid A=0, \bW)] - \EE[\EE(\tilde{Y}_j \mid A=0, \bW)]^2 }} \\
            &= \frac{\tilde{\psi}_{11j} - \tilde{\psi}_{01j}}{\sqrt{ \tilde{\psi}_{02j} - \tilde{\psi}_{01j}^2 }} \\
            &= \frac{\psi_{11j} - \psi_{01j} + o(1)}{\sqrt{ \psi_{02j} - \psi_{01j}^2 + o(1)} }\\
            &= \frac{\psi_{11j} - \psi_{01j} }{\sqrt{ \psi_{02j} - \psi_{01j}^2 } } + o(1) \\
            &= \tau_j + o(1),
        \end{align*}
        where the second last equality holds because $\Var[Y_{j}(0)]=\psi_{02j} - \psi_{01j}^2>0$ as assumed.
    \end{proof}

\subsection{Proof of \Cref{thm:lin-STE}}\label{app:subsec:lin-STE}
    \begin{proof}[Proof of \Cref{thm:lin-STE}]
        In this proof we will abbreviate $\tau_j^{\text{STE}} $ as $\tau_j$ and denote 
        \[
        \begin{aligned}
            \tilde{\tau}_j &= \frac{\EE[\EE(\tilde{Y}_j \mid A=1, \bW)] - \EE[\EE(\tilde{Y}_j \mid A=0, \bW)]}{\sqrt{ \EE[\EE(\tilde{Y}_j^2 \mid A=0, \bW)] - \EE[\EE(\tilde{Y}_j \mid A=0, \bW)]^2 }}.
        \end{aligned}
        \]
        We split the proof into three parts, conditioned on the event when the assumptions hold (which holds with probability tending to one).

        \paragraph{Part (1) Individual ATE estimators.}
        Let $\boldsymbol{\beta}_j=\left(\beta_{0 j}, \beta_{1 j}, \beta_{2 j}\right)^{\top}$ with
        $$
        \beta_{0 j}=\mathbb{E} [Y_j(0) ],\;\beta_{1 j}=\mathbb{E} [Y_j(1) ],\;\beta_{2 j}=\mathbb{E} [Y_j(0)^2 ],
        $$
        Let $\tilde{\boldsymbol{\beta}}_j=\left(\tilde{\beta}_{0 j}, \tilde{\beta}_{1 j}, \tilde{\beta}_{2 j}\right)^{\top}$ with
        $$
        \tilde{\beta}_{0 j}=\EE[\mathbb{E} (\tilde{Y}_j\mid A=0, \bW )],\;\tilde{\beta}_{1 j}=\EE[\mathbb{E} (\tilde{Y}_j \mid A=1, \bW )],\;\tilde{\beta}_{2 j}=\EE[\mathbb{E} (\tilde{Y}_j^2 \mid A=0, \bW )],
        $$
        and define the corresponding estimator $\hat{\boldsymbol{\beta}}_j=(\hat{\beta}_{0 j}, \hat{\beta}_{1 j}, \hat{\beta}_{2 j})^{\top}$ for
        \[
        \hat{\beta}_{0j}=\mathbb{P}_n\left\{\tilde\phi_{01 j  }(\bZ ; \hat{\pi}, \hat{\bmu})\right\}, \qquad \hat{\beta}_{1 j}=\mathbb{P}_n\left\{\tilde\phi_{11 j }(\bZ ; \hat{\pi}, \hat{\bmu})\right\}, \qquad\hat{\beta}_{2 j}=\mathbb{P}_n\left\{\tilde\phi_{02 j }(\bZ ; \hat{\pi}, \hat{\bmu})\right\} .
        \]
        Let $\bphi_j=(\tilde\phi_{01 j }, \tilde\phi_{11 j }, \tilde\phi_{02 j })^{\top}$ and
        \[\delta_{m0} = \max_{1\leq j\leq p}|\EE[\Delta_{m1j}(0)]|,\qquad \delta_{m1} = \max_{1\leq j\leq p}|\EE[\Delta_{m1j}(1)]|,\qquad \delta_{m2} = \max_{1\leq j\leq p}|\EE[\Delta_{m2j}(0)]|.\]
        Similar to \Cref{lem:ATE}, we can show that
        $$
        \hat{\bbeta}_j-\tilde{\bbeta}_j=(\mathbb{P}_n-\PP)\left[\bphi_j(\bZ ; {\pi}, {\bmu})\right]+ \bepsilon_j
        $$
        where $\max_{1\leq j\leq p}\|\bepsilon_j\|_{\infty} =\cO( r_{np})$ and
        \[r_{np} = (\log p)^{1/2}n^{-(1/2+\alpha \wedge \beta)} + \log p /n+ n^{-(\alpha+\beta)} .\]

        \paragraph{Part (2) Bounding $\hat{\tau}_j-\tilde{\tau}_j$.}
        Now since $\hat{\tau}_j=h(\hat{\boldsymbol{\beta}}_j):=(\hat{\beta}_{1 }-\hat{\beta}_{0 } )(\hat{\beta}_{2 }-\hat{\beta}_{0 }^2)^{-1/2}$, an application of Taylor's expansion yields
        \begin{align}            
            & \hat{\tau}_j-\tilde{\tau}_j \notag\\
            & =(\mathbb{P}_n - \PP)\left[ \nabla h(\tilde{\bbeta}_j)^{\top}\bphi_j(\bZ ; {\pi}, {\bmu})\right]+ \nabla h(\tilde{\bbeta}_j)^{\top}  \bepsilon_j + \Op(\|\hat{\bbeta}_j-\tilde{\bbeta}_j\|_2^2) \notag\\
            & =\mathbb{P}_n\left[\frac{\tilde\phi_{11 j }-\tilde\phi_{01 j }}{\sqrt{\tilde{\beta}_{2 j}-\tilde{\beta}_{0 j}^2}}-\tilde{\tau}_j\left\{\frac{\tilde\phi_{02 j }+\tilde{\beta}_{2 j}-2 \tilde{\beta}_{0 j} \tilde\phi_{01 j }}{2\left(\tilde{\beta}_{2 j}-\tilde{\beta}_{0 j}^2\right)}\right\}\right] + \nabla h(\tilde{\bbeta}_j)^{\top}  \bepsilon_j + \Op(\|\hat{\bbeta}_j-\tilde{\bbeta}_j\|_2^2).\label{eq:sate-temp}
        \end{align}
        We denote the first term by $\PP_n\{\tilde\varphi_{j}\}$.
        For remainder term $\nabla h(\tilde{\bbeta}_j)^{\top}  \bepsilon_j$, by proof of Proposition \ref{prop:identification} we have 
        \[
        \max_{1 \leq j \leq p}\|\tilde{\bbeta}_j-\bbeta_j\|_2 \leq \sqrt{3} \max_{1 \leq j \leq p}\|\tilde{\bbeta}_j-\bbeta_j\|_{\infty} = \sqrt{3} \delta_m 
        \]
        Thus, from the bounded variance assumption that $\Var[Y_j(0)]=\beta_{j2}-\beta_{j0}^2\geq c$ and 
        \[
        \max_{1\leq j\leq p}|\EE[\Delta_{mkj}(a)]|= \delta_m \rightarrow 0,
        \] 
        when $m$ is large, $\max_{1 \leq j \leq p }\|\nabla h(\tilde{\bbeta}_j)\|_2 \leq C$ for some constant $C$. We have
        \begin{align*}
            |\nabla h(\tilde{\bbeta}_j)^{\top}  \bepsilon_j| &\leq \|\nabla h(\tilde{\bbeta}_j)\|_2 \|\bepsilon_j\|_2 \leq \sqrt{3} C\|\bepsilon_j\|_{\infty},    \\
            \max_{1 \leq j \leq p} |\nabla h(\tilde{\bbeta}_j)^{\top}  \bepsilon_j| &\lesssim \EE\left[\max_{1 \leq j \leq p} \|\bepsilon_j\|_{\infty}\right] = \cO( r_{np}).
        \end{align*}
        The second-order remainder is bounded as 
        \[
        \begin{aligned}
            \|\hat{\bbeta}_j-\tilde{\bbeta}_j\|_2^2 
            \lesssim & \|(\mathbb{P}_n -\PP)\left[\bphi_j(\bZ ; {\pi}, {\bmu}) \right ]\|_2^2 + \|\bepsilon_j\|_{\infty}^2 ,
        \end{aligned}
        \]
        where we have
        \[
        \begin{aligned}
            &\,\|(\mathbb{P}_n -\PP)\left[\bphi_j(\bZ ; {\pi}, {\bmu}) \right ]\|_2^2 \\
            =&\,|(\mathbb{P}_n -\PP)\left[\tilde\phi_{01j}(\bZ ; {\pi}, {\bmu}) \right ]|^2 + |(\mathbb{P}_n -\PP)\left[\tilde\phi_{11j}(\bZ ; {\pi}, {\bmu}) \right ]|^2+|(\mathbb{P}_n -\PP)\left[\tilde\phi_{02j}(\bZ ; {\pi}, {\bmu}) \right ]|^2
        \end{aligned}
        \]
        Since $\tilde\phi_{01j}, \tilde\phi_{11j}, \tilde\phi_{02j}$ are all bounded, by \Cref{lem:emp-process-term} with $q = \infty$ we have
        \[
            \max_{1 \leq j \leq p}|(\mathbb{P}_n -\PP)\left[\tilde\phi_{01j}(\bZ ; {\pi}, {\bmu}) \right ]|^2 = \cO\left(\frac{\log p}{n}\right) 
        \]
        and same bound holds for $\tilde\phi_{11j}, \tilde\phi_{02j}$ as well. This together with $\EE[\max_{1\leq j\leq p}\|\bepsilon_j\|_{\infty}] =\cO(r_{np})$ implies (note that $r_{np}$ includes the $\log p/n$ term)
        \begin{align}
            \max_{1 \leq j \leq p}\|\hat{\bbeta}_j-\tilde{\bbeta}_j\|_2^2 =\cO\left(\frac{\log p}{n}+r_{np}^2\right)=\cO(r_{np}). \label{eq:thm:AN-STE-remainder}
        \end{align}
        Combining the above bounds with \eqref{eq:sate-temp} implies that
        \begin{align}
            \hat{\tau}_j-\tilde{\tau}_j &= \PP_n\{\tilde\varphi_{j}\} + \varepsilon_j' \label{eq:diff-tau-ttau}
        \end{align}
        where the residual terms satisfy $\max_{j\in[p]}|\varepsilon_j'| = \cO((\log p)^{1/2}n^{-(1/2+\alpha \wedge \beta)} + \log p /n + n^{-(\alpha+\beta)})$.

        \paragraph{Part (3) Bounding $\tilde{\tau}_j -\tau_j$.}
        By Taylor expansion, we also have
        \[
        \tilde{\tau}_j -\tau_j = h(\tilde{\bbeta}_j) - h(\bbeta_j) = \nabla h(\bar{\bbeta}_j) ^{\top} (\tilde{\bbeta}_j-\bbeta_j)
        \]
        where $\bar{\bbeta}_j$ lies between $\tilde{\bbeta}_j$ and $h(\bbeta_j)$ so it also lies in the ball $\cB_j(\bbeta_j,\delta_m)$ where $\delta_m=\max\{\delta_{m0},\delta_{m1},\delta_{m2}\}$ is the maximum bias of derived outcomes.
        When $m$ is large, we also have $\|\nabla h(\bar{\bbeta}_j)\|_2 \leq C$. This implies
        \[
            \max_{1 \leq j \leq p}|\tilde{\tau}_j -\tau_j| \lesssim \|\tilde{\bbeta}_j -\bbeta_j\|_2 \lesssim \delta_m. 
        \]
        Combining the above bounds with \eqref{eq:diff-tau-ttau}, we have
        \[
            \hat{\tau}_j-{\tau}_j = \PP_n\{\tilde\varphi_{j}\}  + \varepsilon_j,
        \]
        where the residual terms satisfy $\max_{j\in[p]}|\varepsilon_j| = \cO((\log p)^{1/2}n^{-(1/2+\alpha \wedge \beta)} + \log p /n + n^{-(\alpha+\beta)} + \delta_m)$.
    \end{proof}

    \subsection{Proof of \Cref{prop:Var-STE}}\label{app:subsec:Var-STE}
    \begin{proof}[Proof of \Cref{prop:Var-STE}]

        We next verify that the conditions in \Cref{lem:var} hold the event when the assumptions hold.
        Recall the centered influence function $\tilde\varphi_j=\tilde\varphi_j^\STE$ defined in \Cref{thm:lin-STE}.
        We write $\tilde\varphi_{j}(\bZ; \PP) = \tilde{\varphi}_{j}(\bZ; \pi,\bmu) $ and $\tilde\varphi_{j}(\bZ; \hat{\PP}) = \tilde{\varphi}_{j}(\bZ; \hat{\pi},\hat{\bmu}) $.
        By the boundedness assumption, we have that 
        \[\max_{j\in[p]}|\tilde\varphi_j(\bZ; \PP) + \tilde\varphi_j(\bZ; \hat{\PP})|=\cO(1).\]
        From the proof of \Cref{lem:ATE}, the individual influence functions satisfy that
        \begin{align*}
            \max_{1\leq j\leq p}  \|\tilde{\phi}_{akj}(\bZ; \hat{\PP}) - \tilde{\phi}_{akj}(\bZ; {\PP})\|_{L_2} &= \cO( n^{-\alpha\wedge\beta}),\qquad a=0,1,\ k=0,1,2,
        \end{align*}
        The estimation error  $\|\tilde\varphi_{j}(\bZ; \hat{\PP}) - \tilde\varphi_{j}(\bZ; {\PP})\|_2$ then depends on the slowest rate among $\|\hat{\bbeta}_j -\tilde{\bbeta}_j\|_2, \hat{\tau}_j-\tilde{\tau}_j$ and $\tilde{\phi}_{akj}(\bZ; \hat{\PP}) - \tilde{\phi}_{akj}(\bZ; {\PP})$. By the proof in Appendix \ref{app:subsec:lin-STE}, we have
        \begin{align*}
            \max_{1 \leq j \leq p}\|\hat{\bbeta}_j-\tilde{\bbeta}_j\|_2  &= \cO\left(\sqrt{\frac{\log p}{n}}+r_{np}\right), \\
            \max_{1 \leq j \leq p} |\hat{\tau}_j - \tilde{\tau}_j| &= \cO\left(\sqrt{\frac{\log p}{n}}+r_{np}\right).
        \end{align*}
        Combining this with the positivity assumption of $\VV[\tilde{Y}_j(0)]$ implies that
        \begin{align}
            \max_{1\leq j\leq p}  \|\tilde\varphi_{j}(\bZ; \hat{\PP}) - \tilde\varphi_{j}(\bZ; {\PP})\|_{L_2} &= \cO\left( \sqrt{\frac{\log p}{n}}+ n^{-\alpha\wedge\beta}\right). \label{eq:expect-diff-influence-est}
        \end{align}
        
        Therefore, by applying \Cref{lem:var}, we have that $\max_{j\in[p]}|\hat{\sigma}_j^2-\sigma_j^2| = \Op(\sqrt{\log p/n}+n^{-\alpha\wedge\beta})$ as $m,n,p\rightarrow\infty$ such that $\log p =o(n^{\min(\frac{1}{2}, 2\alpha, 2\beta)})$.
    \end{proof}

\subsection{Proof of \Cref{thm:lin-QTE}}\label{app:subsec:lin-QTE}
\begin{proof}[Proof of \Cref{thm:lin-QTE}]
    We condition on the event when the assumptions hold.
    In the proof of Proposition \ref{lem:identification-es}, we show
    \[\max_{j\in[p]}|\tilde{\theta}_{aj} - \theta_{aj}| \leq \frac{1}{c} \max_{j\in[p]} \sup_{\theta \in \cB(\tilde{\theta}_{aj}, \delta)}|\EE[\Delta_{mj}(a,{\theta})]|  \lesssim \delta_m.
    \]
    We first inspect the estimation error of counterfactual quantiles for $a\in\{0,1\}$. 
    Note that Condition \ref{asm:quantile-init} implies that
    \begin{align*}
        \PP(\cap_{j=1}^p \{\hat{\theta}_{aj}^{\init}\in \cB(\tilde{\theta}_{aj},\delta)\}) &\rightarrow 1.
    \end{align*}
    Also, Conditions \ref{asm:quantile-boundedness} and \ref{asm:quantile-init} imply that
    \begin{align}
        \max_{j\in[p]} \hf_{aj}(\hat{\theta}_{aj}^{\init}) &= \cO(1) \label{eq:max-hf}\\
        \hat{L}_{np} := \max_{j\in[p]}|\hf_{aj}(\hat{\theta}_{aj}^{\init})^{-1} - f_{aj}(\tilde{\theta}_{aj})^{-1}| &= \cO(n^{-\kappa}) \label{eq:max-L}
    \end{align}
    We begin by writing the estimation error as 
    \begin{align*}
        \hat{\theta}_{aj} - \tilde{\theta}_{aj}
        &=\hat{\theta}_{aj}^{\init} + \frac{1}{\hf_j(\hat{\theta}_{aj}^{\init})} \PP_n[\hat{\omega}_{aj}(\bZ, \hat{\theta}_{aj}^{\init})]  - \tilde{\theta}_{aj}\\
         &=  f_{aj}({\tilde{\theta}}_{aj})^{-1}\PP_n[{\omega}_{aj}(\bZ, \tilde{\theta}_{aj})] \\        
        &\qquad + \hf_{aj}(\hat{\theta}_{aj}^{\init})^{-1}(\PP_n-\PP)[\hat{\omega}_{aj}(\bZ, \hat{\theta}_{aj}^{\init}) - {\omega}_{aj}(\bZ, \hat{\theta}_{aj}^{\init})] \\
        &\qquad +\hf_{aj}(\hat{\theta}_{aj}^{\init})^{-1}\PP[\hat{\omega}_{aj}(\bZ, \hat{\theta}_{aj}^{\init}) - {\omega}_{aj}(\bZ, \hat{\theta}_{aj}^{\init})] \\
        &\qquad + (\hat{\theta}_{aj}^{\init} - \tilde{\theta}_{aj})  + \hf_{aj}(\hat{\theta}_{aj}^{\init})^{-1}\PP_n[{\omega}_{aj}(\bZ, \hat{\theta}_{aj}^{\init})] - f_{aj}({\tilde{\theta}}_{aj})^{-1}\PP_n[{\omega}_{aj}(\bZ, \tilde{\theta}_{aj})] \\
        &=: f_{aj}({\tilde{\theta}}_{aj})^{-1}\PP_n[{\omega}_{aj}(\bZ, \tilde{\theta}_{aj})] + T_{j1} + T_{j2} + T_{j3}.
    \end{align*}
    Next, we analyze the three residual terms separately.

    \paragraph{Part (1) $T_{j1}$.} This empirical process term can be uniformly bounded using \Cref{lem:err-omega} with \eqref{eq:max-hf}:
    \begin{align*}
        \max_{j\in[p]} |T_{j1}| &= \cO( \max_{j\in[p]}|(\PP_n-\PP)\{\hat{\omega}_{aj}(\bZ, \hat{\theta}_{aj}^{\init}) - {\omega}_{aj}(\bZ, \hat{\theta}_{aj}^{\init})\}|) \\
        &= \cO\left( \sqrt{\frac{\log p}{n}} n^{-\alpha\wedge\beta} + \frac{\log p}{n}\right).
    \end{align*}
    
    \paragraph{Part (2) $T_{j2}$.} This bias term can be uniformly bounded using \Cref{lem:err-omega} with \eqref{eq:max-hf}:
    \begin{align*}
        \max_{j\in[p]} |T_{j2}| &= \cO( \max_{j\in[p]}|\PP\{\hat{\omega}_{aj}(\bZ, \hat{\theta}_{aj}^{\init}) - {\omega}_{aj}(\bZ, \hat{\theta}_{aj}^{\init})\}|) \\
        &= \cO\left( n^{-(\alpha+\beta)}\right).
    \end{align*}

    \paragraph{Part (3) $T_{j3}$.} This extra bias term can be decomposed into two terms:
    \begin{align*}
       T_{j3} &=  (\hat{\theta}_{aj}^{\init} - \tilde{\theta}_{aj})  + \hf_{aj}(\hat{\theta}_{aj}^{\init})^{-1}\PP_n[{\omega}_{aj}(\bZ, \hat{\theta}_{aj}^{\init})] - f_{aj}({\tilde{\theta}}_{aj})^{-1}\PP_n[{\omega}_{aj}(\bZ, \tilde{\theta}_{aj})] \\
       &= (\hat{\theta}_{aj}^{\init} - \tilde{\theta}_{aj}) + f_{aj}({\tilde{\theta}}_{aj})^{-1}\PP[{\omega}_{aj}(\bZ, \hat{\theta}_{aj}^{\init}) - {\omega}_{aj}(\bZ, \tilde{\theta}_{aj})] \\
       & \qquad + f_{aj}({\tilde{\theta}}_{aj})^{-1}(\PP_n-\PP)[{\omega}_{aj}(\bZ, \hat{\theta}_{aj}^{\init}) - {\omega}_{aj}(\bZ, \tilde{\theta}_{aj})] \\
       &\qquad + [\hf_{aj}(\hat{\theta}_{aj}^{\init})^{-1} - f_{aj}({\tilde{\theta}}_{aj})^{-1}]\PP_n[{\omega}_{aj}(\bZ, \hat{\theta}_{aj}^{\init})],
    \end{align*}
    where we use the fact that $\PP\{{\omega}_{aj}(\bZ, \tilde{\theta}_{aj})\} = 0$.
    From \Cref{lem:err-omega-htheta} and \eqref{eq:max-L}, we have that
    \begin{align*}
        \max_{j\in[p]}|T_{j3}| &= \cO(n^{-2\gamma}) +\cO\left( \sqrt{\frac{\log p}{n}} n^{-\gamma/2} +\frac{\log p}{n}\right)\\
        &\qquad +\cO\left( n^{-\kappa} \left(\frac{\log p}{n} + \sqrt{\frac{\log p}{n}} +  n^{-\gamma} \right) \right) \\
        &=\cO \left( n^{-2\gamma} + n^{-(\gamma+\kappa)} +   \frac{\log p}{n} + \sqrt{\frac{\log p}{n}} n^{- \frac{\gamma}{2} \wedge \kappa} \right).
    \end{align*}

    Combining the three terms yields 
    \begin{align*}
        \hat{\theta}_{aj} - \tilde{\theta}_{aj} &= f_{aj}({\tilde{\theta}}_{aj})^{-1}\PP_n[{\omega}_{aj}(\bZ, \tilde{\theta}_{aj})] + \cO\left(\frac{(\log p)^{1/2}}{n^{1/2+\alpha \wedge \beta \wedge \kappa \wedge  \frac{\gamma}{2}}}  + \frac{\log p}{n} + n^{-(\alpha+\beta)\wedge(\gamma+\kappa)\wedge(2\gamma)}\right) .
    \end{align*}
    Setting $\tilde\varphi_j^\QTE = f_{1j}({\tilde{\theta}}_{aj})^{-1}\PP_n[{\omega}_{1j}(\bZ, \tilde{\theta}_{1j})]  - f_{0j}({\tilde{\theta}}_{0j})^{-1}\PP_n[{\omega}_{0j}(\bZ, \tilde{\theta}_{0j})] $ gives that
    \begin{align*}
        \hat{\tau}_j^{\QTE} - \tilde{\tau}_j^{\QTE} &= \PP_n\{\tilde\varphi_{j}^{\QTE}\} + \cO\left(\frac{(\log p)^{1/2}}{n^{1/2+\alpha \wedge \beta \wedge \kappa \wedge  \frac{\gamma}{2}}}  + \frac{\log p}{n} + n^{-(\alpha+\beta)\wedge(\gamma+\kappa)\wedge(2\gamma)}\right)
    \end{align*}
    Recall we have
    \[\max_{j\in[p]}|\tilde{\theta}_{aj} - \theta_{aj}| \leq \frac{1}{c} \max_{j\in[p]} \sup_{\theta \in \cB(\tilde{\theta}_{aj}, \delta)}|\EE[\Delta_{mj}(a,{\theta})]|  \lesssim \delta_m,
    \]
    which implies
    \begin{align*}
    \max_{j\in[p]}|\tilde{\tau}_j^{\QTE} - \tau_j^{\QTE}| &=  \cO(\delta_m).
    \end{align*}
    Because the above results hold with probability tending to one, it further implies that
    \begin{align*}
        \hat{\tau}_j^{\QTE} - {\tau}_j^{\QTE} &= \PP_n\{\tilde\varphi_{j}^{\QTE}\} + \Op\left(\frac{(\log p)^{1/2}}{n^{1/2+\alpha \wedge \beta \wedge \kappa \wedge  \frac{\gamma}{2}}}  + \frac{\log p}{n} + n^{-(\alpha+\beta)\wedge(\gamma+\kappa)\wedge(2\gamma)} + \delta_m \right),
    \end{align*}
    which finishes the proof.
\end{proof}

\subsection{Proof of \Cref{prop:AN-QTE}}\label{app:subsec:AN-QTE}
\begin{proof}[Proof of \Cref{prop:AN-QTE}]
    From \Cref{thm:lin-QTE}, when $\log p =o(n^{2\left(\frac{1}{4}\wedge\alpha\wedge\beta\wedge\frac{\gamma}{2}\right)})$ the asymptotic normality follows immediately from the triangular-array central limit theorem in \Cref{lem:lindeberg}.

    We next verify that the conditions in \Cref{lem:var} hold under the assumptions.
    For simplicity, we drop the superscript $\QTE$ from $\tilde\varphi_j^\QTE$ and write $\tilde\varphi_j$.
    By the boundedness assumption, we have that 
    \[\max_{j\in[p]}|\tilde\varphi_j(\bZ; \PP) + \tilde\varphi_j(\bZ; \hat{\PP})|=\cO(1).\]
    Notice that
    \begin{align*}
        &\tilde\varphi_{j}(\bZ; \hat{\PP}) - \tilde\varphi_{j}(\bZ; {\PP}) \\
        &= \sum_{a\in\{0,1\}} \frac{1}{\hf_{aj}(\hat{\theta}_{aj}^{\init})}\hat{\omega}_{aj}(\bZ, \hat{\theta}_{aj}^{\init}) - \frac{1}{f_{aj}({\tilde{\theta}}_{aj})} {\omega}_{aj}(\bZ, {\tilde{\theta}}_{aj}) \\
        &= \sum_{a\in\{0,1\}} \left [\left(\frac{1}{\hf_{aj}(\hat{\theta}_{aj}^{\init})} - \frac{1}{f_{aj}({\tilde{\theta}}_{aj})} \right)\hat{\omega}_{aj}(\bZ, \hat{\theta}_{aj}^{\init}) + \frac{1}{f_{aj}({\tilde{\theta}}_{aj})} \left(\hat{\omega}_{aj}(\bZ, \hat{\theta}_{aj}^{\init})-{\omega}_{aj}(\bZ, \hat{\theta}_{aj}^{\init})\right)\right. \\
        &\left.+\frac{1}{f_{aj}({\tilde{\theta}}_{aj})} \left({\omega}_{aj}(\bZ, \hat{\theta}_{aj}^{\init})-{\omega}_{aj}(\bZ, {\tilde{\theta}}_{aj})\right) \right] .
    \end{align*}
    Then we have that
    \begin{align*}
        & \max_{j\in[p]}\PP[|\tilde\varphi_{j}(\bZ; \hat{\PP}) - \varphi_{j}(\bZ; {\PP})|] \\
        &\leq \max_{j\in[p]}\|\tilde\varphi_{j}(\bZ; \hat{\PP}) - \tilde\varphi_{j}(\bZ; {\PP})\|_{L_2} \\
        &\leq \max_{j\in[p]}\sum_{a\in\{0,1\}} \left[\left|\frac{1}{\hf_{aj}(\hat{\theta}_{aj}^{\init})} - \frac{1}{f_{aj}(\tilde{\theta}_{aj})} \right| \|{\omega}_{aj}(\bZ, \hat{\theta}_{aj}^{\init})\|_{L_2} + \frac{1}{f_{aj}(\tilde{\theta}_{aj})}\| \hat{\omega}_{aj}(\bZ, \hat{\theta}_{aj}^{\init})-{\omega}_{aj}(\bZ, \hat{\theta}_{aj}^{\init})\|_{L_2} \right.\\
        &\left.+  \frac{1}{f_{aj}(\tilde{\theta}_{aj})} \|{\omega}_{aj}(\bZ, \hat{\theta}_{aj}^{\init})-{\omega}_{aj}(\bZ, \tilde{\theta}_{aj})\|_{L_2}\right]\\
        &=\cO(n^{-\alpha\wedge\beta\wedge\kappa\wedge\frac{\gamma}{2}}).
    \end{align*}
    where in the last equality, we use \eqref{eq:lem:err-omega-eq-2} and Lemma \ref{lem:err-omega-htheta}.
    Therefore, by applying \Cref{lem:var}, we have that $\max_{j\in[p]}|\hat{\sigma}_j^2-\tilde{\sigma}_j^2| = \Op(\log p/n + (\log p)^{1/2}n^{-(1/2+\alpha\wedge\beta\wedge\kappa\wedge\gamma/2)}+ n^{-\alpha\wedge\beta\wedge\kappa\wedge\gamma/2})$ as $m,n,p\rightarrow\infty$ such that $\log p =o(n^{2\left(\frac{1}{4}\wedge\alpha\wedge\beta\wedge\kappa\wedge\frac{\gamma}{2}\right)})$.
\end{proof}

\clearpage
\subsection{Helper lemmas}\label{app:subsec:TE-helper}
    Recall that the uncentered influence function is defined as $\phi_j(\bZ;\PP) = \varphi_j(\bZ;\PP) + \tau_j(\PP)$.
    The lemma below provides results for the asymptotic normality of the one-step estimator of $\EE[Y_j(a)]$ in the setting of multiple derived outcomes.    
    It can be analogously adapted to the one of ATE, which we omit for brevity.

    \begin{lemma}[Counterfactual expectation]\label{lem:ATE}
        For $a\in\{0,1\}$, suppose that \Crefrange{asm:consistency}{asm:nuc} hold and $\tilde{\bY}(a)$ is asymptotically unbiased to $\bY(a)$.
        For $j=1,\ldots,p$, define $\tau_{aj} = \EE[Y_j(a)]$, $\pi_a(\bW) = \PP(A=a\mid\bW)$ and $\mu_{aj}(\bW) = \EE[\tilde{Y}_j\mid A=a,\bW]$, the uncentered influence function as
        \begin{align*}
            \phi_{aj}(\bZ; \pi,\bmu) &= \frac{\ind\{A=a\}}{\pi_a(\bW)} (\tilde{Y}_j - \mu_{aj}(\bW) ) + \mu_{aj}(\bW),
        \end{align*}
        and the one-step estimator as $\hat{\tau}_{aj} = \PP_n[\phi_{aj}(\bZ;\hat{\pi},\hat{\bmu})]$, where $\PP_n$ is the empirical measure over $\cD=\{\bZ_1,\ldots,\bZ_n\}$ and $(\hat{\pi},\hat{\bmu})$ is an estimate of $({\pi},{\bmu})$ from samples independent of $\cD$.  
        
        Assume additionally the following holds for $j\in[p]$ with probability at least $1-n^{-c}$ for some constant $c>0$:
        \begin{enumerate}[label=(\arabic*)]
        
            \item Boundedness: $\tilde{Y}_j$, $\mu_{aj}(\bW)$ and $\hat{\mu}_{aj}(\bW)$ are bounded in $[-C,C]$ for some constant $C$.
            \item\label{asm:lem:ATE-bias} Bias of derived outcomes: $\Delta_{mj}(a):=\EE[\tilde{Y}_j(a)\mid\bW,\bS(a)]-Y_j(a)$ satisfies $\max_{j\in[p]}|\EE[\Delta_{mj}(a)]| = \delta_m = o(n^{-1/2})$.
            \item Strict positivity: The true and estimated propensity score functions satisfy $\pi_a\geq\epsilon$ and $\hat{\pi}_a\geq\epsilon$ for some constant $\epsilon\in(0,1/2)$.
            \item Nuisance: 
            The rates of nuisance estimates satisfy 
            $\|\hat{\pi}_a-\pi_a\|_{L_2} = \cO(n^{-\alpha})$, $\max_{j\in[p]}\|\hat{\mu}_{aj}-\mu_{aj}\|_{L_{2}} = \cO(n^{-\beta})$ for some $\alpha,\beta\in(0,1/2)$, and $\alpha+\beta >1/2$.
        \end{enumerate}
        Then as $m,n,p\rightarrow\infty$ such that $\log p =o(n^{\min(\frac{1}{2}, 2\alpha, 2\beta)})$, it holds that
        \begin{enumerate}[label=(\arabic*)]
            \item Asymptotic normality:
            \begin{align*}
                \sqrt{n}(\hat{\tau}_{aj} - \tau_{aj}) = \sqrt{n}(\PP_n-\PP)\phi_{aj}(\bZ; \pi,\bmu) + \sqrt{n}\epsilon_j \dto \cN(0,\sigma_{aj}^2),
            \end{align*}
            where $\sigma_{aj}^2 = \Var[\phi_{aj}(\bZ;\pi,\bmu)]$ and $\max_{j\in[p]}|\epsilon_j| = \Op((\log p)^{1/2}n^{-(1/2+\alpha \wedge \beta)} + \log p /n + n^{-(\alpha+\beta)} + \delta_m)$.

            \item Uniform control of variance estimation error:
            $\max_{j\in[p]}|\hat{\sigma}_{aj}^2 - \sigma_{aj}^2| = \Op(n^{-\alpha\wedge\beta})$, where $\hat{\sigma}_j^2 = \VV_n[\phi_j(\bZ;\hat{\PP})]$.
        \end{enumerate}
    \end{lemma}
    \begin{proof}[Proof of \Cref{lem:ATE}]    
        We write $\tilde{\tau}_{aj}(\PP)=\EE[\phi_{aj}(\bZ; \pi,\bmu)]$. Note that $\tau_{aj} =\tilde{\tau}_{aj} + \EE[\Delta_{mj}(a)]$ from \eqref{eq:prop:identification-eq-1} in the proof of \Cref{prop:identification}.

        From \citet[Example 2]{kennedy2022semiparametric}, $\phi_{aj}$ is the efficient influence function of $\tilde{\tau}_{aj}$.
        Then, from \eqref{eq:decomposition}, we have a three-term decomposition
        \begin{align}
            \hat{\tau}_{aj}(\PP) - \tau_{aj}(\PP) &= \hat{\tau}_{aj}(\PP) - \tilde{\tau}_{aj}(\PP)  + \tilde{\tau}_{aj}(\PP)  - \tau_{aj}(\PP) \notag\\
            &= T_{{\rm S},j} + T_{{\rm E},j} + T_{{\rm R},j} + \EE[\Delta_{mj}], \label{eq:lem:ATE-decomp}
        \end{align}
        where
        \begin{align*}
            T_{{\rm S},j} & = (\PP_n - \PP) \{\phi_{aj}(\bZ; \PP)\} \\
            T_{{\rm E},j} &= (\PP_n - \PP) \{\phi_{aj}(\bZ; \hat{\PP}) - \phi_{aj}(\bZ; \PP)\} \\
            T_{{\rm R},j} &= \PP \{ \phi_{aj}(\bZ; \hat{\PP}) - \phi_{aj}(\bZ; \PP) \}.
        \end{align*}
        From assumption \ref{asm:lem:ATE-bias}, we know that the last term can be uniformly controlled as $\delta_m = \max_{j\in[p]} |\EE[\Delta_{mj}(a)]| = \op(n^{-1/2})$ as $m,n,p$ tend to infinity.

        \paragraph{Part (1) Uniform control of empirical process terms.}
        We next verify that the conditions in \Cref{lem:emp-process-term} hold under the assumptions.
        Note that
        \begin{align*}
            \phi_{aj}(\bZ; \hat{\PP}) - \phi_{aj}(\bZ; \PP) &= \left(1-\frac{\ind\{A=a\}}{\pi_a(\bW)}  \right)(\hat{\mu}_{aj}(\bW) - \mu_{aj}(\bW)) \\
            &\qquad + \frac{\ind\{A=a\}}{\hat{\pi}_a(\bW)\pi_a(\bW) }(\tilde{Y}_j - \hat{\mu}_{aj}(\bW)) (\pi_a(\bW) - \hat{\pi}_a(\bW)).
        \end{align*}
        Then, by the boundedness assumptions, we have that
        \begin{align*}
            \max_{j\in[p]}\|\phi_{aj}(\bZ; \hat{\PP}) - \phi_{aj}(\bZ; \PP)\|_{L_{2}} &\leq  \frac{1}{\epsilon}  \max_{j\in[p]}\|\hat{\mu}_{aj}-\mu_{aj}\|_{L_{2}} + \frac{1}{\epsilon^2} \|\hat{\pi}_a - \pi_a\|_{L_2} = \cO(n^{-\alpha \wedge \beta}) \\
            \left\|\max_{j\in[p]} | \phi_{aj}(\bZ; \hat{\PP}) - \phi_{aj}(\bZ; \PP)| \right\|_{L_{\infty}} &\leq \frac{2C}{\epsilon^2},
        \end{align*}
        From \Cref{lem:emp-process-term}, when $\log p =o(n^{\min(\frac{1}{2}, 2\alpha, 2\beta)})$, it follows that
        \begin{align}
            \max_{j\in[p]}|T_{{\rm E},j}| & = \cO\left(\sqrt{\frac{\log p}{n}}  n^{-\alpha \wedge \beta}  + \frac{\log p}{n} \right) = o\left( n^{-1/2}\right). \label{eq:lem:ATE-CLT}
        \end{align}

        \paragraph{Part (2) Uniform control of remaining bias terms.}
        From \citet[Example 54]{kennedy2022semiparametric}, it follows that
        \begin{align*}
            T_{{\rm R},j} &= \PP\left\{\left(\frac{1}{\pi_a}-\frac{1}{\hat{\pi}_a}\right)(\mu_{aj}  - \hat{\mu}_{aj}) \pi_a\right\}.
        \end{align*}
        When $\hat{\pi}\geq\epsilon$, conditioned on nuisance estimators, by the Cauchy-Schwarz inequality, we have
        \begin{align}
            \max_{j\in[p]}|T_{{\rm R},j}| &\leq \frac{1}{\epsilon} \|\hat{\pi}_a-{\pi}_a\|_{L_2}\max_{j\in[p]}\|\hat{\mu}_{aj}-{\mu}_{aj}\|_{L_2}
            = \cO(n^{-(\alpha+\beta)}) = o(n^{-1/2}), \label{eq:lem:ATE-bias}
        \end{align}
        where the second inequality is from the Cauchy-Schwarz inequality, and the last equality is from the assumed rate for nuisance function estimation.

        From \eqref{eq:lem:ATE-decomp}-\eqref{eq:lem:ATE-bias}, we have that 
        \[\hat{\tau}_{aj}(\PP) - \tau_{aj}(\PP) = (\PP_n-\PP) \{\phi_{aj}(\bZ; \PP)\} + \epsilon_j\]
        such that $\max_{j\in[p]}|\epsilon_j| = \cO((\log p)^{1/2}n^{-(1/2+\alpha \wedge \beta)} + \log p /n + n^{-(\alpha+\beta)} + \delta_m) = o(n^{-1/2})$.
        
        \paragraph{Part (3) Sample average terms.}
        
        By \Cref{lem:lindeberg}, it follows that
        \begin{align*}
            \sqrt{n} T_{{\rm S},j} \dto \cN(0,\sigma_{aj}^2).
        \end{align*}

        \paragraph{Part (4) Uniform control of variance estimates.}
        We next verify that the conditions in \Cref{lem:var} hold under the assumptions.
        Denote $\hat{\phi}_j = \phi_j(\bZ; \hat{\PP}) $ and $\phi_j = \phi_j(\bZ; {\PP})$.
        Note that by the boundedness conditions, we have that
        \begin{align}
            \max_{1\leq j\leq p}|\hat{\phi}_j + \phi_j| \lesssim 1,\qquad \max_{1\leq j\leq p}|\hat{\phi}_j - \phi_j| \lesssim 1, \label{eq:lem:ATE-bound}
        \end{align}
        which verifies the first condition of \Cref{lem:var}.

        From Part (2), we also have
        \begin{align*}
            \max_{1\leq j\leq p}\|\hat{\phi}_j-\phi_j\|_{L_1} \leq \max_{1\leq j\leq p}\|\hat{\phi}_j-\phi_j\|_{L_2} = \cO\left(n^{-\alpha \wedge \beta} \right) 
        \end{align*}
        which verifies the last two conditions of \Cref{lem:var}.
        
        Therefore, we are able to apply \Cref{lem:var} to conclude that \[\max_{j\in[p]}|\hat{\sigma}_j^2-\sigma_j^2| = \cO\left(\frac{\log p}{n}  + \sqrt{\frac{\log p}{n}} n^{-\alpha\wedge\beta} + n^{-\alpha \wedge \beta}\right) = \cO(n^{-\alpha\wedge\beta}). \]
        when $\log p =o(n^{\min(\frac{1}{2}, 2\alpha, 2\beta)})$.
    \end{proof}

    \begin{lemma}[Error bounds of estimating equations]\label{lem:err-omega}
        Under the event that the conditions in \Cref{thm:lin-QTE} hold, it holds that
        \begin{align*}
            \max_{j\in[p]}|(\PP_n - \PP)\{\hat{\omega}_{aj}(\bZ, \hat{\theta}_{aj}^{\init}) - {\omega}_{aj}(\bZ, \hat{\theta}_{aj}^{\init})\}| &= \cO\left( \sqrt{\frac{\log p}{n}} n^{-\alpha\wedge\beta} + \frac{\log p}{n}\right) \\
            \max_{j\in[p]}|\PP\{\hat{\omega}_{aj}(\bZ, \hat{\theta}_{aj}^{\init}) - {\omega}_{aj}(\bZ, \hat{\theta}_{aj}^{\init})\}| &= \cO\left( n^{-(\alpha+\beta)}\right).
        \end{align*}
    \end{lemma}
    \begin{proof}[Proof of \Cref{lem:err-omega}]
        We begin by decomposing the estimation error as
        \begin{align*}
            \hat{\omega}_{aj}(\bZ, \hat{\theta}_{aj}^{\init}) - {\omega}_{aj}(\bZ, \hat{\theta}_{aj}^{\init}) &= \left(\frac{\ind\{A=a\}}{\pi_a(\bW)} - 1 \right)(\hat{\nu}_{aj}(\bW,\hat{\theta}_{aj}^{\init}) - \nu_{aj}(\bW,\hat{\theta}_{aj}^{\init})) \\    
            &\qquad +\left(\frac{1}{\hat{\pi}_a(\bW)}-\frac{1}{\pi_a(\bW)}\right) {\ind\{A=a\}}({\nu}_{aj}(\bW,\hat{\theta}_{aj}^{\init}) - \psi(Y_j,\hat{\theta}_{aj}^{\init})) \\
            &\qquad +\left(\frac{1}{\hat{\pi}_a(\bW)}-\frac{1}{\pi_a(\bW)}\right) {\ind\{A=a\}}(\hat{\nu}_{aj}(\bW,\hat{\theta}_{aj}^{\init}) - \nu_{aj}(\bW,\hat{\theta}_{aj}^{\init})) \\
            &=: D_{j1}(\bZ) + D_{j2}(\bZ) + D_{j3}(\bZ).
        \end{align*}

        Next, we split the proof into two parts.

        \paragraph{Part (1)}     
        From the boundedness assumption \ref{asm:quantile-boundedness}, we have that
        \begin{align*}
            \left\|\max_{j\in[p]}|\hat{\omega}_{aj}(\bZ, \hat{\theta}_{aj}^{\init}) - {\omega}_{aj}(\bZ, \hat{\theta}_{aj}^{\init})|\right\|_{L_{\infty}} &\leq 2\left(\frac{1}{\epsilon}+1\right)  + \frac{4}{\epsilon} + \frac{4}{\epsilon} \leq C', 
        \end{align*}
        for some constant $C'>0$.
        This gives the $L_{\infty}$-boundedness of the estimation error.
    
        We next derive the upper bound of the $L_2$-norm by analyzing the three terms separately.       
        Condition \ref{asm:quantile-init} implies that
        \begin{align*}
            \max_{j\in[p]}\|D_{j1}\|_{L_2} & \leq \frac{1}{\epsilon}  \max_{j\in[p]}\sup_{\theta \in \cB(\tilde{\theta}_{aj},\delta)} \|\hat{\nu}_{aj}(\bW,\theta) - \nu_{aj}(\bW,\theta) \|_{L_2} = \cO(n^{-\alpha})\\
            \max_{j\in[p]}\|D_{j2}\|_{L_2} & \leq \frac{1}{\epsilon^2} \|\hat{\pi}_{a} - \pi_{a} \|_{L_2} = \cO(n^{-\beta}) \\
            \max_{j\in[p]}\|D_{j3}\|_{L_2} & \leq \frac{1}{\epsilon^2} \|\hat{\pi}_{a} - \pi_{a} \|_{L_2} = \cO(n^{-\beta}) .
        \end{align*}
        Thus, we have
        \begin{align}
            \max_{j\in[p]} \|\hat{\omega}_{aj}(\bZ, \hat{\theta}_{aj}^{\init}) - {\omega}_{aj}(\bZ, \hat{\theta}_{aj}^{\init})\|_{L_2} =\cO(n^{-\alpha\wedge\beta}). \label{eq:lem:err-omega-eq-2}
        \end{align}
        From \Cref{lem:emp-process-term}, we have that
        \begin{align*}
            (\PP_n - \PP)\{\hat{\omega}_{aj}(\bZ, \hat{\theta}_{aj}^{\init}) - {\omega}_{aj}(\bZ, \hat{\theta}_{aj}^{\init})\} &= \cO\left( \sqrt{\frac{\log p}{n}} n^{-\alpha\wedge\beta} + \frac{\log p}{n}\right).
        \end{align*}
        which finishes the proof of the first part of the lemma.
    
        \paragraph{Part (2)} 
        On the other hand, because $\EE[\ind\{A=a\}\mid \bW] = \pi_a(\bW)$, we have that
        \begin{align*}
            \PP\{D_{j1}(\bZ)\} &= 0.
        \end{align*}
        Similarly, because $\EE[\psi(Y_j,\hat{\theta}_{aj}^{\init})\mid \bW] = {\nu}_{aj}(\bW,\hat{\theta}_{aj}^{\init})$, we also have that
        \begin{align*}
            \PP\{D_{j2}(\bZ)\} &= 0.
        \end{align*}
        For the third term, we have that
        \begin{align*}
            \max_{j\in[p]}|\PP\{D_{j3}(\bZ)\}| &\leq \frac{1}{\epsilon} \|\hat{\pi}_{a} - \pi_{a} \|_{L_2} \max_{j\in[p]}\sup_{\theta \in \cB(\tilde{\theta}_{aj},\delta)} \|\hat{\nu}_{aj}(\bW,\theta) - \nu_{aj}(\bW,\theta) \|_{L_2} \\
            &= \cO(n^{-(\alpha+\beta)}).
        \end{align*}
        Combining the above results, we have that
        \begin{align*}
            \max_{j\in[p]} |\PP \{\hat{\omega}_{aj}(\bZ, \hat{\theta}_{aj}^{\init}) - {\omega}_{aj}(\bZ, \hat{\theta}_{aj}^{\init})\} | &= \cO(n^{-(\alpha+\beta)}).
        \end{align*}
    \end{proof}
    
    \begin{lemma}[Error bounds of estimating equations evaluated with respect to the initial estimators]\label{lem:err-omega-htheta}
        Under the event that the conditions in \Cref{thm:lin-QTE} hold, it holds that
        \begin{align*}
            \max_{j\in[p]}|(\hat{\theta}_{aj}^{\init} - \tilde{\theta}_{aj}) + f_{aj}({\tilde{\theta}}_{aj})^{-1}\PP[{\omega}_{aj}(\bZ, \hat{\theta}_{aj}^{\init})-{\omega}_{aj}(\bZ, \tilde{\theta}_{aj})]| &= \cO(n^{-2\gamma}) \\
            \max_{j\in[p]}|(\PP_n-\PP)[{\omega}_{aj}(\bZ, \hat{\theta}_{aj}^{\init}) - {\omega}_{aj}(\bZ, \tilde{\theta}_{aj})]| &= \cO\left( \sqrt{\frac{\log p}{n}} n^{-\gamma/2} + \frac{\log p}{n} \right) \\
            \max_{j\in[p]}|\PP_n\{{\omega}_{aj}(\bZ, \hat{\theta}_{aj}^{\init})\}| &= \cO\left(  \frac{\log p}{n} + \sqrt{\frac{\log p}{n}} +  n^{-\gamma} \right).
        \end{align*}
    \end{lemma}
    \begin{proof}[Proof of \Cref{lem:err-omega-htheta}]
        We split the proof into different parts.

        \paragraph{Part (1)} 
        By Taylor's expansion, we have that
        \begin{align}
            \PP\{{\omega}_{aj}(\bZ, \hat{\theta}_{aj}^{\init})\}&=-\PP\{\psi(\tilde{Y}_j(a),\hat{\theta}_{aj}^{\init})\} \notag\\
            &= -f_{aj}(\tilde{\theta}_{aj}) (\hat{\theta}_{aj}^{\init} - \tilde{\theta}_{aj}) + \cO(|\hat{\theta}_{aj}^{\init} - \tilde{\theta}_{aj}|^2) \notag\\
            & = -f_{aj}(\tilde{\theta}_{aj}) (\hat{\theta}_{aj}^{\init} - \tilde{\theta}_{aj}) + \cO(n^{-2\gamma}) \notag\\
            &= \cO(n^{-\gamma}), \label{eq:lem:err-omega-htheta-1}
        \end{align}
        which is uniformly over $j\in[p]$ by Conditions \ref{asm:quantile-boundedness} and \ref{asm:quantile-init}.
        Note that
        \begin{align*}
            \PP\{{\omega}_{aj}(\bZ, \hat{\theta}_{aj}^{\init}) - {\omega}_{aj}(\bZ, \tilde{\theta}_{aj})\} &= -\PP{\psi}\{(Y_j(a), \hat{\theta}_{aj}^{\init})\} + 0 \\
            &= -\PP{\psi}\{(Y_j(a), \hat{\theta}_{aj}^{\init})\}\\
            &= - f_{aj}(\tilde{\theta}_{aj}) (\hat{\theta}_{aj}^{\init} - \tilde{\theta}_{aj}) + \cO(n^{-2\gamma}) .
        \end{align*}
        Then, we have that
        \begin{align*}
            &\max_{j\in[p]}|(\hat{\theta}_{aj}^{\init} - \tilde{\theta}_{aj}) + f_{aj}({\tilde{\theta}}_{aj})^{-1}\PP[{\omega}_{aj}(\bZ, \hat{\theta}_{aj}^{\init})-{\omega}_{aj}(\bZ, \tilde{\theta}_{aj})]| \\
            =& \max_{j\in[p]} |(\hat{\theta}_{aj}^{\init} - \tilde{\theta}_{aj}) + f_{aj}(\tilde{\theta}_{aj})^{-1} (- f(\tilde{\theta}_{aj}) (\hat{\theta}_{aj}^{\init} - \tilde{\theta}_{aj}))| + \cO(n^{-2\gamma})\\
            =& \cO(n^{-2\gamma}),
        \end{align*}
        which finishes the proof of the first part.

        \paragraph{Part (2)} Note that
        \begin{align}
            &\max_{j\in[p]}\|{\omega}_{aj}(\bZ, \hat{\theta}_{aj}^{\init}) - {\omega}_{aj}(\bZ, \tilde{\theta}_{aj})\|_{L_2} \notag\\
             &=\max_{j\in[p]} 
            \left[\left\|\left(\frac{\ind\{A=a\}}{\pi_a(\bW)} - 1\right) \left({\nu}_{aj}(\bW,\hat{\theta}_{aj}^{\init}) - \nu_{aj}(\bW,\tilde{\theta}_{aj})\right) \right\|_{L_2} \right. \\
             & \qquad\left.+ \left \| \frac{\ind\{A=a\}}{\pi_a(\bW)} (\ind\{\tilde{Y}_j \leq \hat{\theta}_{aj}^{\init}\} - \ind\{\tilde{Y}_j \leq \tilde{\theta}_{aj} \}) \right \|_{L_2}\right] \notag\\
             & \lesssim \max_{j\in[p]} 
            \left[\left\| {\nu}_{aj}(\bW,\hat{\theta}_{aj}^{\init}) - \nu_{aj}(\bW,\tilde{\theta}_{aj}) \right\|_{L_2}  + \left \|  \ind\{\tilde{Y}_j(a) \leq \hat{\theta}_{aj}^{\init}\} - \ind\{\tilde{Y}_j(a) \leq \tilde{\theta}_{aj} \} \right \|_{L_2}\right].
            \label{eq:lem:err-omega-htheta-eq-2}
        \end{align}
        Since $\hat{\theta}_{aj}^{\init}$ is estimated from a separate independent sample, in the following analysis, we condition on $\hat{\theta}_{aj}^{\init}$. By Jenson's inequality, we have
        \[
        \begin{aligned}
            &\,\EE[(\nu_{aj}(\bW, \hat{\theta}_{aj}^{\init}) - \nu_{aj}(\bW, \tilde{\theta}_{aj}))^2 ] \\
            = &\,  \EE[(\PP(\tilde{Y}_j \leq \hat{\theta}_{aj}^{\init} \mid \bW, A=a) - \PP (\tilde{Y}_j \leq {\tilde{\theta}}_{aj} \mid \bW, A=a))^2 ] \\
            = &\, \EE [ (\EE[\ind\{\tilde{Y}_j \leq \hat{\theta}_{aj}^{\init}\}-\ind\{\tilde{Y}_j \leq \tilde{\theta}_{aj}\} \mid \bW, A=a])^2] \\
            \leq &\, \EE [ \EE[(\ind\{\tilde{Y}_j \leq \hat{\theta}_{aj}^{\init}\}-\ind\{\tilde{Y}_j \leq \tilde{\theta}_{aj}\})^2 \mid \bW, A=a]] \\
            =&\, \EE [ \EE[\ind\{\tilde{Y}_j \leq \hat{\theta}_{aj}^{\init}\}+ \ind\{\tilde{Y}_j \leq \tilde{\theta}_{aj}\} -2 \ind\{\tilde{Y}_j \leq \hat{\theta}_{aj}^{\init}\wedge \tilde{\theta}_{aj}\} \mid \bW, A=a]] \\
            = &\, \PP(\tilde{Y}_j(a) \leq \hat{\theta}_{aj}^{\init}) + \PP(\tilde{Y}_j(a) \leq \tilde{\theta}_{aj}) - 2\PP(\tilde{Y}_j (a)\leq \hat{\theta}_{aj}^{\init}\wedge \tilde{\theta}_{aj})
        \end{aligned}
        \]
        where the last equation follows from the identification equation $\EE[\EE(\ind\{\tilde{Y}_j \leq \tilde{\theta}\}\mid \bW, A=a) ] = \PP(\tilde{Y}_j (a)\leq \tilde{\theta})$. Since $f_{aj}$ is uniformly bounded in $\cB(\tilde{\theta}_{aj},\delta)$ from Assumption \ref{asm:quantile-boundedness}, by Taylor's expansion we have
        \[
        \begin{aligned}
            &\,\PP(\tilde{Y}_j(a) \leq \hat{\theta}_{aj}^{\init}) + \PP(\tilde{Y}_j(a) \leq \tilde{\theta}_{aj}) - 2\PP(\tilde{Y}_j (a)\leq \hat{\theta}_{aj}^{\init}\wedge \tilde{\theta}_{aj})\\
            = &\, |\PP(\tilde{Y}_j(a) \leq \hat{\theta}_{aj}^{\init}) - \PP(\tilde{Y}_j(a) \leq \tilde{\theta}_{aj})| \\
            \lesssim &\, |\hat{\theta}_{aj}^{\init} - \tilde{\theta}_{aj}| = \cO(n^{-\gamma}).
        \end{aligned}
        \]
        The upper bound is uniform over $j \in [p]$ and we have
        \[
        \max_{j\in[p]} 
            \left\| {\nu}_{aj}(\bW,\hat{\theta}_{aj}^{\init}) - \nu_{aj}(\bW,\tilde{\theta}_{aj}) \right\|_{L_2} = \cO(n^{-\gamma/2}).
        \]
        Similarly, we have
        \[
        \begin{aligned}
            &\, \left \|  \ind\{\tilde{Y}_j(a) \leq \hat{\theta}_{aj}^{\init}\} - \ind\{\tilde{Y}_j(a) \leq \tilde{\theta}_{aj} \} \right \|_{L_2}^2 \\
            = &\, \EE [\ind\{\tilde{Y}_j(a) \leq \hat{\theta}_{aj}^{\init}\} + \ind\{\tilde{Y}_j(a) \leq \tilde{\theta}_{aj}\}-2 \ind\{\tilde{Y}_j(a) \leq \hat{\theta}_{aj}^{\init}\} \ind\{\tilde{Y}_j(a) \leq \tilde{\theta}_{aj}\}] \\
            = &\, \PP(\tilde{Y}_j(a) \leq \hat{\theta}_{aj}^{\init}) + \PP(\tilde{Y}_j(a) \leq \tilde{\theta}_{aj}) - 2\PP(\tilde{Y}_j (a)\leq \hat{\theta}_{aj}^{\init}\wedge \tilde{\theta}_{aj}) \\
            = &\, \cO(n^{-\gamma})
        \end{aligned}
        \]
        \[
        \max_{1 \leq j \leq p}\left \|  \ind\{\tilde{Y}_j(a) \leq \hat{\theta}_{aj}^{\init}\} - \ind\{\tilde{Y}_j(a) \leq \tilde{\theta}_{aj} \} \right \|_{L_2} = \cO(n^{-\gamma/2}).
        \]
        From \Cref{lem:emp-process-term}, we have that
        \begin{align*}
            \max_{j\in[p]}|(\PP_n-\PP)\{{\omega}_{aj}(\bZ, \hat{\theta}_{aj}^{\init}) - {\omega}_{aj}(\bZ, \tilde{\theta}_{aj})\}| &= \cO\left( \sqrt{\frac{\log p}{n}} n^{-\gamma/2} + \frac{\log p}{n} \right),
        \end{align*}
        which finishes the proof of the second part.

        \paragraph{Part (3)} Note that
        \begin{align*}
            \PP_n\{{\omega}_{aj}(\bZ, \hat{\theta}_{aj}^{\init} )\} &= (\PP_n-\PP)[{\omega}_{aj}(\bZ, \hat{\theta}_{aj}^{\init}) - {\omega}_{aj}(\bZ, \tilde{\theta}_{aj})] + \PP_n [{\omega}_{aj}(\bZ, \tilde{\theta}_{aj})] + \PP[{\omega}_{aj}(\bZ, \hat{\theta}_{aj}^{\init})] .
        \end{align*}
        Since $\omega$ is centered and bounded, by Lemma \ref{lem:emp-process-term} we have
        \[
        \max_{1\leq j \leq p}|(\PP_n-\PP) [{\omega}_{aj}(\bZ, \tilde{\theta}_{aj})]| = \cO \left ( \frac{\log p}{n} + \sqrt{\frac{\log p}{n}} \right)= \cO \left ( \sqrt{\frac{\log p}{n}} \right) .
        \]
        Combining it with \eqref{eq:lem:err-omega-htheta-1} and Part (2), we further have
        \begin{align*}
            \max_{j\in[p]}|\PP_n\{{\omega}_{aj}(\bZ, \hat{\theta}_{aj}^{\init} )\}| &=\cO\left( \sqrt{\frac{\log p}{n}} n^{-\gamma/2} + \frac{\log p}{n} + \sqrt{\frac{\log p}{n}} +  n^{-\gamma}\right)\\
            &= \cO\left(  \frac{\log p}{n} + \sqrt{\frac{\log p}{n}} +  n^{-\gamma}\right)
        \end{align*}
        which completes the proof.
    \end{proof}

    \begin{lemma}[Lindeberg CLT for triangular array]\label{lem:lindeberg}
        Let $m=m_n$ and $p=p_n$ be two sequences indexed by $n$.
        Consider the influence-function-based linear expansion for estimator $\hat{\tau}_j$ of $\tau_j$:
        \[\sqrt{n}(\hat{\tau}_j - \tau_j) = \sqrt{n}(\PP_n-\PP)\{\varphi_{m_nj}\} + \epsilon_{m_nj},\qquad j=1,\ldots,p\]
        where $\varphi_{m_nj}$ is the influence function that depends on $m$ and the residual $\epsilon_j$'s satisfy that $\max_{j\in p_n}|\epsilon_{m_nj}|=\op(1)$ as $n\rightarrow\infty$.
        Let $B_n^2=\sum_{i\in[n]}\Var(\varphi_{m_nj}(\bZ_i))$.
        Further assume that (i) there exists a constant $c>0$, such that $\VV(\varphi_{m_nj}(\bZ_1))\geq c$, and (ii) there exists a sequence $\{L_n\}_{n\in\NN}$ such that $\max_{i\in[n]}|\varphi_{m_nj}(\bZ_i)| \leq L_n$ and $L_n/B_n\rightarrow0$, then 
        \[\sqrt{n}\frac{\hat{\tau}_j - \tau_j}{\Var\{\varphi_{m_nj}\}^{1/2}} \dto \cN(0,1) \]
    \end{lemma}
    \begin{proof}[Proof of \Cref{lem:lindeberg}]
        Note that $\varphi_{m_nj}$ is the centered influence function such that $\EE[\varphi_{m_nj}(\bZ)] = 0$.
        Let $X_{nk} = \varphi_{m_nj}(\bZ_k)$.
        From assumption (ii) that $\max_{k\in[n]}|X_{nk}|\leq L_n$ and $L_n/B_n\rightarrow 0$, we have that, for any $\xi>0$,
        \begin{align*}
            \lim _{n \rightarrow \infty} \frac{1}{B_n^2} \sum_{k=1}^{n} \EE\left[ X_{nk}^2 \ind\{|X_{n k}| \geq \xi B_n\} \right]=0.
        \end{align*}
        This verifies Lindeberg's condition for a triangular array of random variables.
        From \citet[Theorem 27.2]{billingsley1995probability}, it follows that
        \[ \frac{n(\PP_n-\PP)\{\varphi_{m_nj}\}}{B_n} \dto \cN(0,1),\]
        as $n\rightarrow\infty$.
        Because $\bZ_i$'s are identically distributed, we have $B_n^2 = n\VV(\varphi_{m_nj})$.
        This implies that
        \[ \frac{\sqrt{n}(\PP_n-\PP)\{\varphi_{m_nj}\}}{\VV(\varphi_{m_nj})^{1/2}} \dto \cN(0,1),\]
        as $n\rightarrow\infty$.

        From assumption (i) that $\VV(\varphi_{m_nj})\geq c>0$ and $\max_{j\in p_n}|\epsilon_{m_nj}|=\op(1)$, we further have 
        \[\max_{j\in p_n} \frac{|\epsilon_{m_nj}|}{\VV(\varphi_{m_nj})^{1/2}} = \op(1)\]
        as $n\rightarrow\infty$.
        Consequently, the conclusion follows.
    \end{proof}

\clearpage
\section{Multiple testing}\label{app:sec:inference}

    \subsection{Proof of \Cref{prop:max-Gaussian}}\label{subsec:max-Gaussian}

    \begin{proof}[Proof of \Cref{prop:max-Gaussian}]
        We present the proof for $\tau_j=\tau_j^\STE$, and the proof for $\tau_j=\tau_j^\QTE$ follows similarly, which we omit for simplicity.
        From \Cref{thm:lin-STE} and \Cref{prop:Var-STE}, we have the following expansion on the statistic:
        \begin{align}
            \sqrt{n}\frac{\hat{\tau}_j - \tau_j^*}{\hat{\sigma}_j} &= \frac{\sqrt{n}}{\hat{\sigma}_j} \frac{1}{n}\sum_{i=1}^n{\varphi}_{ij} + \frac{\sqrt{n} \epsilon_j}{\hsigma_j} ,\notag\\
            &= \frac{\sqrt{n}}{{\sigma}_j} \frac{1}{n}\sum_{i=1}^n{\varphi}_{ij} + \left(\frac{1}{\hat{\sigma}_j}  - \frac{1}{{\sigma}_j} \right)\frac{1}{\sqrt{n}}\sum_{i=1}^n{\varphi}_{ij} + \frac{\sqrt{n} \epsilon_j}{\hsigma_j} ,\notag\\
            &= \frac{\sqrt{n}}{{\sigma}_j} \frac{1}{n}\sum_{i=1}^n\varphi_{ij} + \epsilon_j', \label{eq:lem:max-Gaussian-eq-1}
        \end{align}
        where $\epsilon_j' = \left(\frac{1}{\hat{\sigma}_j}  - \frac{1}{{\sigma}_j} \right)\frac{1}{\sqrt{n}}\sum_{i=1}^n{\varphi}_{ij} + \frac{\sqrt{n} \epsilon_j}{\hsigma_j}$.
        Note that the covariance matrix of the true scaled influence function is given by
        \[\Cov\left(\frac{1}{\sqrt{n}{\sigma}_j} \sum_{i=1}^n {\bvarphi}_{i\cS}\right) = \bD_{\cS}^{-1}\bE_\cS \bD_{\cS}^{-1},\]
        where $\bE_{\cS}=\EE[{\bvarphi}_{i\cS}{\bvarphi}_{i\cS}^{\top}]$ and $\bD_{\cS}=\text{diag}((\hat{\sigma}_j)_{j \in \cS})$.
        The associated gaussian vector is defined as $\bg_{0\cS} \sim\cN(\zero, \bD_\cS^{-1}\bE_\cS\bD_\cS^{-1})$.
        We first consider one-sided test problems with the pair of maximum statistics and the Gaussian vector based on linear expansion \eqref{eq:lem:max-Gaussian-eq-1}:
        \begin{align*}
            \overline{M}_{\cS} &= \max_{j\in\cS}\sqrt{n}\frac{\hat{\tau}_j-\tau_j^*}{\hat{\sigma}_j},\qquad W_{\cS} = \max_{j\in\cS} (g_{\cS})_j,
        \end{align*}
        and the pair based on the true influence functions and variances:
        \begin{align*}
            \overline{M}_{0\cS} &= \max_{j\in\cS} \frac{1}{\sqrt{n}} \sum_{i=1}^n \frac{\varphi_{ij}}{\sigma_j},\qquad             
            W_{0\cS} = \max_{j\in\cS} (g_{0\cS})_j.
        \end{align*}
        We next show that the distribution function of $\overline{M}_\cS$ can be approximated by $W_\cS$ uniformly.

        \paragraph{Step (1) Bounding difference between $W_\cS$ and $W_{0\cS}$.}        
        From Theorem J.1 of \citet{chernozhukov2013gaussian}, 
        \begin{align}
            \sup_{\cS \subseteq \cA^*}\sup_x |\PP(W_\cS > x) - \PP(W_{0\cS} > x\mid \{\bZ_i\}_{i=1}^n)| \pto 0. \label{eq:lem:max-gaussian-eq-1}
        \end{align}
        is implied if $\sup_{\cS \subseteq \cA^*} \|\hat{\bD}_{\cS}^{-1}\hat{\bE}_{\cS}\hat{\bD}_{\cS}^{-1} - \bD_{\cS}^{-1}\bE_{\cS}\bD_{\cS}^{-1}\|_{\max} =\Op(n^{-c})$ for some $c>0$.
        Because $\max_{j\in\cS}\sigma_j^2\geq c>0$, when $\log p =o(n^{2\left(\frac{1}{4}\wedge\alpha\wedge\beta\right)})$ and $\delta_m = o(n^{-1/2})$, we have that $\max_{j\in\cS}|\hat{\sigma}_j^{-1}-\sigma_j^{-1}|=\Op(r_{\sigma})$ from the proof of \Cref{prop:Var-STE}, where $r_{\sigma}= n^{-\alpha\wedge\beta} + \sqrt{\log p /n}$.
        On the other hand, \eqref{eq:lem:max-Gaussian-eq-1} also implies that,
        \begin{align*}
            &\sup_{\cS \subseteq \cA^*}\|\hat{\bE}_{\cS} - \bE_\cS\|_{\max} \\
            &=\max_{k,\ell\in\cA^*}|\hat{E}_{k\ell} - E_{k\ell}| \\
            &= \max_{k,\ell\in\cA^*}\left|\frac{1}{n}\sum_{i=1}^n(\hat{\varphi}_{ik}\hat{\varphi}_{i\ell} - \EE[\varphi_{ik}\varphi_{i\ell}])\right| \\
            &= \max_{k,\ell\in\cA^*}\left|\frac{1}{n}\sum_{i=1}^n[\hat{\varphi}_{ik}(\hat{\varphi}_{i\ell}- \varphi_{i\ell}) +  (\hat{\varphi}_{ik}-\varphi_{ik})\varphi_{i\ell} + (\varphi_{ik}\varphi_{i\ell}-\EE[\varphi_{ik}\varphi_{i\ell}])
            ]\right|\\
            &\leq    \max_{k,\ell\in\cA^*} \left|\frac{1}{n}\sum_{i=1}^n\hat{\varphi}_{ik}(\hat{\varphi}_{i\ell}- \varphi_{i\ell})\right| + \max_{k,\ell\in\cA^*}\left|\frac{1}{n}\sum_{i=1}^n (\hat{\varphi}_{ik}-\varphi_{ik})\varphi_{i\ell}\right| + \max_{k,\ell\in\cA^*}\left|\frac{1}{n}\sum_{i=1}^n (\varphi_{ik}\varphi_{i\ell}-\EE[\varphi_{ik}\varphi_{i\ell}])
            \right|\\
            &\leq \max_{k,\ell\in\cA^*} \max_{i}|\hat{\varphi}_{ik}|\cdot \frac{1}{n}\sum_{i=1}^n|\hat{\varphi}_{i\ell}- \varphi_{i\ell}| + \max_{k,\ell\in\cA^*} \frac{1}{n}\sum_{i=1}^n|\hat{\varphi}_{ik}-\varphi_{ik}|\cdot \max_i|\varphi_{i\ell}| + \Op\left(\sqrt{\frac{\log p}{n}} \right)\\
            &= \Op\left(r_{\varphi} \right)
        \end{align*}
        where in the last inequality we use \eqref{eq:expect-diff-influence-est} with $r_{\varphi}=n^{-\alpha\wedge\beta} + \sqrt{\log p /n}$ and the sub-Gaussianity of $\varphi_{ik}\varphi_{i\ell}$'s.
        Then we have
        \begin{align*}
            \sup_{\cS \subseteq \cA^*}\|\hat{\bD}_\cS^{-1}\hat{\bE}_\cS\hat{\bD}_\cS^{-1} - \bD_\cS^{-1}\bE_\cS\bD_\cS^{-1}\|_{\max} &\leq \sup_{\cS \subseteq \cA^*} \|\hat{\bD}_\cS^{-1}\hat{\bE}_\cS\hat{\bD}_\cS^{-1} - \bD_\cS^{-1}\hat{\bE}_\cS\hat{\bD}_\cS^{-1}\|_{\max} \\
            &\quad +  \sup_{\cS \subseteq \cA^*}\|\bD_\cS^{-1}\hat{\bE}_\cS\hat{\bD}_\cS^{-1} - \bD_\cS^{-1}\hat{\bE}_\cS\bD_\cS^{-1}\|_{\max} \\
            &\quad + \sup_{\cS \subseteq \cA^*}\|\bD_\cS^{-1}\hat{\bE}_\cS\bD_\cS^{-1} - \bD_\cS^{-1}\bE_\cS\bD_\cS^{-1}\|_{\max}\\
            &= \Op\left(r_{\varphi} + r_{\sigma}\right).
        \end{align*}
        Under the condition that $\log^2 (pn)\max\{\log^5(pn)/n, n^{-(\alpha\wedge\beta)}\}= o(1)$, we have that $(\log p)^2(r_{\varphi} + r_{\sigma})= o(1)$.
        Then, from Theorem J.1 of \citet{chernozhukov2013gaussian}, it follows that \eqref{eq:lem:max-gaussian-eq-1} holds.

        \paragraph{Step (2) Bounding difference between $W_{0\cS}$ and $\overline{M}_{0\cS}$.}
        Under \Cref{asm:cor} and the condition that $\log(pn)^7/n\leq C_2 n^{-c_2}$, because $\EE[\varphi_{ij}^2]/\sigma_j^2=1$ and constant $c, C$ is independent of $\cS$, from Corollary 2.1 of \citet{chernozhukov2013gaussian}, we have that
        \begin{align}
            \sup_{\cS \subseteq \cA^*}\sup_x | \PP(W_{0\cS} > x) - \PP(\overline{M}_{0\cS} > x)| \rightarrow 0. \label{eq:lem:max-gaussian-eq-2}
        \end{align}

        \paragraph{Step (3) Bounding difference between $\overline{M}_{0\cS}$ and $\overline{M}_{\cS}$.}
        We begin by bounding $\max_{j\in\cS}|\epsilon_j'|$.        
       
        Because $\sigma_j$ is uniformly lower bounded away from zero, we have $|\hsigma_j-\sigma_j| = |(\hsigma_j^2-\sigma_j^2)/(\hsigma_j+\sigma_j)| \lesssim |\hsigma_j^2-\sigma_j^2|$, which implies that $\max_{j\in[p]}|\hsigma_j-\sigma_j| \lesssim r_{\sigma}$ with probability tending to one from the proof of \Cref{prop:Var-STE}.        
        From the proof of \Cref{thm:lin-STE} and the boundedness assumptions, we have
        \begin{align*}
            \max_{j\in\cS} |\epsilon_j'|\leq \max_{j\in\cA^*} |\epsilon_j'| &\leq \max_{j\in\cA^*} \left|\frac{1}{\sqrt{n}}\sum_{i}\frac{\varphi_{ij}}{\sigma_j\hsigma_j} (\sigma_j-\hsigma_j) \right| + \max_{j\in\cA^*} \left|\frac{\sqrt{n}\epsilon_j}{\hsigma_j} \right| \notag\\
            &=\Op\left( \max_{j\in\cA^*}|\sigma_j-\hsigma_j| \max_{j\in\cA^*} \left|\frac{1}{\sqrt{n}}\sum_i \varphi_{ij} \right| + \max_{j\in\cA^*}|\sqrt{n}\epsilon_j| \right) \notag\\
            &=\Op \left( r_{\sigma}\sqrt{\log p} + n^{-(\alpha \wedge \beta)}\sqrt{\log p}  + (\log p) /\sqrt{n} + n^{1/2-(\alpha+\beta)} + \sqrt{n }\delta_m\right) \notag\\
            &\, =\Op \left(  n^{-(\alpha \wedge \beta)}\sqrt{\log p}  + (\log p) /\sqrt{n} + n^{1/2-(\alpha+\beta)} + \sqrt{n }\delta_m\right) \notag\\
            &\, = \op(\xi_n) ,
        \end{align*}
        where $\xi_n = [n^{-(\alpha \wedge \beta)}\sqrt{\log p}  + (\log p) /\sqrt{n} + n^{1/2-(\alpha+\beta)}+\sqrt{n}\delta_m]  \log (n) $.        
        Then,
        \begin{align}
            \sup_{\cS \subseteq \cA^*}\PP(|\overline{M}_{0\cS}-\overline{M}_{\cS}|>\xi_n) \leq \PP\left(\max_{j \in \cA^* } |\epsilon_j'|>\xi_n\right)  \rightarrow 0.\label{eq:Delta-1}
        \end{align}
        Then we have
        \begin{align*}
            &| \PP(\overline{M}_{0\cS} > x) - \PP(\overline{M}_{\cS} > x)| \\
            \leq&\, \PP(\overline{M}_{0\cS} \leq x, \overline{M}_{\cS} >x) + \PP(\overline{M}_{0\cS} > x, \overline{M}_{\cS} \leq x) \\
            \leq&\,   \PP(\overline{M}_{0\cS} > x - \xi_n, \overline{M}_{0\cS} \leq x)+\PP(\overline{M}_{0\cS} \leq x+\xi_n, \overline{M}_{0\cS} >x) + 2\PP\left(\max_{j \in \cA^* }|\epsilon_j'|>\xi_n \right)\\
            \leq&\, \PP(x-\xi_n < \overline{M}_{0\cS} \leq x+\xi_n) + 2\PP\left(\max_{j \in \cA^* }|\epsilon_j'|>\xi_n \right)\\
            \leq &\, \PP(x-\xi_n < W_{0\cS} \leq x+\xi_n) + 2\PP\left(\max_{j \in \cA^* }|\epsilon_j'|>\xi_n \right) + 2 \sup_{\cS \subseteq \cA^*}\sup_x | \PP(W_{0\cS} > x) - \PP(\overline{M}_{0\cS} > x)|
        \end{align*}
        which implies that
        \begin{align*}
            &\sup_{\cS \subseteq \cA^*} \sup_x | \PP(\overline{M}_{0\cS} > x) - \PP(\overline{M}_{\cS} > x)| \\
            \leq&\,  \sup_{\cS \subseteq \cA^*}\sup_x \PP(x-\xi_n \leq W_{0\cS} \leq x+\xi_n) + 2\PP\left(\max_{j \in \cA^* }|\epsilon_j'|>\xi_n \right) \\
            &\qquad + 2 \sup_{\cS \subseteq \cA^*}\sup_x | \PP(W_{0\cS} > x) - \PP(\overline{M}_{0\cS} > x)|\\
            \lesssim& \, \xi_n \sqrt{\log (pn)} + 2\PP\left(\max_{j \in \cA^* }|\epsilon_j'|>\xi_n \right) + 2 \sup_{\cS \subseteq \cA^*}\sup_x | \PP(W_{0\cS} > x) - \PP(\overline{M}_{0\cS} > x)|,
        \end{align*}
        where the second inequality is from the anti-concentration inequality \citep[Lemma 2.1]{chernozhukov2013gaussian}.
        Under the condition that $\max\{\log(pn)^7/n, \log(pn)^2 n^{-(\alpha\wedge\beta)}, \sqrt{ n\log (pn)}\delta_m\} \leq C_2 n^{-c_2}$, we have $\xi_n\sqrt{\log(pn)}\lesssim n^{-c_2}\log n = o(1)$.
        This implies that
        \begin{align}
            \sup_{\cS \subseteq \cA^*} \sup_x | \PP(\overline{M}_{0\cS} > x) - \PP(\overline{M}_{\cS} > x)| \rightarrow 0. \label{eq:lem:max-gaussian-eq-3}
        \end{align}

        \paragraph{Step (4) Combining results from previous steps.} Finally, combining \eqref{eq:lem:max-gaussian-eq-1}-\eqref{eq:lem:max-gaussian-eq-3} yields that
        \begin{align*}
            \sup_{\cS \subseteq \cA^*}\sup_x | \PP(\overline{M}_{0\cS} > x) - \PP(W_{0\cS} > x\mid \{\bZ_i\}_{i=1}^n)| \pto 0.
        \end{align*}
        Analogously results also hold for $\overline{M}_{\cS}' = \max_j  - \sqrt{n}(\hat{\tau}_j - \tau_j)/\hsigma_j$, $W_{\cS}'=\max_j -g_j$, $\overline{M}_{0\cS}' =\max_j -n^{-1/2}\sum_{i=1}^n\varphi_{ij}/\sigma_j $ and $W_{0\cS}'=\max_j -(g_0)_j$ , by applying the same argument.
        Thus, we have
        \begin{align*}
            & \sup_{\cS \subseteq \cA^*} \sup_x | \PP({M}_{\cS} > x) - \PP(\|\bg_\cS\|_{\infty} > x\mid \{\bZ_i\}_{i=1}^n)| \\
            \leq & \sup_{\cS \subseteq \cA^*}\sup_x | \PP(\overline{M}_{\cS} > x) - \PP(W_{\cS} > x\mid \{\bZ_i\}_{i=1}^n)| \\
            &\qquad + \sup_{\cS \subseteq \cA^*}\sup_x | \PP(\overline{M}_{\cS}' < x) - \PP(W_{\cS}' < x\mid \{\bZ_i\}_{i=1}^n)| 
            \pto  0,
        \end{align*}
        which finishes the proof.
        \end{proof}

    \subsection{Proof of \Cref{prop:simul-inference}}\label{app:subsec:simul-inference}
    \begin{proof}[Proof of \Cref{prop:simul-inference}]
        We split the proof into different parts.

        \paragraph{Part (1) Exact recovery of the active set.}
        From the proof of \Cref{prop:max-Gaussian}, we have
        \begin{align}
            \max_{1\leq j\leq p}|\hat{\sigma}_j^2 - \sigma_j^2| = \Op( r_{\sigma}).  \label{eq:sigma-uniform-conv}
        \end{align}
        Recall $\cA^* =  \{j \in [p] \mid \sigma_j^2 \geq c_1 \}$.
        From \Cref{asm:cor}, $\min_{j\in\cA^*}\sigma_j^2\geq c_1$ for some constant $c_1>0$.
        
        To screen out noninformative coordinates, define $\cn$ as in Proposition \ref{prop:simul-inference} for some constant $c>0$ and
        \begin{align*}
            {\cA}_1 &= \{j \in [p]\mid \hat{\sigma}_j^2 \geq \cn \}
        \end{align*}
       which is a random quantity because $\hat{\sigma}_j^2$ is the empirical variance.
        Then we have
        \[
        \PP(\cA^* \neq {\cA}_1 ) = \PP\left(\max_{j\in {\cA}_1 \setminus\cA^*} \hat{\sigma}_j^2 \geq \cn \right) + \PP\left(\min_{j\in \cA^*\setminus{\cA}_1} \hat{\sigma}_j^2 < \cn\right).
        \]
        For the first term, we have
        \begin{align*}
             \PP\left(\max_{j\in {\cA}_1 \setminus\cA^*} \hat{\sigma}_j^2 \geq \cn \right)  
            =&\PP\left(\max_{j\in {\cA}_1 \setminus\cA^*} \hat{\sigma}_j^2 - \min_{j\in {\cA}_1 \setminus\cA^*} {\sigma}_j^2 \geq \cn - \min_{j\in {\cA}_1 \setminus\cA^*} {\sigma}_j^2\right)  \\
            \leq &\PP\left(\max_{j\in {\cA}_1 \setminus\cA^*} |\hat{\sigma}_j^2 - {\sigma}_j^2| \geq \cn - \min_{j\in {\cA}_1 \setminus\cA^*} {\sigma}_j^2\right)  \\
            \leq& \PP\left(\max_{j\in [p]} |\hat{\sigma}_j^2-{\sigma}_j^2| \geq c_n - \max_{j\in \cA^{*c}} {\sigma}_j^2 \right)  \rightarrow 0,
        \end{align*}
         where we use the fact $c_n \rightarrow 0, \min_{j \in \cA^*} \sigma_j^2 \geq c_1$ and $\max_{j\in [p]} |\sigma_j^2 - \hat{\sigma}_j^2|=\Op(r_{\sigma}) = \op(c_n)$ as $m,n,p\rightarrow\infty$, from \Cref{asm:cor} and \eqref{eq:sigma-uniform-conv}.
        
        For the second term, similarly we have
        \begin{align*}
            &\,\PP\left(\min_{j\in \cA^*\setminus{\cA}_1} \hat{\sigma}_j^2 < \cn\right)\\            
            \leq &\, \PP\left(\min_{j\in \cA^*\setminus{\cA}_1} \hat{\sigma}_j^2 < \cn, \max_{j\in [p]} |\sigma_j^2 - \hat{\sigma}_j^2| \leq c_n\right) + \PP\left( \max_{j\in [p]} |\sigma_j^2 - \hat{\sigma}_j^2|> c_n\right)\\
            \leq &\, \PP\left(\min_{j \in \cA^*} {\sigma}_j^2 < 2\cn \right) + \PP\left( \max_{j\in [p]} |\sigma_j^2 - \hat{\sigma}_j^2|> c_n\right) \rightarrow 0.
        \end{align*}
        Thus, we have
        \begin{align*}
            \PP(\cA^* = {\cA}_1 ) \rightarrow 1.
        \end{align*}

        \paragraph{Part (2) FWER.} From Part (1), the family-wise error rate satisfies that
        \begin{align*}
            \FWER &= \PP(\hat{\cA}\cap \cV^{*c}\neq\varnothing) \\
            &= \PP(\hat{\cA}\cap \cV^{*c}\neq\varnothing, \cA^* = \cA_1) + o(1).
        \end{align*}
        Recall the maximal statistic is defined as $M_1 = \max_{j\in\hat{\cA}_1}|\sqrt{n}(\hat{\tau}_j-\tau_j)/\hat{\sigma}_j|$ and the null hypothesis is rejected only if $M_1>\hat{q}_{1}(\alpha)$ where $\hat{q}_{1}(\alpha)$ is the multiplier bootstrap quantile.
        From \Cref{lem:max-Gaussian-quantile}, we have that
        \begin{align*}
            \limsup\PP(\hat{\cA}\cap \cV^{*c}\neq\varnothing) &\leq \limsup\PP\left(\max_{j\in{\cA}^*}|\sqrt{n}(\hat{\tau}_j-\tau_j^*)/\hat{\sigma}_j| > \hat{q}_{1}(\alpha), \cA^* = \cA_1\right)\\
            &\leq \alpha .
        \end{align*}
        Since when $\cA^* = \cA_1$, $\hat{q}_{1}(\alpha)$ is also the bootstrap quantile of $\max_{j\in{\cA}^*}|\sqrt{n}(\hat{\tau}_j-\tau_j^*)/\hat{\sigma}_j|$. Therefore, combining the above results yields that
        \begin{align*}
            \limsup \FWER\leq \alpha,
        \end{align*}
        which finishes the proof.
    \end{proof}

    \subsection{Proof of \Cref{thm:FDPex}}\label{app:subsec:FDPex}
    \begin{proof}[Proof of \Cref{thm:FDPex}]
        We split the proof into different parts.

        \paragraph{Part (1) FDPex.}
        Recall $\cV$ is the output of the step-down and augment processes, and $M_j$ is the maximal statistic at step $j$.
        Let $\cV_{\ell^*}$ be the set of discoveries returned by the step-down process.
        From \Cref{prop:simul-inference} (1), the active set $\cA_{\ell+1}\subseteq\cA_{\ell}\subseteq\cA^*$ for all $\ell=1,2,\ldots,\ell^*-1$ with probability tending to one.        
        If $\cV_{\ell^*}\cap\cV^{*}\neq\varnothing$, then for the first $H_0^{(j)}$ generating false discoveries, from \Cref{lem:max-Gaussian-quantile} and \Cref{prop:simul-inference}, under the null hypothesis $H_0^{(j)}$, we always have 
        \[\limsup\FWER = \limsup\PP(M_j>\hat{q}_{j}(\alpha)) \leq \alpha\]
        as $m,n,p\rightarrow\infty$.
        This shows that the step-down procedure controls the family-wise error rate.
        By \citet[Theorem 1]{genovese2006exceedance}, it follows that the FDP after the augmentation step satisfies that \[\limsup \PP(\FDP>c)\leq \alpha,\]
        which finishes the proof of the first conclusion.

        \paragraph{Part (2) Power.}
        From the proof of \citet[Corollary 1]{chang2018confidence}, the standard results on Gaussian maximum imply that $\max_{\ell=1,\ldots,\ell^*}\hat{q}_{\ell}(\alpha)= C\sqrt{\log p}+\op (1)$ for some constant $C>0$.

        We next show that if there exists a $j_0\in\cA_{\ell^*}\cap\cV^*$, then the proposed maximum test on $\cA_{\ell^*}$ is able to reject the null hypothesis that $\tau_{j_0}=\tau_{j_0}^*$ for $\tau_{j_0}^*=0$, which implies the power is converging to 1. Formally,
        Let $\hat{q}_{\ell^*}(\alpha)$ be the corresponding estimated upper $\alpha$ quantile of the maximum statistic from the multiplier bootstrap procedure at $\ell^*$-th step. Let $E$ denote the event that stepdown process stops at $\ell^*$-th step and by definition of $\ell^*$ we have $\PP(E)=1$. Notice that 
        \begin{align*}
        &\,\PP(\cA_{\ell^*}\cap\cV^* \neq \varnothing )\\
         = &\, \PP(\cA_{\ell^*}\cap\cV^* \neq \varnothing , E) \\
         = &\, \PP\left(\max _{j \in \cA_{\ell^*}} \sqrt{n}\frac{\left|\hat{\tau}_j\right|}{{\hat{\sigma}}_j} > \hat{q}_{\ell^*}(\alpha), \cA_{\ell^*}\cap\cV^* \neq \varnothing , E\right) + \PP\left(\max _{j \in \cA_{\ell^*}} \sqrt{n}\frac{\left|\hat{\tau}_j\right|}{{\hat{\sigma}}_j} \leq \hat{q}_{j_0}(\alpha), \cA_{\ell^*}\cap\cV^* \neq \varnothing , E\right)
        \end{align*}
        By definition the stepdown process does not stop at $\ell^*$-th step if $M_{\ell^*} = \max _{j \in \cA_{\ell^*}} \sqrt{n}\frac{\left|\hat{\tau}_j\right|}{{\hat{\sigma}}_j} > \hat{q}_{\ell^*}(\alpha)$. Hence, the first term is zero. For the second term, suppose $j_0 \in \cA_{\ell^*}\cap\cV^*$ and we have
        \begin{align*}
             &\mathbb{P}\left(\max _{j \in \cA_{\ell^*}} \sqrt{n}\frac{\left|\hat{\tau}_j\right|}{{\hat{\sigma}}_j} \leq \hat{q}_{\ell^*}(\alpha),\cA_{\ell^*}\cap\cV^* \neq \varnothing \right) \\
             \leq & \mathbb{P}\left(\sqrt{n}\frac{\left|\hat{\tau}_{j_0}\right|}{{\hat{\sigma}}_{j_0}} \leq \hat{q}_{\ell^*}(\alpha),\cA_{\ell^*}\cap\cV^* \neq \varnothing \right)\\ 
            \leq& \mathbb{P}\left(\sqrt{n}\frac{\left|\hat{\tau}_{j_0} - \tau_{j_0}^*\right|}{{\hat{\sigma}_{j_0}}} \geq \sqrt{n}\frac{\left|{\tau}_{j_0}^*\right|}{{\hat{\sigma}_{j_0}}} - \hat{q}_{\ell^*}(\alpha)\right) ,
        \end{align*}
        Because $\left|\tau_{j_0}^*\right| / \hat{\sigma}_{j_0} = \left|\tau_{j_0}^*\right| /\sigma_{j_0}\cdot(1+\op(1)) \geq c(\log (p) / n )^{1 / 2}$ for large $m$ and $n$, under the assumed minimal signal strength condition, by Lemma \ref{prop:max-Gaussian} we have
        \[\mathbb{P}\left(\sqrt{n}\frac{\left|\hat{\tau}_{j_0} - \tau_{j_0}^*\right|}{{\hat{\sigma}_{j_0}}} \geq \sqrt{n}\frac{\left|{\tau}_{j_0}^*\right|}{{\hat{\sigma}_{j_0}}} - \hat{q}_{\ell^*}(\alpha)\right) \rightarrow 0,\]        
        as $m,n,p \rightarrow \infty$, when $c>C$. Therefore we have
        \[
        \PP(\cA_{\ell^*}\cap\cV^* \neq \varnothing ) \rightarrow 0,
        \]
        and
        \[
        \PP(\cV^* \subseteq \cV_{\ell^*} ) \rightarrow 1.
        \]
        Because the augmentation step only adds more discoveries, $\cV_{\ell^*}\subset \cV$, it does not decrease the power.
        Therefore, the final set of discoveries $\cV$ has power 1 asymptotically as $m,n,p$ tend to infinity.
    \end{proof}

    \subsection{Helper lemmas}\label{app:subsec:multiple-test-helper}

    \begin{lemma}[Quantile estimation based on Gaussian approximation]\label{lem:max-Gaussian-quantile}
        Suppose the conditions of \Cref{prop:max-Gaussian} hold.
        For $\bg_{\cS} \sim \cN(\zero,\hat{\bD}_\cS^{-1} \hat{\bE}_\cS \hat{\bD}_\cS^{-1})$, define the conditional $\alpha$-quantile $q_{\cS}(\alpha) := \inf\{t\in\RR\mid \PP(\|\bg_{\cS}\|_{\infty} \leq t\mid\{\bZ_i\}_{i=1}^n) \geq \alpha\}$.
        As $m,n,p\rightarrow \infty$, it holds that
        \begin{align*}
            \sup_{H_0^{\cS}:\cS \subseteq \cA^*} \sup_{\alpha\in(0,1)} |\PP(\overline{M}_{\cS} \leq q_{\cS}(\alpha) \mid \{\bZ_i\}_{i=1}^n) - \alpha| \pto 0.
        \end{align*}
    \end{lemma}
    \begin{proof}[Proof of \Cref{lem:max-Gaussian-quantile}]
        From \eqref{eq:Delta-1} in the proof of \Cref{prop:max-Gaussian}, we know that condition (14) in \citet{chernozhukov2013gaussian} holds for $\Delta_{1\cS} = \max_{j\in[p]} |\overline{M}_{\cS}-\overline{M}_{0\cS}|$ uniformly over $\cS\subseteq\cA^*$.
        On the other hand, with probability tending to one,
        \begin{align*}
            \Delta_{2\cS} &\leq \Delta_{2\cA^*}\\
            &= \max_{j\in\cA^*}\PP_n [(\varphi_{ij}/\sigma_j - \hat{\varphi}_{ij}/\hat{\sigma}_j)^2] \\
            &\leq \max_{j\in\cA^*}\PP_n [(\varphi_{ij} - \hat{\varphi}_{ij})^2]/{\sigma}_j^2 + \max_{j\in\cA^*}\PP_n [ \hat{\varphi}_{ij}^2](1/{\sigma}_j- 1/\hat{\sigma}_j)^2 \\
            & \lesssim  \max_{j\in\cA^*}\PP_n [(\varphi_{ij} - \hat{\varphi}_{ij})^2] + \max_{j\in\cA^*}(\sigma_{j} - \hat{\sigma}_{j})^2 \\
            & \lesssim  \max_{j\in\cA^*}\PP_n [|\varphi_{ij} - \hat{\varphi}_{ij}|] + \max_{j\in\cA^*}|\sigma_{j} - \hat{\sigma}_{j}| \\
            & = \Op\left(r_{\varphi} + r_{\sigma} \right),
        \end{align*}
        which follows similarly as in the proof of Step (1) for \Cref{prop:max-Gaussian}.
        Under the conditions that $\max\{\log(pn)^7/n, \log(pn)^{2}n^{-(\alpha\wedge\beta)}\}\leq C n^{-c}$, we have 
        \begin{align*}
            \sup_{\cS\subseteq\cA^*}\PP(\log(pn)^2\Delta_{2\cS} > n^{-c}) \leq \sup_{\cS\subseteq\cA^*}\PP(\log(pn)^2\Delta_{2\cS} > n^{-c}/\log (n)) \rightarrow 0,    
        \end{align*}
        which verifies condition (15) in \citet{chernozhukov2013gaussian}.
        From the proof of Corollary 3.1 in \citet{chernozhukov2013gaussian}, the conclusion follows.
    \end{proof}

\clearpage
\section{Experiment details}\label{app:sec:experiment}

\subsection{Estimation of QTE}\label{app:subsec:est-QTE}

\subsubsection{Initial estimator}
The initial IPW estimator $\hat{\theta}_{aj}^{\init}$ of the quantile $Y_j(a)$ can be obtained by solving the following estimating equation:
\begin{align*}
    \PP_n \left\{\frac{\ind\{A=a\}}{\hat{\pi}_a(\bW)} \psi(Y_j, {\theta})  \right\} = 0
\end{align*}
where $\psi$ is defined in \eqref{eq:quantile-pop}.

\subsubsection{Counterfactual density estimation}
Consider an unconfounded observational study with $(W, A, Y ) \sim \cP$, where $A$ is binary and $Y$ is continuous. Assume consistency, positivity, and exchangeability as usual.
Below, we derive an identifying expression and estimator for the density of $Y(a)$ at a point.

Define $V := \ind\{Y \leq y\}$.
Noting that the density of $Y(a)$ is simply the derivative of $\mathbb{P}\left(Y(a) \leq y\right)$. Letting $p_{Y(a)}(y)$ and $p_Y(y)$ denote the densities of $Y(a)$ and $Y$ respectively, we have
\begin{align*}
    p_{Y(a)}(y) & =\frac{\rd}{\rd y} \mathbb{P}\left(Y(a) \leq y\right) \\
    & =\frac{\rd}{\rd y} \mathbb{E}\left(V^a\right)=\frac{\rd}{\rd y} \mathbb{E}\left\{\mathbb{E}\left(V^a \mid W\right)\right\} \\
    & =\frac{\rd}{\rd y} \mathbb{E}\left\{\mathbb{E}\left(V^a \mid W, A=a\right)\right\} \\
    & =\frac{\rd}{\rd y} \mathbb{E}\{\mathbb{E}(V \mid W, A=a)\} \\
    & =\mathbb{E}\left\{\frac{\rd}{\rd y} \mathbb{P}(Y \leq y \mid W, A=a)\right\} \\
    & =\mathbb{E}\left\{p_Y(y \mid W, A=a)\right\},
\end{align*}
where the exchanging of integrals and derivatives is permitted by Leibniz's integral rule combined with the fact that $V$ is bounded.
As for estimation, we can reduce counterfactual density estimation to statistical density estimation. A natural doubly robust analog for this problem \citep{kim2018causal} that takes inspiration from one-step correction and kernel density estimation is given by
$$
\widehat{p}_{Y(a)}(y):=\frac{1}{h} \PP_n\left\{\frac{\ind\{A=a\}}{\widehat{\pi}\left(a, \bW\right)}\left(K\left(\frac{y-Y}{h}\right)-K\left(\frac{y-\widehat{\mu}\left(a,\bW\right)}{h}\right)\right)+K\left(\frac{y-\widehat{\mu}\left(a,\bW\right)}{h}\right)\right\}
$$
where $K$ is a kernel and $h$ is the kernel smoothing bandwidth. 
Similar pseudo observations could be used with IPW \citep[Remark 4]{kallus2019localized} or plug-in style techniques as well, but we omit them for brevity.

\subsubsection{SQTE}\label{app:subsubsec:SQTE}
    When comparing the quantile effects among genes with different scales, one can also consider the standardized quantile treatment effects (SQTE):
    \begin{align}
        \tau_{j}^{\SQTE_{\varrho}} &= \frac{Q_{\varrho}(Y_j(1)) - Q_{\varrho}(Y_j(0))}{\IQR(Y_{j}(0))} , \label{eq:SQTE}
    \end{align}
    where for a random variable $U$, $Q_{\varrho}[U]$ denote the $\varrho$-quantile of random variable $U$, and $\IQR(U)=Q(0.75)-Q(0.25)$ denote the median and interquartile range of $U$ with quantile function $Q$.
    Typically, when  $\varrho=0.5$, the $\varrho$-quantile equals to the median $Q_{\varrho}(U)=\Median(U)$, and we reveal the standardized median treatment effects $\tau_{j}^{\SQTE} = (\Median[Y_{j}(1)] - \Median[Y_{j}(0)] ) / \IQR(Y_{j}(0))$.

    \subsection{Extra experimental results}\label{app:subsec:extra-experiment}

    \subsubsection{Simulation}\label{app:subsubsec:simu}

    Recall that $\lambda_j$ is the mean of the couterfactuals $X_j(0)$ for $j\in[p]$ and $X_j(1)$ for $j\not \in \cV^*$.
    To generate the set of active genes $\cV^*$, we draw a sample from a Multinomial with 200 trials and $p=8000$ categories with probability
    \[\mathrm{softmax}(\{ \log(\mathrm{sd}(\exp(\lambda_j))) \}_{j\in[p]}).\]
    where the sample standard deviation across cells $i$ is used to estimate the above probability.
    The setup suggests that genes with higher variations are more likely to be active.
    For a maximum signal strength $\theta_{\max}$, we first draw a relative signal strength $r_j\sim\mathrm{Beta}(1,\beta_r)$ and set the final signal strength to be $s_j := \theta_{\max}r_j$.
    
    Then, we consider two simulation scenarios for $X_j(1)$ with $j\in\cV^*$ in the treatment group.

    \begin{enumerate}[label=(\arabic*)]
        \item Mean shift with high SNR

        In this case, we set $(\theta_{\max},\beta_r) = (1,0.5)$ so that the more signals have magnitudes close to $\theta_{\max}$.
        Then, we adjust the effect sizes ($\lambda_j$) by adding or subtracting a signal ($\theta$) with equal probability. 
        Specifically, this can be represented as:
        \[
        X_j(1)\sim\mathrm{Poisson}(\lambda_{j} + s_j \delta_{j}),\qquad j\in\cV^*,
        \]
        where $\delta_{j} \sim\mathrm{Bernoulli}(0.5)$.

        \item Median shift with low SNR 

        In this case, we set $(\theta_{\max},\beta_r) = (10,2)$ so that the more signals have magnitudes close to $0$.
        Then we draw
        \[
        X_{j}(1) \sim [\mathrm{LogNormal}(\lambda_j - s_j^2/2, s_j)] ,\qquad j\in\cV^*,
        \]
        which ensures that $X_j(1)$ has the same mean as $X_j(0)$, while their medians are different.   
        Above, $[x]$ indicates rounding to ensure the generated expression levels are integers.
    \end{enumerate}
    These DGPs aim to simulate complex data structures that reflect real-world phenomena, such as varying levels of signal perturbation and the impact of such variations on statistical analyses, particularly in the context of mean and median shifts in treated distributions.

    We also inspect the effect of cross-fitting, which helps fulfill the sample splitting requirement and improves the estimation and inference accuracy.
    More specifically, we randomly split $n$ observations into $K$ disjoint folds $\cN_1,\ldots,\cN_K$.
    Then, for the $k$th fold $\cN_k$, we compute the influence function values $\hat{\varphi}_{ij}$ for $i\in\cN_k$ with nuisance functions estimated from observations in other folds $\cN_1,\ldots,\cN_{k-1},\cN_{k+1},\ldots\cN_K$.
    We set the number of folds $K=5$.
    
    As shown in \Cref{fig:simu-split}, the BH procedure for the ATE test controls the \FDR at the desired level.
    When a large sample size, the \FDR for the STE test with the BH procedure also gets closer to the desired threshold.
    This may be because the test statistics get more positively dependent on cross-fitting.
    However, the BH procedure does not control the \FDR for the QTE test.
    In terms of \FDPex, the proposed multiple testing procedure gets tighter control, while the BH procedure still fails to control it.
    Finally, the power deteriorates with cross-fitting compared to results in \Cref{fig:simu}.

    \begin{figure}[!ht]
        \centering
        \includegraphics[width=0.7\textwidth]{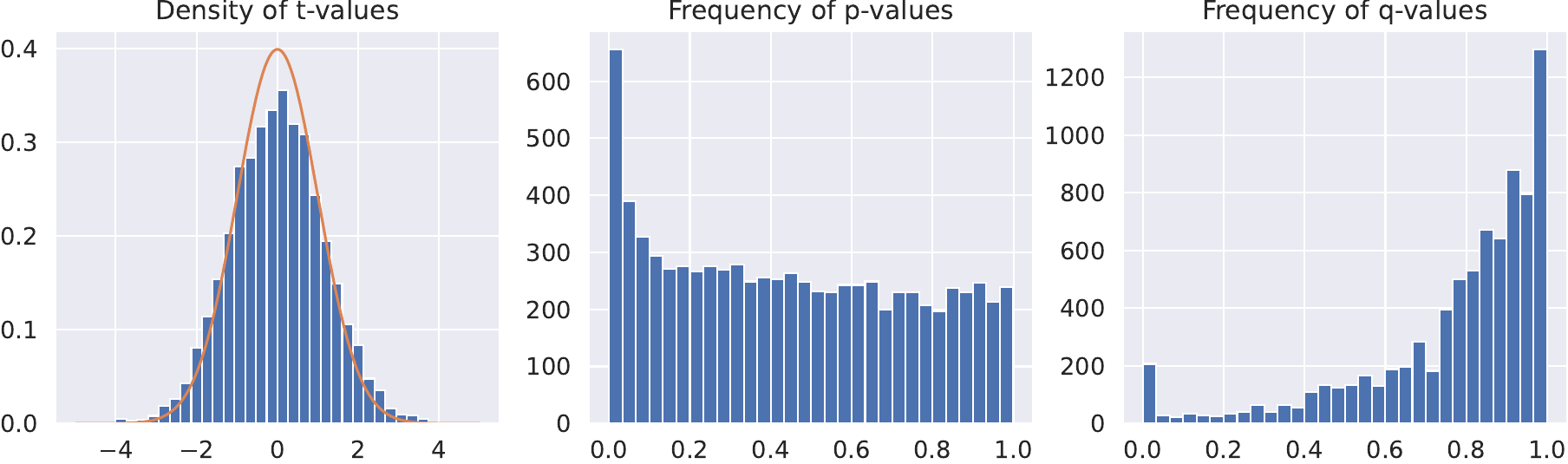}
        \caption{
        The histogram of different statistics in one simulation of \Cref{fig:simu} under mean shifts with $n=100$.
        In this experiment, the number of true non-nulls is 200, while BH produces 258 discoveries with a q-value cutoff of 0.1, yielding 30\% false discoveries.
        }
        \label{fig:BH-stat}
    \end{figure}

    \begin{figure}[!ht]
        \centering
        \includegraphics[width=0.65\textwidth]{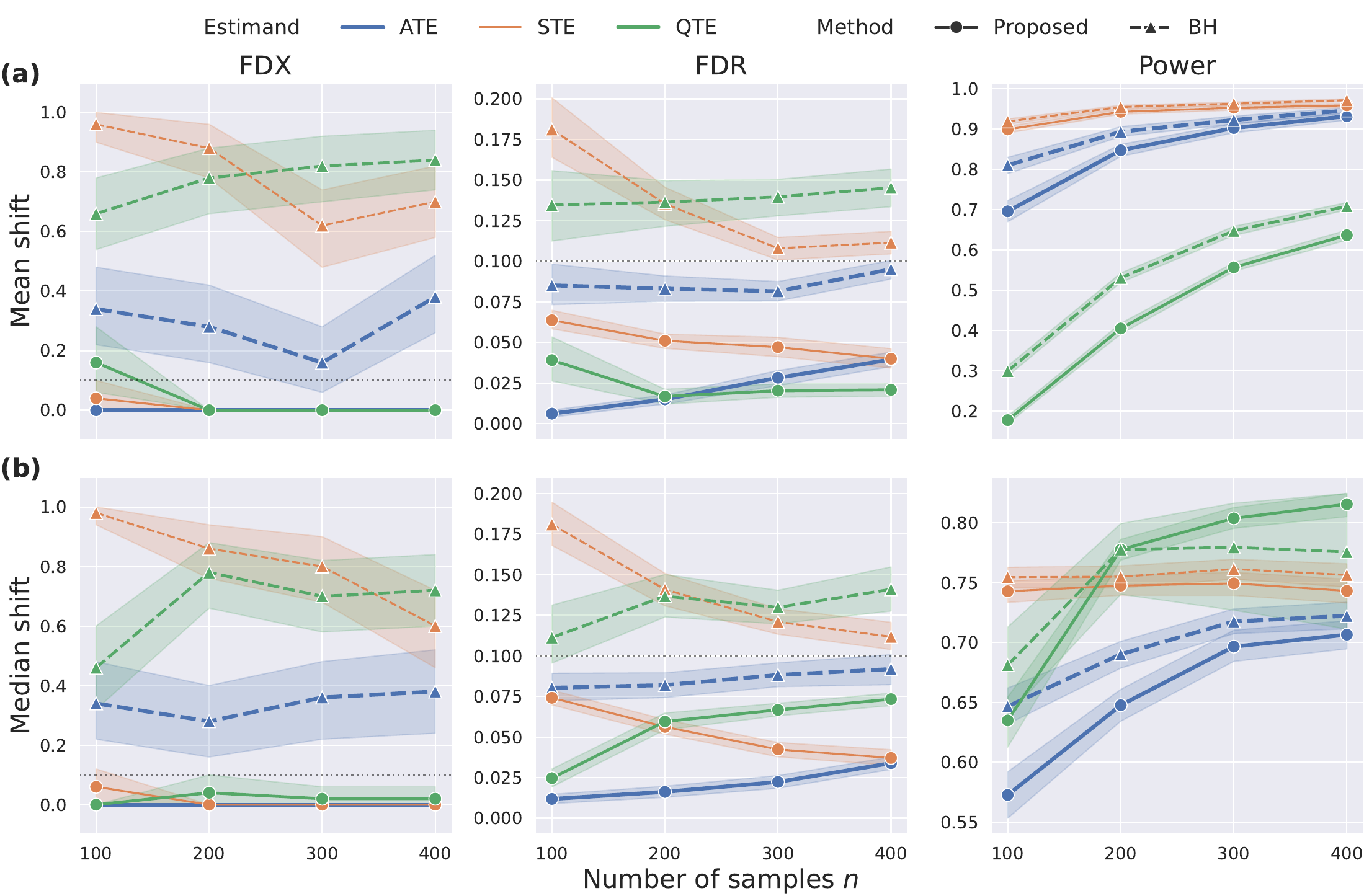}
        \caption{
        Simulation results of the hypothesis testing of $p=8000$ outcomes based on different causal estimands and FDP control methods for detecting differential signals under \textbf{(a)} mean shifts and \textbf{(b)} median shifts averaged over 50 randomly simulated datasets with 5-fold cross-fitting.        
        The gray dotted lines denote the nominal level of 0.1.}
        \label{fig:simu-split}
    \end{figure}

    \clearpage
    \subsubsection{LUHMES data}\label{app:subsec:LUHMES-data}

    \begin{table}[!ht]
        \centering
        \begin{tabularx}{0.8\textwidth}{cC{3.3cm}C{3.3cm}C{3.3cm}}
        \toprule
        \textbf{gRNA} & \textbf{Number of cells $n$} & \textbf{Number of genes $p$} & \textbf{Number of perturbed cells} \\
        \midrule
        \emph{PTEN} & 1014 & 1444 & 458 \\
        \emph{CHD2} & 756 & 1360 & 200\\
        \emph{ASH1L} & 725 & 1360 & 169 \\
        \emph{ADNP} & 895 & 1416 & 339 \\
        \bottomrule
        \end{tabularx}
        \caption{The summary of sizes of data under different perturbations.}
        \label{tab:LUHMES-data}
    \end{table}

    \begin{table}[!ht]
        \centering
        \begin{tabularx}{\textwidth}{cC{7cm}C{3.3cm}C{3.3cm}}
        \toprule
        \textbf{gRNA} & \textbf{Common} & \textbf{ATE only} & \textbf{STE only} \\
        \midrule
        \emph{PTEN} & PTH2, PTGDS, NEFM, EEF1A1, C21orf59, MFAP4, ALCAM, NEFL, ITM2C, EIF3E, CRABP2, SLC25A6, EIF3L, WLS, PPP1R1C, GNB2L1, SVIP, RGS10, H3F3A, DRAXIN, GNG3, TCP10L, EIF3K & PCP4, MYL1, RAI14, DNER, MAP7, SNCA, TSC22D1, NRP2, SKIDA1 & CCER2, PRDX1, TCF12 \\\cmidrule(lr){2-4}
        \emph{CHD2} & EEF1A1, NEFL, GNG3, ID4, EEF2, STMN2 & PRDX1, TUBB4B & PCP4 \\\cmidrule(lr){2-4}
        \emph{ASH1L} & MT-CO1, MT-CYB, FXYD7 &  & PKP4 \\\cmidrule(lr){2-4}
        \emph{ADNP} & C21orf59, MAP1B & PTGDS, KLHL35, LHX2 &  \\
        \bottomrule
        \end{tabularx}
        \caption{Sigfinicant genes for different guide RNA mutation on the late-stage cells.
        The last three columns show the discoveries that are significant in (1) both the ATE and the STE tests, (2) only the ATE tests, and (3) only the STE tests.}
        \label{tab:LUHMES-DE-genes}
    \end{table}
    
    \clearpage
    \subsubsection{Lupus data}
    \begin{figure}[!ht]
        \centering
        \includegraphics[width=0.55\textwidth]{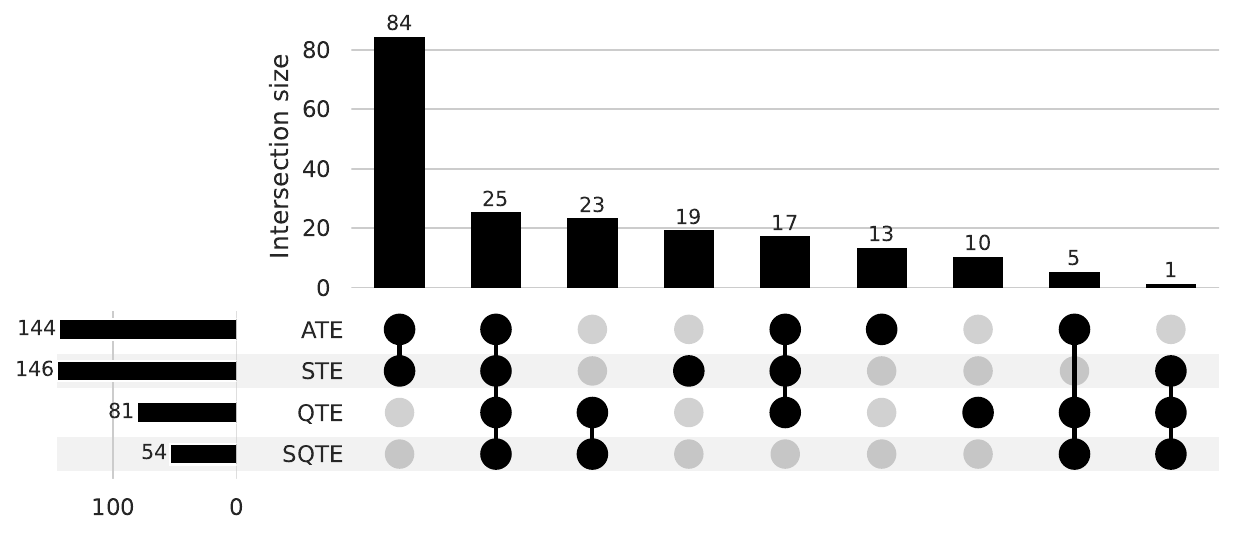}
        \includegraphics[width=0.55\textwidth]{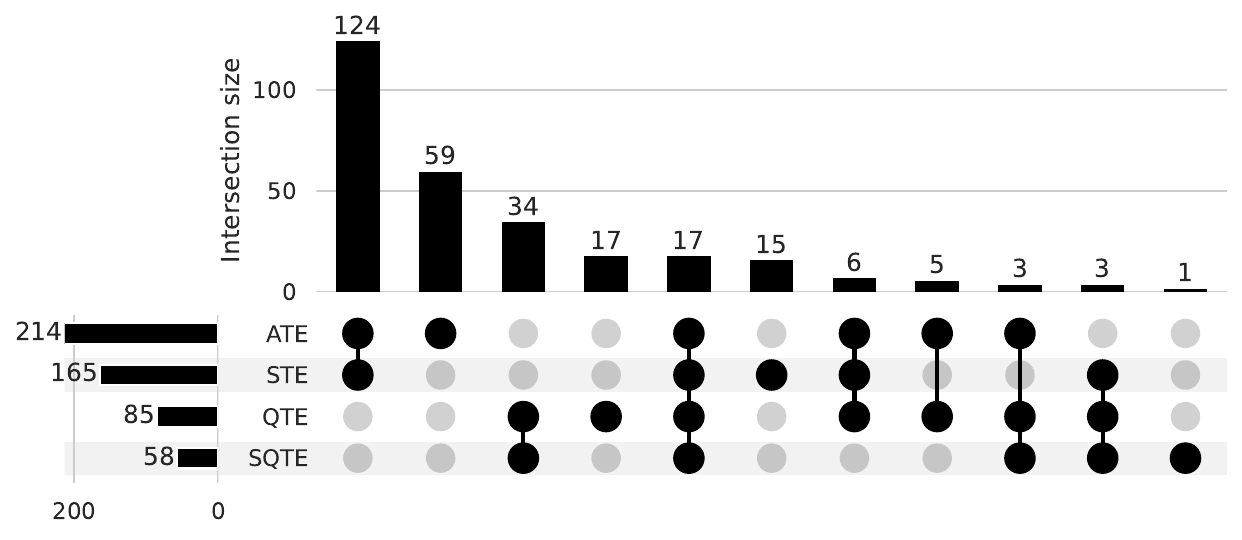}
        \includegraphics[width=0.55\textwidth]{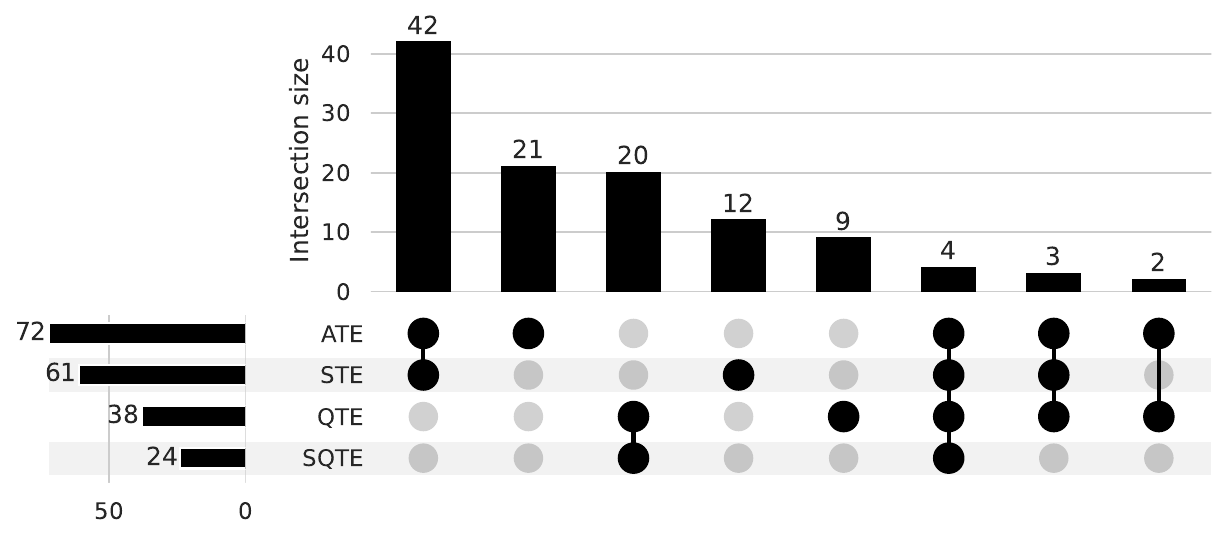}
        \includegraphics[width=0.55\textwidth]{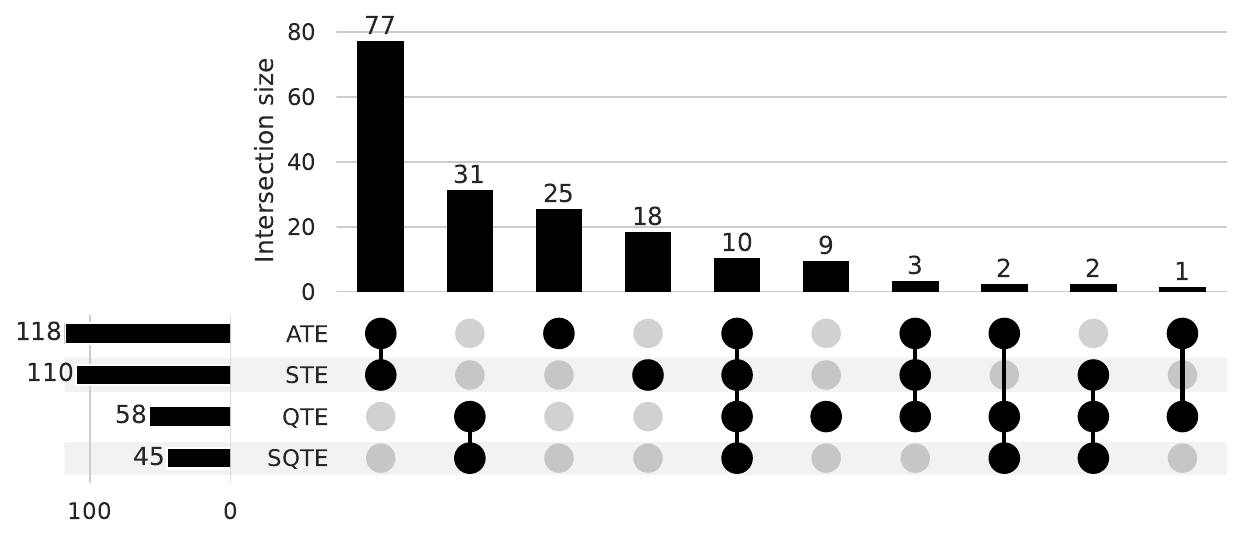}
        \includegraphics[width=0.55\textwidth]{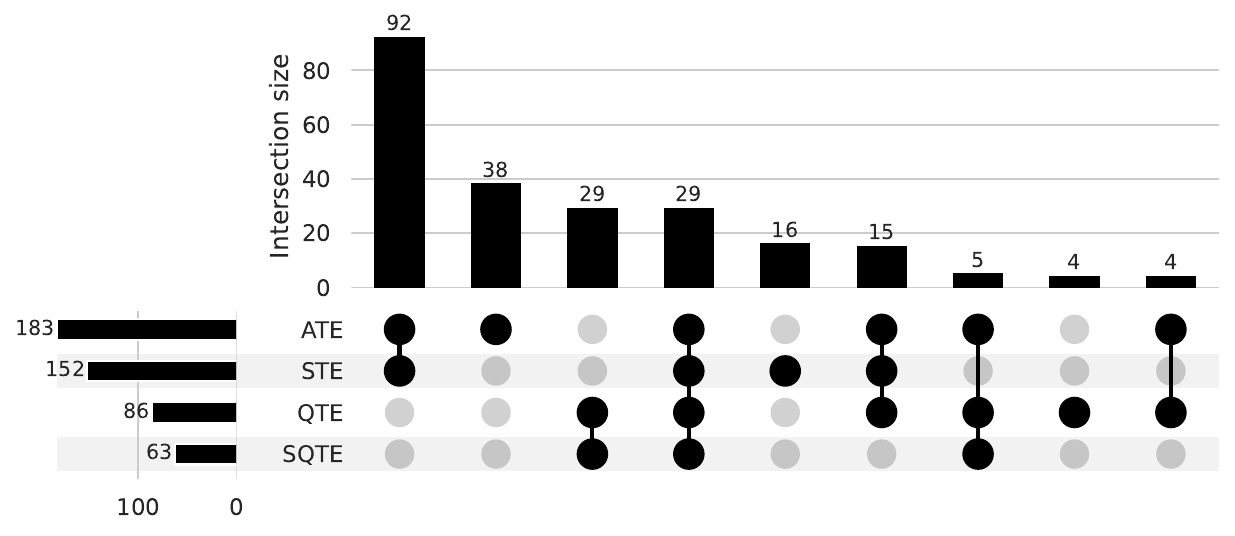}
        \caption{Upset plot of discoveries by tests based on different causal estimands on the T4, T8, NK, B, and cM cell types of Lupus data set.}
        \label{fig:lupus-upset-extra}
    \end{figure}

\end{document}